\documentclass[11pt,a4paper]{article}
\usepackage[english]{babel}
\usepackage{amsmath}
\usepackage{amsfonts}
\usepackage{amsthm}
\usepackage{mathtools}

\newtheorem{thm}{Theorem}
\newtheorem{cor}{Corollary}
\newtheorem{lem}{Lemma}

\theoremstyle{definition}

\newtheorem{regime}{Asymptotic regime}
\newtheorem{rem}{Remark}
\theoremstyle{definition}

\newtheorem{model}{Model}
\usepackage[ruled,vlined]{algorithm2e}
\usepackage{graphicx}
\counterwithout{figure}{section}
\DeclarePairedDelimiter\floor{\lfloor}{\rfloor}

\usepackage{xcolor}
\usepackage{multirow}
\usepackage{comment}

\usepackage{geometry}
\usepackage{comment}
\usepackage{tikz,graphics,color,fullpage,float,epsf,caption,subcaption}

\usepackage{csquotes}
\usepackage{setspace}

\usepackage[authoryear,round]{natbib}
\usepackage[colorlinks=true, allcolors=blue]{hyperref}
\title{Linear-cost unbiased posterior estimates for crossed effects and matrix factorization models via couplings}

\newcommand\add[1]{\textcolor{black}{#1}}
\DeclareMathOperator{\erf}{erf}

\newcommand{\bb}{\boldsymbol}
\newcommand\X{\boldsymbol{X}}
\newcommand\Y{\boldsymbol{Y}}
\newcommand\x{\boldsymbol{x}}
\newcommand\y{\boldsymbol{y}}
\newcommand\br{\boldsymbol{r}}
\newcommand\s{\boldsymbol{s}}
\newcommand\bmu{\boldsymbol{\mu}}
\newcommand\sX{\mathcal{X}}
\newcommand\sP{\mathcal{P}}

\newcommand{\bigo}[1]{O(#1)}
\newcommand{\prob}{\pi}

\usepackage{authblk}
\author[a]{Paolo Maria Ceriani}
\author[a]{Andrea Pandolfi}
\author[a,b]{Giacomo Zanella}
\affil[a]{Department of Decision Sciences, Bocconi University, Milan, Italy}
\affil[b]{Bocconi Institute for Data Science and Analytics, Bocconi University, Milan, Italy}
\date{}                   
\begin{document}
\maketitle

\begin{abstract}
We design and analyze unbiased Markov chain Monte Carlo (MCMC) schemes based on couplings of blocked Gibbs samplers (BGSs), whose total computational costs scale linearly with the number of parameters and data points. 
Our methodology is designed for and applicable to high-dimensional BGS with conditionally independent blocks, which are often encountered in Bayesian modeling.
We provide bounds on the expected number of iterations needed for coalescence for Gaussian targets, as well as on the tails of the coalescence times distribution. These imply that practical two-step coupling strategies achieve coalescence times that match the relaxation times of the original BGS scheme up to logarithmic factors.  
To illustrate the practical relevance of our methodology, we apply it to high-dimensional crossed random effect and probabilistic matrix factorization models, for which we develop a novel BGS scheme with improved convergence speed.
Our methodology provides unbiased posterior estimates at linear cost (usually requiring only a few BGS iterations for problems with thousands of parameters), matching state-of-the-art procedures for both frequentist and Bayesian estimation of those models.
\end{abstract}

\noindent%

\section{Introduction}
In recent years, unbiased Markov Chain Monte Carlo via couplings (UMCMC) has emerged as a promising framework to remove bias from MCMC estimates, thus potentially allowing for early stopping, simplifying the 
convergence diagnostic process and facilitating parallelization \citep{glynn_rhee, jacob2019unbiased}. 
In UMCMC, coupled chains are run for a random number of iterations (at least up to coalescence) and their values are combined to produce unbiased estimates. A natural question that arises is whether this class of estimates incurs a greater computational cost than conventional MCMC based on simple ergodic averages and to quantify this potential difference. Framing the question differently, one may ask whether it is possible to devise UMCMC methods with computational cost matching top performing MCMCs, while enjoying the above mentioned benefits.

On a different line of research, various works have shown how carefully designed blocked Gibbs Samplers (BGSs), i.e.\ Gibbs sampling schemes that update entire blocks of coordinates jointly, can achieve state-of-the-art performance for sampling from the posterior distributions of various challenging high-dimensional Bayesian models, such as non-nested models with crossed dependencies \citep{papaspiliopoulos2018scalable, tim_scalable}. In particular, BGSs achieve linear computational costs in the number of parameters and observations in asymptotic regimes where both diverge to infinity. 

In this work, we seek to combine these two lines of research, aiming to design UMCMC BGS methods with linear computational cost in the aforementioned high-dimensional regimes.
Specifically, we provide a theoretical contribution, i.e.\ the analysis of BGS couplings for Gaussian targets via explicit bounds on the expected number of iterations, showing that practical two-step BGS coupling schemes achieve coupling times that match relaxation times up to a logarithmic factor; and some methodological ones, discussing implementation aspects of couplings of BGS with conditionally independent blocks and developing a novel BGS scheme for probabilistic matrix factorization which empirically reduces the MCMC complexity to linear for those models. 
To illustrate the practical relevance of our methodology, we apply it to crossed random effect models \citep{gelman2005analysis,baayen2008mixed}, a commonly used class of additive models connecting a response variable to categorical predictors, and to probabilistic matrix factorization (PMF) models \citep{matfactmnih}, dimensionality reduction models based on low-rank representations.  

The remaining part of the article is organized as follows: after briefly presenting the objectives of the paper and the three running examples that motivate our research in Section \ref{sec:mot}, we review how to exploit couplings to obtain unbiased estimates in Section \ref{sec:backg}. In particular, Section \ref{sec:methodology} discusses the use of couplings for BGSs.
The main theoretical results are presented in Section \ref{sec:bound}: we provide bounds on the expected number of iterations needed for coalescence of coupled chains, as well as on the tails of such coalescence times, with important implications for parallel computations.
We apply the methodology and the theoretical results to Gaussian crossed random effect models in Section \ref{sec:crem}, generalized linear mixed models (GLMMs) with crossed effects in Section \ref{sec:ngcrem} and to PMF models in Section \ref{sec:pmf}. 
The code to reproduce the simulations reported in this paper can be found at \href{https://github.com/paoloceriani/couplings_bgs}{github.com/paoloceriani/coupling\_bgs}.
Proofs, additional results and simulations are deferred to the supplement.

\section{Motivation and objectives}
\label{sec:mot}
\add{BGSs are particularly well-suited for high-dimensional Bayesian models with a high degree of conditional independence, for which they achieve state-of-the-art performance, and, unlike most available sampling schemes, can result in a total computational cost that scales linearly in the number of observations and parameters. }
In this paper, we consider the following three models as running examples, motivating the methodology and theory developed later.
Despite the relatively simple formulations, these models are computationally challenging to estimate: correlated errors lead to expensive GLS estimates and strong posterior dependence 
results in slow mixing of standard MCMCs.

\begin{model}[Gaussian crossed random effects]\label{ex:gcrem}
Cross-classified data, where each observation can be simultaneously classified according to two or more variables, are commonly found in the scientific literature, with applications in various domains including health and social sciences \citep{gelman2005analysis,baayen2008mixed}. 
A univariate response variable $y$ is assumed to depend additively on the 
effects of $K$ categorical variables, termed \textit{factors}, each one with $I_k$ different possible values, called \textit{levels}, for $k=1,...,K$. The effect of the $i$-th level of the $k$-th factor is described by a random variable ${a}_{k,i}$.
Let $\mathbf{a}_{k} = (a_{k,1},..., a_{k,I_k})$ denote the $I_k$-dimensional vector of effects of the $k$-th factor,  for $k=1,...,K$; $\mathbf{y}=(y_n)_{n=1}^N$, $\mathbf{a}=(\mathbf{a}_k)_{k=1}^K$ and  $\boldsymbol{\tau}= (\tau_k)_{k=0}^K$ the vectors of all data, effects and precisions respectively.  Let $i_k[n]$ denote the level of the $k$-th factor associated to the $n$-th observation \citep[see e.g. Section 1.1 and Chapter 11 of ][ for details on this notation]{Gelman_Hill_2006}.
In this paper, for clarity of exposition, we will consider the intercept-only version, although the concepts discussed extend to more general versions with covariates and random slopes \citep{gao_16,tim_scalable}.
The model with its standard prior can then be written as
\begin{gather}
\begin{aligned}
		\label{eq:gauss_crem}
		y_n | \mu, \mathbf{a}, \tau_0 &\sim N\left(\mu +\sum_{k=1}^K a_{k,i_{k}[n]},\tau_0^{-1}\right) & n=1,...,N,\\
		a_{k,i} | \tau_k &\sim N(0, \tau_k^{-1}) & i=1,..., I_k, \, k=1,..., K , \\
		p(\tau_k) &\propto \tau_k^{-0.5} &\text{ for } k=0,..., K ,
\end{aligned} 
\end{gather}
and $p(\mu)\propto 1$, where $\mu$ is a random intercept. In \eqref{eq:gauss_crem}, $p(\cdot)$ denotes the density of the random variable inside the brackets and $N(\mu,\sigma^2)$ denotes a Gaussian distribution with mean $\mu$ and variance $\sigma^2$. We will often consider the model with $K=2$ factors, where one can think of $y_n$ as the rating that user $i_1{[n]}$ gave to film $i_2{[n]}$.
\end{model}

\begin{model}[GLMMs with crossed effects]	\label{ex:ngcrem}
Generalized linear mixed models (GLMMs) extend the framework of linear mixed models to accommodate non-Gaussian response variables by incorporating a link function, but still retaining the same dependence structure.
They are a powerful tool widely used in many academic fields, such as political science, biology, and medicine \citep{wood2017, jiang}.
\add{Standard choices for the link function include the binomial link to model binary or discrete responses \citep{ghitza2013,book:bates}, Poisson or negative binomial to model count data \citep[Ch.\ 10]{book:DunnSmyth2018}, gamma or log-normal link to model continuous positive data \citep[Ch.\ 11]{book:DunnSmyth2018}.}
Extending Model \ref{ex:gcrem} to allow for general response gives
	\begin{align}\label{eq:general_crem}
		\mathcal{L}(y_n | \mu, \mathbf{a} ) = \mathcal{L}(y_n | \eta_n )\text{ with }\eta_n = \mu +\sum_{k=1}^K a_{k, i_{k}[n]} & \text{ for } n=1,...,N, 
	\end{align} 
	where $\mathcal{L}(\cdot)$ denotes the law of the random variable within brackets. Different choices of the conditional distribution $\mathcal{L}(y_n | \eta_n )$ lead to different models, e.g. $\mathcal{L}(y_n | \eta_n )  = N(\eta_n, \tau_0^{-1})$ is equivalent to Model \ref{ex:gcrem} or $\mathcal{L}(y_n | \eta_n )= Bern(p)$ with $p =\frac{1}{1+e^{-\eta_n}}$ leads to a logit model for binary data. In Section \ref{sec:ngcrem} we will consider the case $ \mathcal{L}(y_n | \mu, \mathbf{a} )= Lapl( \eta_n ,1/\sqrt{2})$, where $Lapl(\mu, b)$ denotes the univariate Laplace distribution with mean $\mu$ and scale $b$.
\end{model}

\begin{model}[Probabilistic matrix factorization]
	\label{ex:pmf}
	Low rank matrix factorization methods provide one of the simplest and most effective approaches to collaborative filtering. \add{These methods have been widely used in recommender systems after the success of the Netflix Prize competition, where they proved to achieve competitive predictive performance \citep{koren2009netflix, matfactmnih}.} 
	As for Models \ref{ex:gcrem} and \ref{ex:ngcrem} with $K=2$, one can think of observation $y_n$ as representing the rating that user $i[n]$ gives to film $j[n]$. Denoting by $\textbf{u}_i$ and $\textbf{v}_j$ respectively the $d$-dimensional latent user-specific and film-specific factors for $I=1,..., I_1$ and $j =1,...,I_2$, and by $\textbf{u}=(\textbf{u}_{i})_{i=1}^{I_1}\in\mathbb{R}^{I_1\times d}$, $\textbf{v}=(\textbf{v}_{j})_{j=1}^{I_2}\in\mathbb{R}^{I_2\times d}$ their collections, the model can be formulated as
	\begin{gather}
	\begin{aligned}
		\label{eq:pmf}
		y_n | \rho, \textbf{u}, \textbf{v}, \tau_0 & \sim N(\rho \textbf{u}_{i[n]}^\top \textbf{v}_{j[n]}, \tau_0^{-1}) &n=1,...,N,\\
		\textbf{u}_i,\textbf{v}_j  &\sim N(\textbf{0},1_d) &i=1,...,I_1, \; \; j=1,...,I_2,  \\
		\tau_0&\sim \hbox{Gamma}(c,d),\, &\rho^{-\frac{1}{2}} \sim \hbox{Gamma}(a,b), 
	\end{aligned}
	\end{gather}
	where $\hbox{Gamma}(a,b)$ denotes a Gamma random variable with shape parameter $a$ and scale parameter $b$, 
$1_d$ denotes the $d$-dimensional identity matrix and $\rho$ indicates a positive quantity acting as a scaling factor for the random effects. PMF models can be seen as a multiplicative extension of Models \ref{ex:gcrem} and \ref{ex:ngcrem}, and are usually more challenging to estimate (due to invariance with respect to orthogonal transformations, a lower degree of linearity, etc.).
\end{model}

\medskip
\add{Despite the computational challenges posed by these models, a Bayesian approach might still be preferred. For example, in the political science literature, multilevel regression and poststratification \citep{ghitza2013, goplerud2022} are widely used to derive precise vote choice for small subgroups of the population. A Bayesian approach regularizes estimates for these units, borrowing strength across related groups.
Similar advantages can be found in large rating systems, where only few ratings are available for each user or item.
For PMF models, a Bayesian approach can lead to  an increase in predictive accuracy \citep{mnih2007pmf} and allows dealing with uncertainty more effectively \citep{matfactmnih}.}

\subsection{Asymptotic regimes of interest and computational cost}
\label{sssec:theory}
Models \ref{ex:gcrem}, \ref{ex:ngcrem} and \ref{ex:pmf} naturally lead to situations where both the number of observations $N$ and parameters $p=\bigo{\sum_{k=1}^K I_k}$ are large. We use the notation $(T_n)_{n \in \mathbb{N}}=\bigo{f(n)}$ if there exist constants $c,C\in\mathbb{R}$ with $0<c<C<\infty$ such that $c f(n)\le T_n \le C f(n)$ for all $n$. 
In the following, we will talk about asymptotic regimes in terms of $N\to\infty$, implicitly assuming that $p$ is a function of $N$ that is also diverging as $N\to\infty$. 

Also, the above models are commonly used in \emph{sparse} settings, where a small fraction of the possible combinations of effects are observed, i.e.\ $N\ll\prod_{k=1}^KI_k$.
For example, when $K=2$ one often has $1\ll p< N\ll I_1 \times I_2$ (see \cite{gao_16} for further discussion).
Using the analogy of films and ratings for recommender systems, the above corresponds to assuming that the number of ratings, users and films is large, but each user rates only a small fraction of the films.  
Depending on the degree of sparsity in the observation design, one could have either $p=\bigo{N}$ or $p/N\to 0$ as $N\to\infty$.

We consider the task of performing posterior inference for the above models using MCMC methods.
We are interested in quantifying the computational effort required for the posterior estimation (both in the MCMC and UMCMC frameworks) as $N\to\infty$.
In the (U)MCMC context, the total cost is defined as the product of the cost per iteration and the expected number of iterations for the convergence (coalescence) of the chains.
As discussed below, recent works suggest that BGS can achieve state-of-the-art performances of $\bigo{N}$ posterior estimation cost. 
Our main objective is to assess whether UMCMC methods with the same cost can be devised for this problem, as well as to provide some guidance on how to do so.

\subsection{Related literature and block-updating schemes}
\label{ssec:comp_cost}
Models with crossed dependencies are computationally harder than classical Bayesian hierarchical models with nested structures. 
For example, even in the Gaussian case (i.e.\ Model \ref{ex:gcrem}), evaluating the marginal likelihood once (e.g.\ computing $\mathcal{L}(\boldsymbol{y}\mid\boldsymbol{\tau},\mu)$ for a given $\boldsymbol{\tau}$ and $\mu$) requires the inversion of a $\bigo{\sum_{k=1}^KI_k}$-dimensional matrix. Despite the matrix being sparse, the crossed dependence structure leads to a dense Cholesky factor \citep[Sec.3]{andrea2024}, and more generally prevents the use of efficient sparse linear algebra tools available for, e.g., nested or spatial hierarchical model, leading to a computational cost of at least $\bigo{N^{3/2}}$ for each evaluation \citep{gao_16,perry2017fast,tim_scalable,menictas2023streamlined}. 
The situation is obviously worse in the non-Gaussian case, where analogous computations involve general $\bigo{\sum_{k=1}^KI_k}$-dimensional integrals.

On the other hand, the above models lend themselves naturally to block updating schemes, such as BGSs in the sampling context or block coordinate ascent (aka backfitting) for maximum a posteriori (MAP) or generalized least square (GLS) computations. For example, given the conditional independence structure of Model \ref{ex:gcrem}, 
the posterior conditional distribution of $\textbf{a}_k$ factorizes as
$\mathcal{L}(\textbf{a}_k|\mu,\textbf{a}_{-k}, \boldsymbol{\tau}, \textbf{y})
=\otimes_{i=1}^{I_k}\mathcal{L}(a_{k,i}|\mu,\textbf{a}_{-k}, \boldsymbol{\tau}, \textbf{y})$, where $\otimes$ denotes the product of independent distributions.
Thus a BGS with components $\mu, \boldsymbol{a}_1, ...,\boldsymbol{a}_K$ and $\boldsymbol{\tau}$, which we will call \emph{vanilla} BGS, can be trivially implemented at $\bigo{N}$ cost per iteration for Model \ref{ex:gcrem}. 
However, this vanilla version can mix slowly.
In particular, \cite{gao_16} showed that for Model \ref{ex:gcrem} with $K=2$ factors, known variances and full observation designs, the vanilla BGS requires $\bigo{\sqrt{N}}$ to converge, leading to a prohibitive $\bigo{N^\frac{3}{2}}$ total cost.
This follows from the fact that observed values create strong a posteriori dependence between unknown factors. 
\cite{papaspiliopoulos2018scalable} proposed a collapsed Gibbs Sampler (see Algorithm \ref{alg:cg} below) which preserves the $\bigo{N}$ cost per iteration and converges in $\bigo{1}$ iterations under appropriate assumptions \citep[see also][Thm.2]{tim_scalable}.
Similar techniques have been employed to develop a \textit{back-fitting} algorithm to perform GLS estimation for an analogue of Model \ref{ex:gcrem} with $\bigo{N}$ cost in \cite{gosh_back}.  
A first question of interest that we consider is whether the same computational efficiency can be extended to the UMCMC context, which we answer positively in Section \ref{sec:crem}.
\add{In the UMCMC case, one can stop MCMC runs after a few iterations (e.g.\ around $10$, see Section \ref{ssec:numerics_gcrem}), which is similar to the number of iterations typically required for the convergence of backfitting \citep{gosh_back}. Note that UMCMC methods require running many independent chains, which increases the overall computational cost. However, when many parallel machines are available, the wall-clock time can be made comparable with that of backfitting. A more detailed discussion is provided in Section \ref{sec:tail_bound}.}

\section{Background on couplings for estimation and BGSs}
\label{sec:backg}
We now provide some background material on UMCMC and BGSs. 
Specifically, Section \ref{ssec:coup_est} provides a concise recap of how to exploit couplings for unbiased MCMC estimation, as presented in \cite{jacob2019unbiased}, while Sections \ref{ssec:bgs_notation} and \ref{ssec:2step} introduce, respectively, BGS kernels and two-step coupling algorithms.

\subsection{Notation}
In the following, vectors are denoted in bold, matrices in capital letters and univariate quantities in standard lowercase. We denote the space of probability measures over a space $\sX$ by $\mathcal{P}(\sX)$.
Given $p,q\in \mathcal{P}(\sX)$, $\Gamma(p,q)$ is the set of couplings between $p$ and $q$, i.e.\ joint distributions on $\mathcal{X} \times \mathcal{X}$ whose first and second marginals are, respectively,  $p$ and $q$. For a kernel $P$ on $\mathcal{X}$, we denote by $\bar{P}$, or more explicitly $\bar{P}[P]$, a kernel on $\mathcal{X} \times \mathcal{X}$ such that 
$\bar{P}[P]((\x,\y), \cdot)\in\Gamma \left( P(\x, \cdot), P(\y, \cdot)  \right)$ for every $(\x,\y)\in \mathcal{X} \times \mathcal{X}$. 
We denote by $(p \otimes q)\in\sP(\sX \times \mathcal{Y})$ the product measure defined as $(p \otimes q)(A\times B ) = p(A)q(B)$ for all $A\subseteq \sX$ and $B\subseteq  \mathcal{Y}$. 
With a slight abuse of notation we use $\Gamma(p,q)$ to denote both the collection of distributions and that of random variables, i.e.\ we also write $(\X,\Y)\in\Gamma(p,q)$ for random vectors $(\X,\Y)$  such that $\X \sim p, \textbf{Y} \sim q$.
A coupling $(\X,\Y)\in\Gamma(p,q)$ is called \textit{maximal} if it maximizes the probability of equality between the realizations of the two variables, i.e.\ if $\Pr(\X = \Y) = 1- \| p-q \|_{TV}$ where $ \| \cdot \|_{TV}$ denotes the norm induced by the total variation distance. We will denote by $\Gamma_{max}(p,q)\subset \Gamma(p,q) $ the collection of maximal couplings of $p,q$. Analogously, we write $\bar{P}[P] \in \Gamma_{max}[P]$ if $\bar{P}((\x,\y),\cdot) \in \Gamma_{max}(P(\x, \cdot), P(\y, \cdot))$ for every $\x,\y \in \mathcal{X}$. 
A coupling $(\mathbf{X},\mathbf{Y})\in\Gamma \left( p, q \right)$ minimizing $\mathbb{E}[\|\mathbf{X-Y}\|^2 ]$ among all couplings of $p$ and $q$ is called Wasserstein-2 ($W_2$) optimal, and we will denote the family of such optimal couplings as $\Gamma_{W_2}(p,q)$. Analogously, we say that $\bar{P}$ is a $W_2$-optimal coupling of a kernel $P$, and write $\bar{P}[P] \in \Gamma_{W_2}[P]$, if $\bar{P}[P]((\x,\y),\cdot) \in \Gamma_{W_2}(P(\x, \cdot), P(\y, \cdot))$ for every $\x,\y \in \mathcal{X}$. A recap on maximal and $W_2$ optimal couplings can be found in the supplementary material.

\subsection{Background on UMCMC}\label{ssec:coup_est}

We are interested in approximating expectations of the form 
$\mathbb{E}_{\pi}[h] = \int_{\sX} h(\x) \pi(d\x),$ 
where $\pi \in \mathcal{P(X)}$ 
and $h:\mathcal{X}\to\mathbb{R}$ a test function. 
Following \cite{glynn_rhee} and \cite{jacob2019unbiased}, we consider unbiased estimators of $\mathbb{E}_{\pi}[h] $ based on coupled Markov chains 
that marginally evolve according to a common $\pi$-invariant transition kernel $P$. 

Let  $(\X^{t},\Y^{t})_{t \ge 0}$ be a Markov chain  on $\mathcal{X} \times \mathcal{X}$ with coupled kernel $\bar{P}[P]$ such that: the two chains evolving according to $\bar{P}$ must  meet after finite time, i.e.\ if we define the meeting time $T= \min\{ t \ge 0 \;: \; \X^{t} = \Y^{t}\}$, 
it must hold $\Pr(T<\infty)=1$; and after meeting the two chains stay together, i.e.\ $\X^{t}=\Y^{t}$ for all $t\geq T$.
The initial distribution is taken to be $(\X^{0}, \Y^{0})\sim(\pi_0 P)\otimes \pi_0$ for some $\pi_0$, meaning that we initialize $\X^{-1} \sim \pi_0$ and $\Y^{0} \sim \pi_0$, with  $\X^{-1} $ and $\Y^{0}$ independent, and then take $\X^{0}|\X^{-1} \sim P(\X^{-1}, \cdot)$.

Under the above assumptions and some regularity conditions on the distribution of $T$ (see Sec.2.1 of \cite{jacob2019unbiased} or milder conditions in \cite{coupl_ext}), the random variable 
 \begin{align*}
	H_k\left((\X^t)_{t \ge 1},(\Y^t)_{t \ge 1} \right) &= h \left(\X^{k} \right) +\sum_{t= k+1}^{T-1}\left(h(\X^{t}) -h(\Y^{t})\right)
	&k\geq 0,
\end{align*}
is an unbiased estimator of $\mathbb{E}_\pi[h]$.
Note that $H_k = h \left(\X^{k} \right)$ if $k+1 > T-1$. 
Taking the average of $H_l$ for $l\in\{k,\dots,m\}$, where $k$ is a burn-in value and $m$ a maximum number of iterations, leads to the unbiased estimator \citep{jacob2019unbiased}
\begin{equation}\label{eq:unbiased_estimator} 
		\begin{aligned}
			H_{k:m}\left((\X^t)_{t \ge 1},(\Y^t)_{t \ge 1} \right) &= \frac{1}{m-k+1} \sum_{l=k}^m h(\textbf{X}^l) \\
	  &+ \sum_{l=k+1}^{T-1} \min \left(1, \frac{l-k}{m-k+1}\right) \left(h(\X^{l}) -h(\Y^{l-1})\right)
			&0\le k <m,
		\end{aligned}
	\end{equation}
which coincides with the usual MCMC ergodic average estimate plus a bias correction term.  

Standard guidelines in \cite{jacob2019unbiased} suggest to choose $k$ as a large quantile of the meeting time $T$ and $m$ as a multiple of $k$. Hence, for the method to be most practical, the meeting time should occur as early as possible.

\add{The introduction of a bias correction term makes UMCMC amenable to parallelization: independent unbiased estimates can be computed on different machines and then averaged. It can be shown that this does not lead to a significant increase in variance, in particular if the lagged variants of \eqref{eq:unbiased_estimator} are considered \citep{VanettiDoucet2020}. 
UMCMC also allows to estimate the rate of convergence in total variation distance or Wasserstein distance \citep{llag}.}

\subsection{Blocked Gibbs Sampler kernels}
\label{ssec:bgs_notation}
We now formally define BGS kernels. 
Let $\x=(\x_{1},\dots,\x_{K}) \sim \pi$, with $\pi\in\sP(\sX)$ and $\sX=\sX_1\times\cdots\times\sX_K$ partitioned in $K$ blocks of dimension $I_k$ for $k=1,...,K$, i.e.\ $\x_{k}\in\sX_k \subseteq \mathbb{R}^{I_k}$. 
We indicate by $\x_{-k}=(\x_{j})_{j\neq k}$ the whole vector except the $k$-th block and by $\pi\left(\x_{k}|\x_{-k} \right)$ the so-called full conditional distribution of the $k$-th block.
For a given (possibly random) updating sequence of the blocks $(k_1,\dots,k_s)\in \{1,\dots,K\}^s$ with $s\in\mathbb{N}$, a BGS iteratively samples from the full conditional distributions of each block given the current values of the other blocks. The resulting kernel can be written as the following composition of $s$ kernels
\begin{align}
	P&=P_{k_s} \cdots P_{k_1}\,,
	\label{eq:Gibbs_kernel}\\
	P_k(\x,d\x')
	&=
	\pi(d\x'_{k}|\x_{-k})\delta_{\x_{-k}}(d\x'_{-k})\,,
	\qquad \x\in\sX\,.\label{eq:kernel_k_update}
\end{align}
Various BGS variants can be derived depending on the chosen updating order.
For example, if we consider the sequence $(1,\dots,K)$ we obtain the (deterministic-scan) \emph{forward} version of BGS, which we denote with $P^{(F)}$.
Other natural updating orders include the backward order, as well as the forward-backward or random-scan versions.

\subsection{Two-step couplings}
\label{ssec:2step}

In this paper we consider couplings $\bar{P}$ that follow a two-step strategy: whenever the chains are \enquote{far away} we employ \emph{contractive} couplings $\bar{P}^c$ whose aim is to bring the chains closer to each other; when the chains are \enquote{close enough} we employ maximal couplings $\bar{P}^m$ (strategies to implement maximal couplings can be found in the supplement), which makes the chains meet with a certain probability. A detailed description of this strategy can be found in Section \ref{sec:methodology} below.

Two-step couplings have been previously used in the literature, see e.g.\ \cite{ROBERTS2002, Beskos2005,hmc_eberle,  biswas2019estimating}. The motivation behind this construction is that \emph{one-step} couplings, which aim for exact chain meeting at each step, are generally suboptimal in terms of meeting times \citep{Griffeath75}. The intuitive reason is that high meeting probability and effective contraction are typically separate qualities in couplings: when a maximal coupling fails, preserving marginals might imply sampling distant points in $\mathcal{X}$, thus reducing the probability of meeting in subsequent steps.  
\add{For these reasons, and for mathematical convenience, we mainly focus on \emph{two-step} couplings.
Nevertheless, we empirically observe that \emph{one-step} couplings also perform well in our applications (see Section \ref{ssec:numerics_gcrem}). Providing theory supporting this empirical findings is an interesting direction for future work.}

\subsubsection{Coupling strategies for blocked Gibbs samplers}
\label{sec:methodology}
To implement the two-step strategy described in the previous section, we need to specify $\bar{P}^c[P]$ and $\bar{P}^m[P]$, when $P$ is defined as a composition of kernels, i.e.\ $P=P_{k_s}\cdots P_{k_1}$ as in \eqref{eq:Gibbs_kernel}. A natural strategy is to sequentially compose maximal or optimally contractive couplings of $P_{k_i}$ for $i=1,\dots,s$. We denote the resulting coupling kernels as
\begin{align}
	\bar{P}^{m}((\x,\y), \cdot) &= \bar{P}_{max}[P_{k_s}] \cdots \bar{P}_{max}[P_{k_1}] \left((\x,\y), \cdot\right) &\forall \x,\y \in \mathcal{X},  \label{eq:coupled_kernels_m}\\
	\bar{P}^{c}((\x,\y), \cdot) &= \bar{P}_{W_2}[P_{k_s}] \cdots \bar{P}_{W_2}[P_{k_1}] \left((\x,\y), \cdot \right) &\forall \x,\y \in \mathcal{X}, \label{eq:coupled_kernels_c}
\end{align} 
where $\bar{P}_{max}[P_k]\in \Gamma_{max}[P_k]$ and $\bar{P}_{W_2}[P_k]\in \Gamma_{W_2}[P_k]$ for all $k=1,...,K$. 
A pseudocode for this strategy is provided in Algorithm \ref{alg:2s}, and it is the one we will refer to in the following theoretical analysis and numerical experiments.

By construction, $\bar{P}^{m}((\x,\y), \cdot)$ and $\bar{P}^{c}((\x,\y), \cdot)$ belong to $\Gamma[P]$. 
The appeal of $\bar{P}^{m}$ and $\bar{P}^{c}$ is that, in order to implement them, one needs to work only with the individual full conditionals involved in the original BGS scheme, which are often available in closed form, while the joint distribution $P(\x,\cdot)$ might be harder to work with. 
However these strategies are not guaranteed to be optimal: in general $\bar{P}^{m} \notin \Gamma_{max}[P]$ and $\bar{P}^{c}\notin \Gamma_{W_2}[P]$.
For example, in the case $P=P^{(F)}$, $\bar{P}^{c}$ coincides with the so-called Knothe-Rosenblatt map \citep{rosenblatt1952remarks,knothe1957contributions} of $P^{(F)}(\x,\cdot)$ and $P^{(F)}(\y,\cdot)$, which, in most cases, is different from the optimal transport one \citep[Section 2.3]{santambrogio}.
Nonetheless, we still observe very fast contraction of $\bar{P}^{c}$ in our numerics, which might be partly explained by the fact that in the Gaussian case one indeed has $\bar{P}^{c} \in \Gamma_{W_2}[P]$, as shown in the following lemma, which builds upon well known results about contractive couplings of Gaussian distributions. .
\begin{lem}[Optimality of composition of $W_2$ couplings for Gaussians]
	\label{lem:crn_opt}
	Let $\pi=N(\boldsymbol{\mu},\Sigma)$ and $\bar{P}^{c}$ as in \eqref{eq:coupled_kernels_c}, with $s=K$ and $(k_1,\dots,k_K)$ being a permutation of $(1,\dots,K)$. Then for all integers $n \ge 1$ it holds 
	$\left(\bar{P}^{c}\right)^n \in \Gamma_{W_2}[P^n]$.
\end{lem}

Finally, we defer to the supplementary material a discussion of the relationtship between the block-wise maximal coupling $\bar{P}^{m}$ in \eqref{eq:coupled_kernels_m} and the full maximal coupling of $\bar{P}[P]\in \Gamma _{max}[P]$.

\begin{algorithm}[h!]
	\textbf{Input:} \add{initial distribution $\pi_0$, kernels $P_1$, \dots, $P_K$, $\bar{P}^c$, $\bar{P}^m$, updating order $(k_1,\dots,k_s)\in \{1,\dots,K\}^s$, $\varepsilon$}\\	 
	sample $\X^{-1} \sim \pi_0, \Y^{0} \sim \pi_0$ and  $\X^{0} \sim P_{k_s}\cdots P_{k_1}(\X^{-1}, \cdot)$\\
		\While{$ \X^{t} \ne \Y^{t}$}
	{
		\If{\add{ $\|\X^{t} - \Y^{t}\| > \varepsilon$}}
		{ 
		\For{$k= k_1, \dots,k_s$}{
			$ (\X^{t+1},\Y^{t+1})  \sim  \bar{P}^c[P_k] ((\X^{t},\Y^{t}), \cdot)$
		}}
		\Else
		{
			\For{$k= 1, \dots,K$}{\add{
			$ (\X^{t+1},\Y^{t+1})  \sim  \bar{P}^m[P_k] ((\X^{t},\Y^{t}), \cdot)$}\\
			\If{\add{$k>1$\emph{ and maximal coupling failed for }$k'< k$}}{
				\add{$ (\X^{t+1},\Y^{t+1})  \sim  \bar{P}^c[P_k] ((\X^{t},\Y^{t}), \cdot)$}
			}
		}}
		$t \leftarrow t+1$ \\
	}
	\textbf{Output:} trajectory $(\X^t,\Y^t)_{t\in\{0,\dots,T\}}$
	\caption{\add{Blocked two-step coupling algorithm}}
	\label{alg:2s}
\end{algorithm}

\section{Bounds for couplings of Gaussian Gibbs Samplers}
\label{sec:bound}
In this section, we provide bounds on the expected meeting time of BGS coupled via Algorithm \ref{alg:2s} when the target distribution is Gaussian.

\add{Throughout  this section we take $\sX_k=\mathbb{R}^{I_k}$ for $k=1,\dots,K$, so that $\sX=\mathbb{R}^d$ with $d=I_1+\dots+I_K$, and $\pi=N(\boldsymbol{\mu}, \Sigma)$ a $d$-dimensional multivariate Gaussian.
In this case,  the Markov chain induced by BGS takes the form of a Gaussian auto-regression (see Lemma \ref{lem:rob} below).
Finally, recall that the \emph{relaxation time} of an irreducible $\pi$-reversible kernel $P$ is defined as $T_{rel} = 1 / AbsGap(P)$,  where $AbsGap(P)= 1- \sup_{\lambda \in \sigma(P), |\lambda| \ne 1} |\lambda|$ is the (absolute) spectral Gap of $P$ and $\sigma(P)$ denotes the spectrum of $P$. 
Relaxation times are closely related to mixing times \citep[see e.g.,][Section 12]{levin2017markov} and can be interpreted as the number of iterations needed for the chain to be $\varepsilon$ close to the target distribution, up to a multiplicative factor that depends on $\varepsilon$ and on the starting distribution (see also \cite{T_rel_ros} for more discussion on the link between convergence, asymptotic variances and the spectrum of reversible Markov chains).}

\add{In Section \ref{sec:rel_times}, we will show that $T_{rel}$ provides an upper bound (up to logarithmic factors) on the expected meeting time $\mathbb{E}[T]$.
The underlying assumption is that $T_{rel}$ might be informally seen as providing a lower bound on the size of $\mathbb{E}[T]$.
While we are not aware of rigorous results in this direction, this seems plausible given that $T_{rel}$ provides lower bounds on total variation mixing times \citep[Section 12]{levin2017markov} and that the quantiles of $T$ can be used to derive non-asymptotic upper bounds on those \citep{llag}. 
While interesting, we leave a more detailed and rigorous exploration of lower bounds to $\mathbb{E}[T]$ to future work.
Finally, in Section \ref{sec:bound_gcrem} we will relate the results of this section to the ones available in the literature for crossed random effect models \citep{papaspiliopoulos2018scalable}.}

\begin{lem}\label{lem:rob}
Let  $\pi = N(\boldsymbol{\mu}, \Sigma)$ and $P$ be a BGS kernel with updating order given by $(k_1,\dots,k_s)\in \{1,\dots,K\}^s$ for some $s\in\mathbb{N}$, i.e.\ $P=P_{k_s}\cdots P_{k_1}$ with $P_k$ defined in \eqref{eq:kernel_k_update}.
Then, one has
	\begin{equation}
		\label{eq:gaus_kernel_form}
		P(\x, \cdot) = N \left( B \x + \textbf{b}, \Sigma -B \Sigma B^\top \right),
	\end{equation}
where $B$ depends on the updating order $(k_1,\dots,k_s)$ and on the target precision matrix $Q=\Sigma^{-1}$, and $\mathbf{b}=(I-B)\boldsymbol{\mu}$. 
	Furthermore the relaxation time of $P$ is given by $T_{rel}= 1/(1-\rho(B))$, where $\rho(B)$ denotes the largest modulus eigenvalue of $B$.
\end{lem}
Lemma \ref{lem:rob} is a well-known result, whose proof we omit, see e.g.\ \citet[Lemma 1]{robsahu} for the explicit expression of $B$ in the forward updating case, corresponding to $s=K$, $k_i=i$ for all $i$ and $P=P^{(F)}$.

\subsection{Bound for reversible chains}
\label{ssec:bound_pirev}
Our first bound applies to $\pi$-reversible BGS kernels, i.e.\ one where the updating order in Algorithm \ref{alg:2s} satisfies $(k_1,\dots,k_s)=(k_s,\dots,k_1)$.
A classical example is the forward-backward kernel, defined as 
$P^{(FB)}=P_1\cdots P_{K-1} P_K P_{K-1}\cdots P_1$. 
Algorithmically, $P^{(FB)}$ performs updates from $\pi(\x_{k}\mid \x_{-k})$ sequentially for $k=1,\dots,K-1,K,K-1,\dots,1$.  
If $P$ is a $\pi$-reversible BGS kernel, it holds $\Sigma B^\top = B \Sigma$, with $B$ as in  Lemma \ref{lem:rob} \citep[see e.g.\ Proposition 4.27 of ][]{Khare2009RATESOC}. This allows for neater theoretical results.  
In the supplementary material, we extend the result to non-reversible Gibbs samplers, such as those generated by the forward kernel $P^{(F)}$, where the result requires additional technical assumptions.

\add{We introduce the notation required for our main result.
For a given covariance matrix $\Sigma$, we denote its inverse by $Q=\Sigma^{-1}$, and by $Q_{kk}$ the $k$-th diagonal block of $Q$. Additionally, we define $\Delta = diag (Q_{11} ,\dots, Q_{KK} )$ as a block-diagonal matrix, and $\bar{Q} = \Delta ^{-1/2} Q \Delta ^{-1/2}$.}

\begin{thm}[Bound for reversible chains]
	\label{thm:bound_expected_rev}
\add{Let $\pi=N(\boldsymbol{\mu},\Sigma)$ and $(\X^{t}, \Y^{t})_{t\ge 0}$ be a Markov chain marginally evolving with $\pi$-invariant BGS kernel and coupled via Algorithm \ref{alg:2s}. Assume $\varepsilon$ satisfies
\begin{equation}
  \label{eq:eps_cond_block}
  \varepsilon \leq \frac{\sqrt{\pi / 2}}{K \max (1 , \rho(\bar{Q})) \, \sqrt{\rho (\Delta )}}\,,
\end{equation}
with $Q$, $\Delta$, and $\bar{Q}$ defined above. Then $T= \min\{ t \ge 0\,:\, \X^{t} = \Y^{t}\}$ satisfies
	\[
		\mathbb{E}[T \mid \X ^0 , \Y^0] \leq 3 + \frac{4 + \log \|\X^{0}-\Y^{0}\|- 2\log \varepsilon + 2.5\log \kappa(\Sigma) +  \log K + \log \rho(\Delta ^{-1}) }{-\log \rho(B)} \,,
	\]
with $B$ as in \eqref{eq:gaus_kernel_form}.
If \eqref{eq:eps_cond_block} holds with equality, the upper bound can be simplified to
\begin{equation}\label{eq:bound}
	\mathbb{E}[T \mid \X ^0 , \Y^0] \leq 3 + \frac{4+\log \|\X^{0}-\Y^{0}\|+3\log K   +6  \log \kappa(\Sigma) }{-\log \rho(B)}\,.
\end{equation}
}
\end{thm}

\add{We make some comments on the upper bound in \eqref{eq:bound}. By Lemma \ref{lem:crn_opt}, the contraction coupling $\bar{P}^c$ employed in Algorithm \ref{alg:2s} is $W_2$ optimal for Gaussian targets. As a consequence, the BGS kernel reduces the expected distance between the chains at geometric rate $\rho(B)$, which appears in the denominator.
Since the BGS kernel contracts exponentially fast, the term in the numerator appears in logarithmic form.
In particular, the $\log K$ factor accounts for the number of maximal couplings that needs to be successful for the two chains to meet.
The $\log \kappa(\Sigma)$ terms reflect the fact that the distance between the chains in Algorithm \ref{alg:2s} is measured with the euclidean distance instead of the ``correct'' metric induced by $\Sigma$, namely $d(\x, \y) = \sqrt{(\x-\y)^\top \Sigma^{-1} (\x-\y)}$.
For this reason, even if the Gibbs sampler is invariant under block diagonal linear transformations that preserve the $K$-partite block structure, the distribution of $T$ and the bound in \eqref{eq:bound} are not.}

For the proof of Theorem \ref{thm:bound_expected_rev}, we used the following bound on the expected squared distance between Gaussian distributions coupled via maximal reflection coupling (whose description can be found in the supplementary material),
which may be of independent interest.
\begin{lem}
	\label{lem:bound_gen}
\add{Let $p=N(\bb \xi, \Sigma)$ and $q=N(\bb\nu, \Sigma)$ be $d$-dimensional Gaussians, and $(\X,\Y)\in \Gamma_{max}(p,q)$ coupled via maximal reflection coupling. 
Then for every $A \in \mathbb{R}^{m \times d}$, we have 
\begin{align*}
    2\dfrac{\|A(\bb \xi-\bb \nu)\|}{\left\|\Sigma^{-1 / 2}(\bb \xi-\bb \nu)\right\|^{2}} 
    & \leq
    \mathbb{E}\left[\|A(\X-\Y)\|\mid \X \neq \Y\right]
   \leq\dfrac{\|A(\bb \xi-\bb \nu)\|}{\left\|\Sigma^{-1 / 2}(\bb \xi-\bb \nu)\right\|^{2}} \,\times
    \\
    &
    \left(2\,e ^{- \frac{1}{2}\left\|\Sigma^{-1 / 2}(\bb \xi-\bb \nu)\right\|^{2}} + \sqrt{\frac{8}{\pi}} \left\|\Sigma^{-1 / 2}(\bb \xi-\bb \nu)\right\| + \left\|\Sigma^{-1 / 2}(\bb \xi-\bb \nu)\right\|^ 2\right)
\,.\end{align*}
}
\end{lem}
\add{Note that, for fixed $\Sigma$ and $A$, Lemma \ref{lem:bound_gen} shows that the expectation of $\| A(\X - \Y )\|$ scales as $\bigo{\|\boldsymbol{\xi}-\boldsymbol{\nu}\|^{-1}}$ as $\|\boldsymbol{\xi}-\boldsymbol{\nu}\| \rightarrow 0$.
Lemma \ref{lem:bound_gen} is also informative in the case where the maximal coupling is attempted when the means are ``far away''. It shows that, even if the coupling is unsuccessful, $\| A(\X - \Y )\|$ does not diverge and is of the same order of $\|A(\boldsymbol{\xi}-\boldsymbol{\nu})\|$, as $\|\boldsymbol{\xi}-\boldsymbol{\nu}\| \rightarrow +\infty$. This partially explains the success of the one-step coupling in the experiments of Section \ref{ssec:non_gauss_crem_simulations} and \ref{sec:pmf}. However, extending the analysis to maximal coupling across successive blocks is particularly challenging and we leave such investigation to future work.
}

\subsection{Connection to relaxation times}\label{sec:rel_times}
Theorem \ref{thm:bound_expected_rev} can be interpreted in terms of the relaxation time $T_{rel}$ of the BGS kernel.
\begin{cor}
	\label{cor:bound_rev}
	Under the assumptions of Theorem \ref{thm:bound_expected_rev}, denoting $T_{rel}=1/(1-\rho(B))$, we have
	\begin{equation}\label{eq:bound_pirev_t_rel}
		\mathbb{E}[T|\X^{0}, \Y^{0} ] \leq 3 + T_{rel} \left[ \log \left(4\, \| \X ^0 - \Y ^0 \| \, K ^3 \, \kappa (\Sigma) ^6 \right) \right]\,.
	\end{equation}
\end{cor}

Corollary \ref{cor:bound_rev} provides interesting insights and implications. 
In particular, interpreting $T_{rel}$ as the number of iterations required for $\X^{t}$ to converge, it suggests that in this context UMCMC provides unbiased estimates with an average number of iterations (and an overall computational cost) that is comparable to the minimal number of iterations required by standard MCMC to converge (up to a logarithmic factor).
Also, from a high-dimensional asymptotics perspective, it also implies that whenever the relaxation time of BGS is bounded as the number of data points and parameters grows (see e.g.\ Section \ref{sec:bound_gcrem}), then also the meeting time is bounded in expectation, meaning that UMCMC does not increase the overall complexity (while allowing for e.g.\ early stopping and parallelization).
On the other hand, \eqref{eq:bound_pirev_t_rel} implies that whenever the meeting times of the \emph{two-step} strategy diverge for a chosen BGS scheme, also the respective $T_{rel}$ diverges. 

\subsection{Bounds on tails of the meeting time distribution}
\label{sec:tail_bound}

\add{One of the main advantages of UMCMC is that it allows for early stopping of a Markov chain as soon as the meeting time occurs. As a consequence, improvements in precision are obtained not by extending a single run, but by generating multiple unbiased estimators in parallel. In this context, the upper bound of Theorem \ref{thm:bound_expected_rev} on the average meeting time is only partially informative: such a bound is most relevant when estimators are produced sequentially. The following result shows that the tail of the meeting time distribution decreases exponentially fast, up to a logarithmic factor, providing more direct theoretical guarantees for efficient parallelization of UMCMC schemes.}

\begin{thm}\label{thm:tail_bound_block}
\add{Under the assumptions of Theorem \ref{thm:bound_expected_rev}, 
for every 
  \[
    t> \max \left(\alpha ^{-1} ,\, \left\lceil\frac{\log \left\|\X^{0}-\Y^{0}\right\| + \frac{1}{2}\log \kappa (Q) -\log (\varepsilon)}{-\log \rho(B)}\right\rceil \right)\,
   \]
   with $\alpha = - \log \rho (B) \kappa (Q)^2 \rho(\Delta^{-1})\varepsilon^{-2}$, we have
\begin{equation}\label{eq:tail_bound}
    Pr ( T > 1 + 2t ) \leq 14\,\exp \left( \frac{t \,\log \rho (B)}{2\,\log (\alpha t) + 7 }\right)\,.
  \end{equation}
}
\end{thm}

\add{Theorem \ref{thm:tail_bound_block} implies that $\mathbb{E}[T^k] < +\infty$ for any $k>1$, and thus, by \citet[Theorem 2.1]{atchade2024unbiasedmarkovchainmonte}, it follows that the unbiased estimator \eqref{eq:unbiased_estimator} of $\mathbf{E}_\pi [h]$, with $h \in L^m(\pi)$, has finite moments of order $p$ for any $p<m$.}

\section{Application to Gaussian crossed effect models}\label{sec:crem}
In this section, we combine the findings of Section \ref{sec:bound} with existing results on Model \ref{ex:gcrem}. 
We highlight that all the theoretical results we will derive hold under the assumption of fixed $\boldsymbol{\tau}$ in \eqref{eq:gauss_crem} of Model \ref{ex:gcrem}.
We first describe in Section \ref{ssec:cgs} state-of-the-art marginal algorithms for Model \ref{ex:gcrem}, then present in Section \ref{sec:bound_gcrem} the resulting bound on the meeting times if a two-step coupling is implemented, and finally report numerical simulations in Section \ref{ssec:numerics_gcrem}.

\subsection{Collapsed Gibbs sampler for Model \ref{ex:gcrem}}
\label{ssec:cgs}
Despite the favorable cost per iteration of the vanilla Gibbs sampler for Model \ref{ex:gcrem} presented in Section \ref{ssec:comp_cost}, there are many settings of interest where its mixing is provably poor, often leading to a super linear overall computational cost. \cite{papaspiliopoulos2018scalable} noted that integrating out the global mean $\mu$  while updating the remaining regression parameters in $K$ blocks, leads to a much more efficient (i.e faster mixing) updating scheme, while preserving the same $\bigo{N}$ cost per iteration of vanilla BGS. The resulting algorithm is called \emph{collapsed Gibbs sampler} \citep{papaspiliopoulos2018scalable} and reported in  Algorithm \ref{alg:cg}: at every iteration we first sample from $\mathcal{L}(\mu|\mathbf{a}_{-k}, \boldsymbol{\tau}, \mathbf{y})$ and then iteratively update the factor effect block from $\mathcal{L}(\mathbf{a}_{k}| \mu, \mathbf{a}_{-k}, \boldsymbol{\tau}, \mathbf{y})$, repeating the procedure for $k=1,..,K$.
\begin{algorithm}
	\For{k= 1,...,K}
	{
		draw $\mu \sim \mathcal{L}(\mu | \mathbf{a}_{-k}, \boldsymbol{\tau}, \mathbf{y}) 
		$\\
		\For{ i=$1,..., I_k$}
		{draw $a_{k,i} \sim \mathcal{L}(a_{k,i}|\textbf{a}_{-k},\mu, \boldsymbol{\tau}, \textbf{y})
			$ 
		}
		draw $\tau_k \sim \mathcal{L}(\tau_{k}|\textbf{a}, \mu, \boldsymbol{\tau}_{-k}, \textbf{y})$
	}
	draw $\tau_0 \sim \mathcal{L}(\tau_0|\textbf{a}, \mu, \boldsymbol{\tau}_{-0}, \textbf{y})$
	\caption{One iteration of the collapsed Gibbs sampler for Model \ref{ex:gcrem}}
	\label{alg:cg}
\end{algorithm}

\subsection{Bound on meeting times under random design assumptions}\label{sec:bound_gcrem}
\label{ssec:bound_crem}
For Model  \ref{ex:gcrem} with $K=2$ factors, balanced level designs (i.e.\ the same number of observations is observed for every level of each factor) and fixed $\boldsymbol{\tau}$, \cite{papaspiliopoulos2018scalable} show that the relaxation time of the collapsed algorithm, denoted by $T_{cg}$, is upper bounded by $T_{cg} \le C \, T_{aux}$, where $C = 1+ \frac{\tau_0}{\min\{ \tau_1, \tau_2\} }$ is constant with respect to $N$ and $p$, and $T_{aux}$ is the relaxation time of the auxiliary two-block Gibbs sampler on the discrete space $\{ 1,..., I_1\} \times \{1,..., I_2\}$ with invariant distribution $\Pr\left((i,j) \right)=n_{ij}/N$, where $n_{ij}=\sum_{n=1}^N\mathbb{I}(i_1[n]=i)\mathbb{I}(i_2[n]=j)$ denotes the number of observations of level $i$ of factor 1 and $j$ of factor 2.
Under random design assumptions, the quantity $T_{aux}$ can be bounded using random graph theory results, as done in \cite{tim_scalable}. In particular, for $N$ multiple of $d_1$ and $d_2$, denote by $\mathcal{D}(N,d_1,d_2)$ the collection of all the possible observation patterns with exactly $N$ observations, $I_1=N/d_1$, $I_2=N/d_2$ and binary balanced levels (i.e.\ there must be exactly $d_1$ and $d_2$ observations for each level of factor 1 and 2 respectively and $n_{ij}\in \{0,1\}$ for all $i=1,...,I_1$ and $j=1,...I_2$). 
Then, supposing uniformly at random designs among $\mathcal{D}(N, d_1,d_2)$ with $d_1,d_2 >4$, one has
$
T_{aux} \le 1 + 2 (\min\{d_1,d_2\}-2)^{-\frac{1}{2}} + \gamma, 
$
asymptotically almost surely as $N \rightarrow + \infty$, for every $\gamma > 0$. The result follows from relating $T_{aux}$ to the spectrum of a random bipartite bi-regular graph, and then applying an extension of the Friedman's second largest eigenvalue theorem to bipartite graphs developed in \cite{Brito2018SpectralGI}. \add{Combining the above with Corollary \ref{cor:bound_rev}, we obtain the following bound for the expected meeting time, which notably extends to the case $P= P^{(F)}$ in the case $K=2$.}
\begin{cor}\label{cor:gcrem_bound}
	\add{Let $\pi=N(\boldsymbol{\mu}, \Sigma)$ be the posterior distribution of Model \ref{ex:gcrem} with $K=2$ factors, fixed $\boldsymbol{\tau}$.
	Let $\left(\X^{t},\Y^{t}\right)_{t\ge 0}$ be a Markov chain marginally evolving with $P^{(F)}$, coupled via Algorithm \ref{alg:2s}.
	Then, if $\varepsilon$ satisfies \eqref{eq:eps_cond_block}, it holds
	\[\Pr\left(\mathbb{E}[T|\X^{0},\Y^{0} ] \leq
	6 + C(\bb \tau , \, \| \X ^0 - \Y ^0 \|,\, \kappa(\Sigma))
 \left( 1 + \frac{2}{\sqrt{\min\{d_1,d_2\}-2}} + \gamma  \right)
 	\right)\to 1\,, \]
as $N \rightarrow + \infty$, with $C(\bb \tau , \, \| \X ^0 - \Y ^0 \|,\, \kappa(\Sigma)) = \left( 1+ \frac{\tau_0}{\min\{ \tau_1, \tau_2\} }\right) \log \left(32\, \| \X ^0 - \Y ^0 \| \, \kappa (\Sigma) ^6 \right)$, and the probability is with respect to a random design $(n_{ij})_{ij}\sim \mathcal{D}(N,d_1,d_2)$.
}
\end{cor}

Interestingly, Corollary \ref{cor:gcrem_bound} provides an upper bound on the average coupling time that remains bounded as $N$ and $p$ diverge.

\subsection{Numerics}\label{ssec:numerics_gcrem}
We compare the bound of Theorems \ref{thm:bound_expected_rev} with the average meeting time of simulated coupled chains for both the vanilla and collapsed Gibbs samplers of Sections \ref{ssec:comp_cost} and \ref{ssec:cgs}, for synthetic and real data in Section \ref{sssec:sim_data} and \ref{sssec:real_data} respectively.
Although the bound is valid only for Gaussian chains (Model \ref{ex:gcrem} with fixed $\boldsymbol{\tau}$), we compare them with the average meeting time of coupled chains with known (i.e.\ fixed) as well as unknown (i.e.\ assigning to it a prior and including it into the Bayesian model) $\boldsymbol{\tau}$, yielding similar behaviors.
The results support the intuition that, for the models under consideration, the convergence properties of the known and unknown variances case are similar and the bounds are reasonably predictive also of the behavior in the practically-used unknown variance case.

We implement the two-step coupling procedure of Algorithm \ref{alg:2s} for both the vanilla and the collapsed Gibbs sampler of Algorithm \ref{alg:cg}, with the standard priors in \eqref{eq:gauss_crem}. 
As $\bar{P}_{max}[P_k]$, we use the  maximal reflection coupling whenever implementable, and maximal rejection coupling otherwise. As $\bar{P}_{W_2}[P_k]$ we use the \emph{common random number} (crn) coupling. More details are reported in the supplementary material. 
The  threshold parameter $\varepsilon$ in Algorithm \ref{alg:2s} is set to $\bigo{ (K \, I)^{-1}}$.

\subsubsection{Simulated data}
\label{sssec:sim_data}
We simulate data according to Model \ref{ex:gcrem}, with $\tau_0=\tau_1=...=\tau_k = 1$, $I_1=...=I_K=I$ for fixed $I$,  and different number of factors $K$.
Observations are generated according to two different asymptotic regimes, with completely missing at random designs.
\begin{regime}
	\label{reg1}
 Each combination of factor levels is observed once or not with probability $\prob=0.1$, independently from the rest, i.e.\ $n_{ij}^{(s,l)}\stackrel{iid}\sim Bern(\prob)$ for $i=1,\dots,I_s, \, j=1,\dots,I_{l}$ and $s \ne l \in \{1,...,K\}$. Here $n_{ij}^{(s,l)}=\sum_{n=1}^N\mathbb{I}(i_s[n]=i)\mathbb{I}(i_l[n]=j)$ denotes the number of observations of level $i$ of factor $s$ and $j$ of factor $l$. In this regime $ I = \bigo{N^{1/K}}$. 
\end{regime}
\begin{regime}
	\label{reg2}
 Same as {Regime 1} but with $\prob= 10/I^{K-1}$. This regime induces more sparsity and one has $ I = \bigo{N}$. 
\end{regime}

We plot the average of the meeting times as a function of the total number of parameters of the model, i.e.\ $1+KI$ plus the number of scale parameters if any. 
Figure \ref{fig:two} reports results for the collapsed Gibbs sampler (left) and vanilla BGS (right), for $K=2$, $I_1= I_2=I \in\{50, 100,250,500,1000 \}$ levels, Regime \ref{reg1}, fixed and free variances.
For each scheme, the correspondent bound of Theorem \ref{thm:bound_expected_rev} is also reported, using the true data generating values for the variance parameters.
\begin{figure}[h!]
	\centering
	\begin{subfigure}{.5\textwidth}
		\centering
		\includegraphics[width=\linewidth]{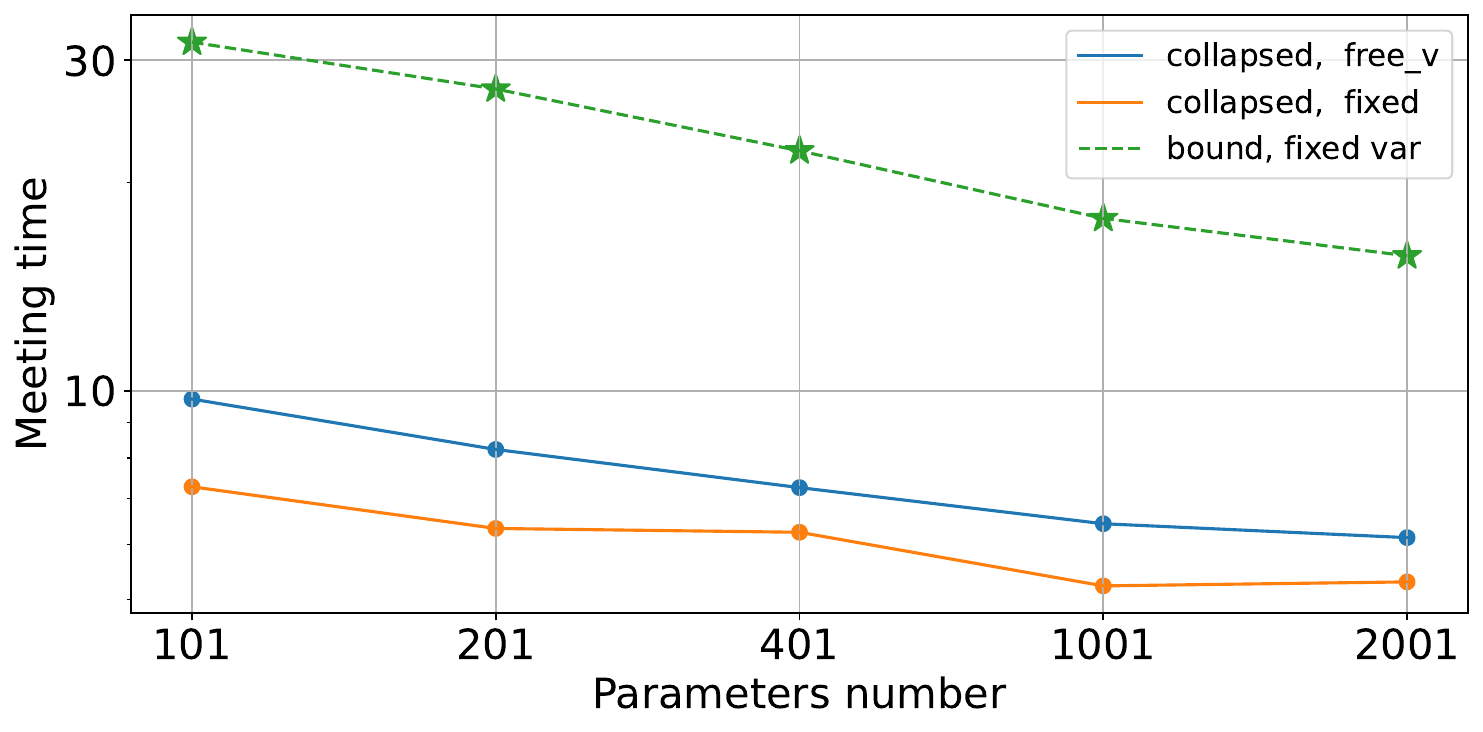}
	\end{subfigure}%
	\begin{subfigure}{.5\textwidth}
		\centering
		\includegraphics[width=\linewidth]{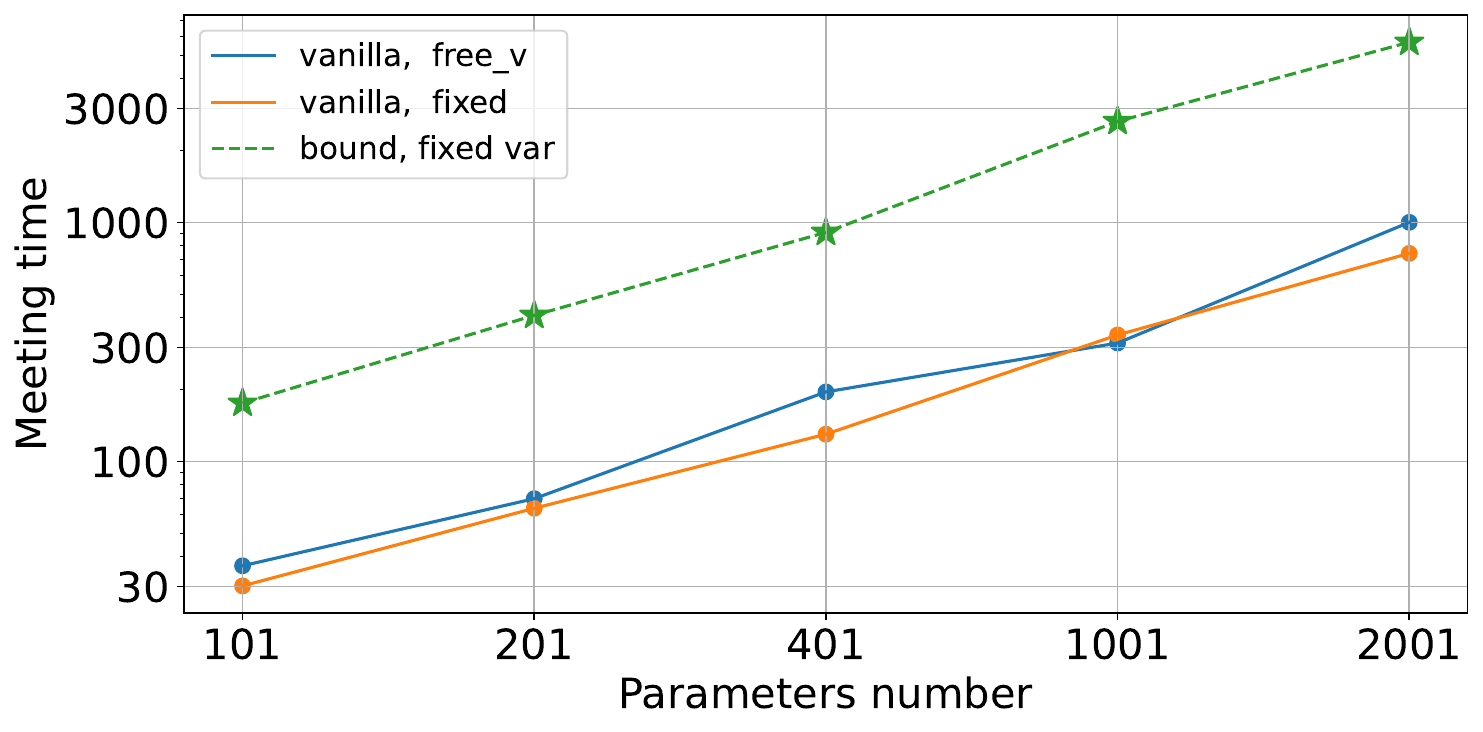}
	\end{subfigure}
	\caption{\add{Average meeting times and estimated bounds in a log-log scale for $K=2$, $I_1 = I_2$, $\tau_0=\tau_1 = \tau_2 =1$, Regime \ref{reg1}. Left: Algorithm  \ref{alg:cg}, right: vanilla BGS.}}
	\label{fig:two}
\end{figure}
The results yield remarkably low meeting times and highlight a close resemblance of the meeting time behavior with that of the bound. As expected, the provably higher relaxation time of the vanilla Gibbs scheme results in a higher bound and, more importantly, in a higher average meeting time of the coupled chains.

\add{In Figure \ref{fig:diff_k} we report the average meeting time for $K=4$ factors, for Regime \ref{reg1} (left) and Regime \ref{reg2} (right), and the bound of Theorem \ref{thm:bound_expected_rev}.} For these models, the relaxation time is not computable explicitly even under the usual simplifying assumptions (fixed variances and balanced levels or cells), see e.g.\ \cite{papaspiliopoulos2018scalable}. Thus proving the scalability of the meeting times (or lack thereof), in light of Theorem \ref{thm:bound_expected_rev}, provides interesting insights on the mixing properties of the single chains themselves.

\begin{figure}[h!]
\centering
\begin{subfigure}{.5\textwidth}
 \centering
 \includegraphics[width=\linewidth]{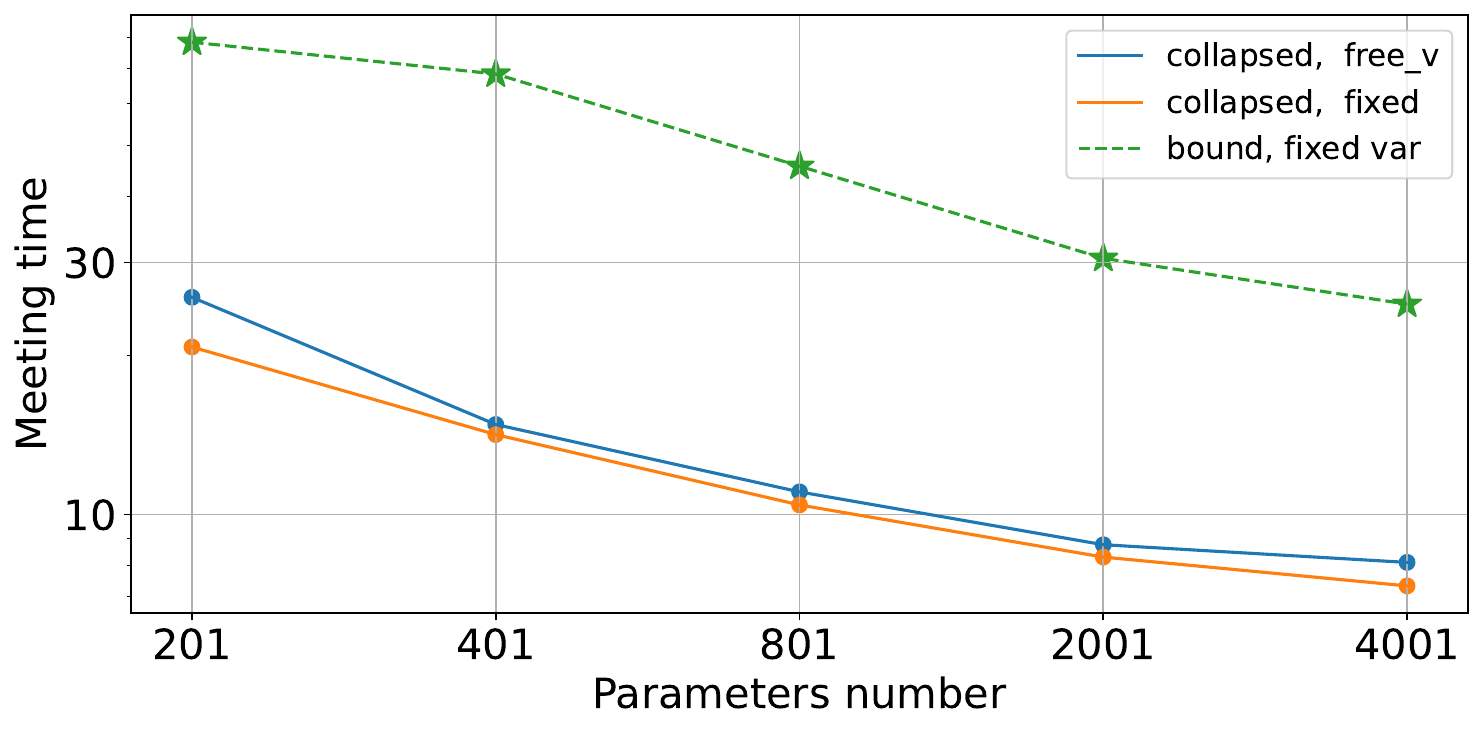}
\end{subfigure}%
\begin{subfigure}{.5\textwidth}
 \centering
 \includegraphics[width=\linewidth]{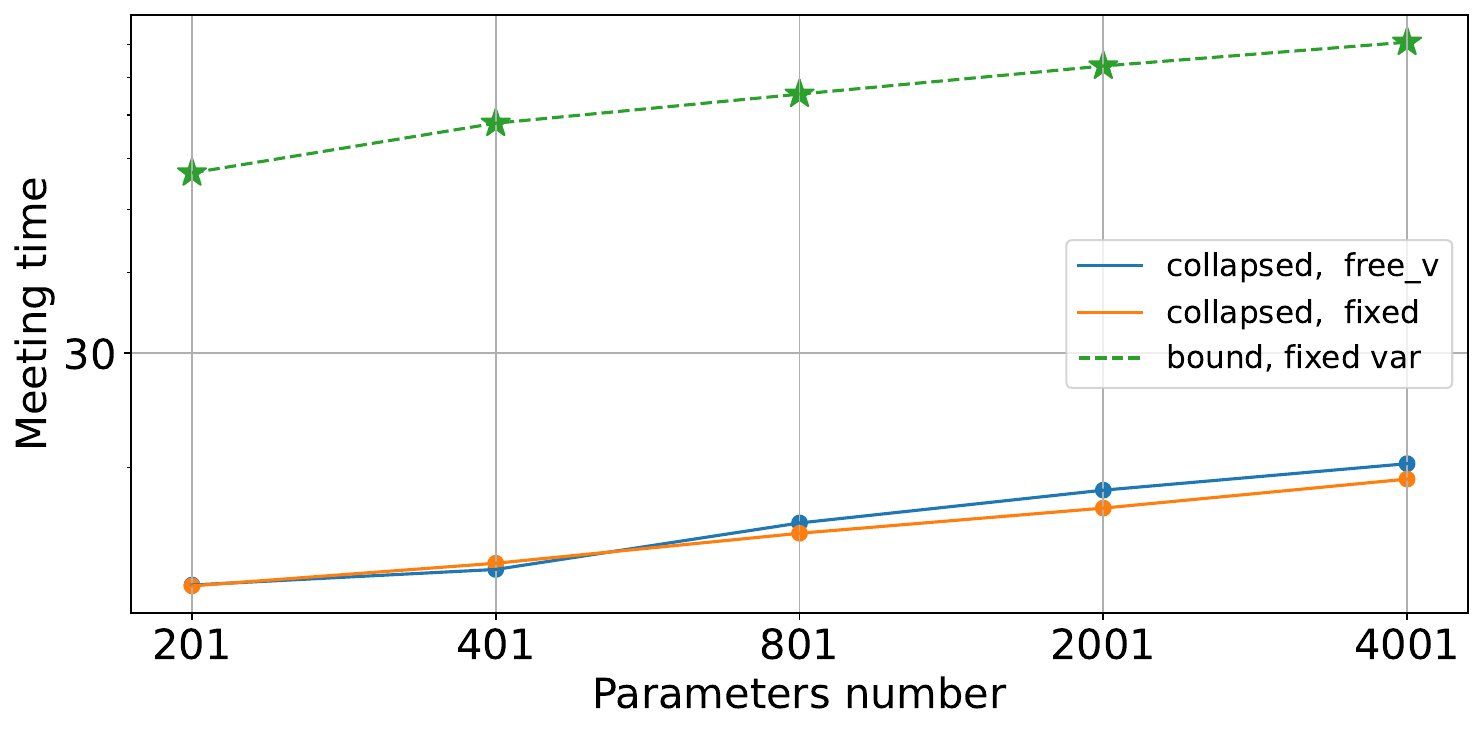}
\end{subfigure}
\caption{\add{Average meeting times and estimated bounds in a log-log scale for $K=4$, $I_1= ... = I_4$,  $\tau_k = 1$ for $k\in\{0,1,2,3,4\}$. Regime \ref{reg1} (left), Regime \ref{reg2} (right), Algorithm \ref{alg:cg}.} }
\label{fig:diff_k}
\end{figure}

\add{The blessing of dimensionality under Regime \ref{reg1} apparent in Figures \ref{fig:two} and \ref{fig:diff_k} is coherent with Corollary \ref{cor:gcrem_bound}, which shows that expected meeting time decreases as $\bigo{1 + d ^{-1/2}}$ as $N \to +\infty$, where $d$ is the number of observation per level. Under Regime \ref{reg1}, the expected value of $d=I_1 \prob$, increases with $N$, and $\mathbb{E}[T]$ should decrease accordingly. A more extensive analysis can be found in \cite{papaspiliopoulos2018scalable,andrea2024}.}

\begin{rem}
	Interestingly, unlike many sampling algorithms (such as gradient-based ones), no adaptation of tuning parameters is required for BGS.
	This couples particularly well with the UMCMC framework: specifically, it avoids the need for potentially long preliminary runs or adaptation phases, thus genuinely allowing for parallelizable short runs and \enquote{early stopping} in case of fast mixing chains, see also \cite{biswas2019estimating} for analogous examples.
\end{rem}

\subsubsection{Real data example}
\label{sssec:real_data}
We now consider a real dataset, 
containing university lecture evaluations by students at ETH Zurich. The dataset is freely available from the R package \textbf{lme4} \citep{insteval} under the name \enquote{InstEval}. Each observation includes a score ranging from 1 to 5, assigned to a lecture, along with 6 factors that may potentially impact the score, including the identity of the student giving the rating or department that offers the course. Following the notation in \eqref{eq:gauss_crem}, we have $N = 73421, K = 6$ and $(I_1,...,I_6) = (2972,1128,4,6,2,14)$. 
We compute the estimated meeting times for different numbers and combinations of factors for Model \ref{ex:gcrem} with fixed variances (estimated via standard MCMC and plugged in the coupling procedure). We numerically computed the bound for each combination using the MCMC variance estimates.
\add{Table \ref{tab:res} reports the results obtained when including the first two and the first five factors of the dataset. We exclude the last factor because it is nested in the second one, leading to additional complications beyond the scope of this work (see \cite{pfvi} and \cite{andrea2024} for further discussion).}
\begin{table}
\begin{center}
	\begin{tabular}{ |c|c|c|c| } 
	\hline
	 Algorithm & Factor number & mean \#iter & bound for fixed $\boldsymbol{\tau}$\\
	\hline
	\multirow{2}{4em}{collapsed collapsed} & [1,2] & \add{6.9} & \add{36.1} \\ 
	& \add{[1,2,3,4,5]} & \add{9.3} & \add{39.7} \\ 
	\hline
	\multirow{2}{4em}{vanilla vanilla} & [1,2] & \add{22.8} & \add{85.3}\\ 
	& \add{[1,2,3,4,5]} & \add{68.3}& \add{113.4}\\ 
	\hline
	\end{tabular}
\end{center}
\captionof{table}{Average meeting times for InstEval Dataset}
\label{tab:res}
\end{table}

\section{High-dimensional GLMMs with crossed effects}
\label{sec:ngcrem}
We now consider applications to Model \ref{ex:ngcrem}. 
First, Section \ref{ssec:alg_ngcrem} reviews state-of-the-art samplers and their computational cost for this class of models, and briefly discusses our coupling strategy, which requires to extend the methodology of Section \ref{sec:methodology} to the case of Metropolis-within-Gibbs algorithms. Then Section \ref{ssec:non_gauss_crem_simulations} reports experimental results on simulated data.

\subsection{Algorithms for Model \ref{ex:ngcrem}}
\label{ssec:alg_ngcrem}
Similarly to what seen for Model \ref{ex:gcrem} in Section \ref{ssec:comp_cost}, also for Model \ref{ex:ngcrem} the vanilla Gibbs procedure, i.e.\ the one updating from the full conditionals $\mu, \mathbf{a}_{1}, ...,\mathbf{a}_{K}, \boldsymbol{\tau}$, suffers from slow mixing, due to strong posterior correlation among random and fixed effects.
Given the impossibility of analytically integrating out the global mean, \cite{tim_scalable} propose to perform at each iteration a $\mathcal{L}\left(\mu, \mathbf{a}_{k} | \mathbf{y}, \boldsymbol{\tau}, \mathbf{a}_{-k}\right)$-invariant update using local centering within each block, hence updating a new pair of variables $(\mu, \boldsymbol{\xi}_{k})$, where $\boldsymbol{\xi}_{k} = \mu + \mathbf{a}_{k}$.
For the re-parametrized model it holds that $\mathcal{L}(\mu| \mathbf{a}_{-k}, \boldsymbol{\xi}_{k}, \boldsymbol{\tau}, \mathbf{y}) = \mathcal{L}(\mu| \boldsymbol{\xi}_{k})$ and that the $\boldsymbol{\xi}_{k}$ are conditionally independent, although their full conditional might not be available in closed form. In such cases one can replace exact Gibbs updates of $\boldsymbol{\xi}_{k}$ with more general invariant Markov updates. The resulting scheme is described in Algorithm \ref{alg:non_gauss}.
\begin{algorithm}
	\For{k= 1,...,K}
	{
		reparametrize $ (\mu, \mathbf{a}_{k}) \rightarrow (\mu, \boldsymbol{\xi}_{k})$ \\
		draw $\mu$ from $\mathcal{L}(\mu | \boldsymbol{\xi}_{k}) 
		$\\
		\For{ i=$1,..., I_k$}
		{update ${\xi}_{k,i}$ with a $\mathcal{L}({\xi}_{k,i}|  \mu, \mathbf{a}_{-k}, \boldsymbol{\tau},\mathbf{y})$-invariant Markov kernel
		}
		reparametrize $ (\mu, \boldsymbol{\xi}_{k}) \rightarrow (\mu, \mathbf{a}_{k}) $ \\
    draw ${\tau}_k$ from $\mathcal{L}(\tau_k |\mu, \textbf{a}, \textbf{y})$\\
	}
	\caption{One iteration of Metropolis-within-Gibbs sampler with local centering for Model \ref{ex:ngcrem}}
	\label{alg:non_gauss}
\end{algorithm} 

For Gaussian targets and fixed $\boldsymbol{\tau}$, Corollary 1 in \cite{tim_scalable} shows that $T_{cg} < T_{lc} < T_{cg}+C,$ where $T_{cg}$ and $T_{lc}$ denote the relaxation times of Algorithms \ref{alg:cg} and Algorithm \ref{alg:non_gauss}, respectively, and $C$ is a constant depending only on $\boldsymbol{\tau}$. 
The previous inequality allows to directly extend the results developed in Section \ref{ssec:bound_crem} for the collapsed Gibbs scheme to the local centering version in Algorithm \ref{alg:non_gauss}. 
Although the inequality holds only for Gaussian targets, numerical results in \cite{tim_scalable} show that also in the non conjugate case, where sampling of $\xi_{k,i}$ is done through Metropolis-Hastings updates, the convergence speed remains bounded as $N$ and the number of parameters increase.

\begin{rem}[Couplings of Metropolis-Hastings]
Analogously to Model \ref{ex:gcrem}, $\mathcal{L}(\boldsymbol{\xi}_{k}| \mu, \mathbf{a}_{-k}, \boldsymbol{\tau},\mathbf{y})$ factorizes as $\prod_{i=1}^{I_k} \mathcal{L}({\xi}_{k,i} |\mu, \mathbf{a}_{-k}, \boldsymbol{\tau},\mathbf{y})$. The difference is that $\mathcal{L}({\xi}_{k,i} |\mu, \mathbf{a}_{-k}, \boldsymbol{\tau},\mathbf{y})$ are not available in closed form and thus we update each ${\xi}_{k,i}$ with a MH step in Algorithm \ref{alg:non_gauss}.
\add{Nevertheless, the special structure of these models (with few blocks, each with weak conditional dependence) facilitates the design of efficient couplings. E.g., in the supplementary material, we show that leveraging conditional independence in the coupling of the MH kernels allows to have meeting times that grow at most logarithmically with $I_k$.}
In the MH case, however, there is additional flexibility \citep{OlearyWang}, such as deciding how to couple both the proposal and acceptance steps as well as which of those to factorize.
In the supplement, we discuss these aspects in some detail, suggesting to use a fully factorized strategy both on the proposal and acceptance.
\add{Note that implementing efficient contractive couplings of MH kernels is particularly challenging. Alternatively, one could consider rejection-free kernels, such as slice samplers \citep[see e.g.][]{biswas2019estimating}. We leave the investigation of such strategies to future work.
In the following experiments, we avoid the two-step strategy of Algorithm \ref{alg:2s} and instead use a one-step approach. }
\end{rem}

\subsection{Numerical results}
\label{ssec:non_gauss_crem_simulations}

We simulate data from Model \ref{ex:ngcrem} with $K=2$ and $K=3$ factors, with Laplace response and $\tau_k=1$, $k=\{1,2,3\}$. For completeness, we consider both balanced ($I_1=I_2$, for $K=2$) and unbalanced designs ($I_1<I_2<I_3$, for $K=3$).
We consider Regime \ref{reg2} of Section \ref{ssec:numerics_gcrem}. We implement the Metropolis-within-Gibbs (MwG) scheme of Algorithm \ref{alg:non_gauss} with Random Walk Metropolis (RWM) updates and unknown $\boldsymbol{\tau}$. As for the coupling strategy, we consider Gaussian proposals  $Q(\x, d\x')$ coupled via maximal independent coupling, and paired acceptance.
We report in Figure \ref{fig:Lapl} the resulting average of the meeting times obtained for different numbers of RWM steps within each iteration ($S=5, 20$).

\begin{figure}[h!]
	\centering
	\begin{subfigure}{.5\textwidth}
		\centering
		\includegraphics[width=.95\linewidth]{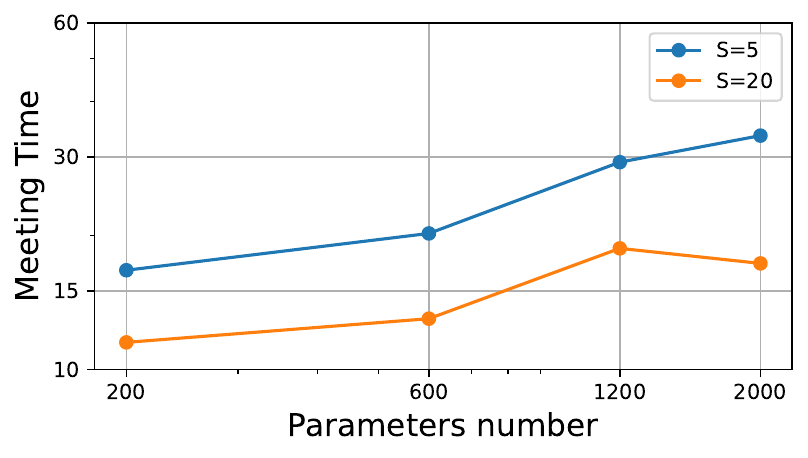}
	\end{subfigure}%
	\begin{subfigure}{.5\textwidth}
		\centering
		\includegraphics[width=.95\linewidth]{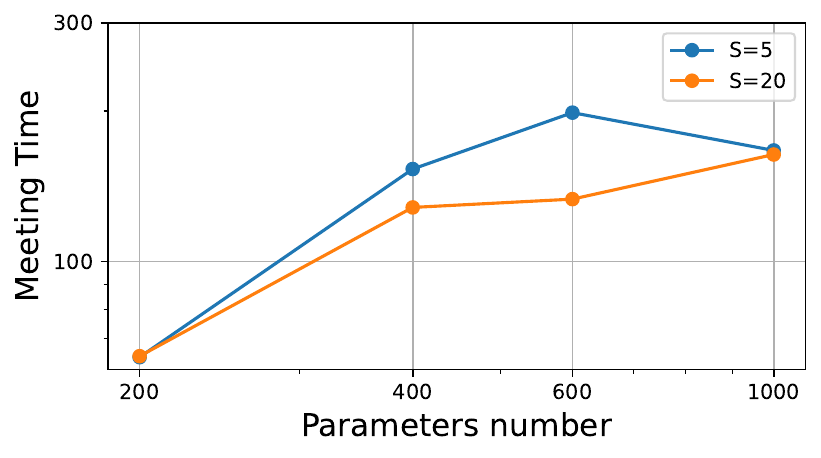}
	\end{subfigure}
	\caption{ \add{Model \ref{ex:ngcrem} with Laplace response, $K=2$ and $I_1 = I_2$ (left), and $K=3$ and $I_1<I_2<I_3$ (right). Average meeting times in log-log scale, for Algorithm \ref{alg:non_gauss} obtained with unknown $\boldsymbol{\tau}$ and for different number of Metropolis steps $S$.} }
	\label{fig:Lapl}
\end{figure}
\add{As $S$ grows, the conditional updates get closer to exact Gibbs updates and the resulting chain converges faster, leading to smaller meeting times. This comes at a higher computational cost per iteration. The optimal choice of $S$ depends on the application at hand.}

\add{To put these results in perspective, \cite{tim_scalable} show that MwG scheme without local centering and black-box samplers such as No-U-Turn-Sampler \citep{hoffman2014no} typically experience a degradation of performance as the dimension increases.
On the contrary, Figure \ref{fig:Lapl} shows that MwG sampler with local centering almost requires $\bigo{1}$ number of full likelihood evaluations per unbiased estimate, as $N$ increases.}

\add{Finally, note that for $K=2$, when one block coalesces, the other block coalesces automatically, and a meeting occurs. This is not true for $K>2$, which partly explains the higher meeting times observed for $K=3$.}

\section{Probabilistic matrix factorization}
\label{sec:pmf}

Finally, we consider Model \ref{ex:pmf}. 
A well-known feature of such model is that, for fixed $\textbf{u}$, the model reduces to a standard linear regression with coefficients $\textbf{v}$, and viceversa for fixed $\textbf{v}$. On the contrary, the likelihood is not analytically tractable if $\textbf{u}$ and $\textbf{v}$ vary jointly.
As mentioned in Section \ref{sec:mot}, this structure naturally lends itself to conditional updating schemes, such as BGS for sampling or block coordinate ascent (usually called alternating least squares whenever applied to Model \ref{ex:pmf}) for optimization. However, the vanilla BGS scheme, which updates  $\textbf{u}, \textbf{v}, \rho, \boldsymbol{\tau}$ from their full conditionals at every iteration, often results in slow mixing of the chain, for reasons analogous to the ones discussed for Models \ref{ex:gcrem} and \ref{ex:ngcrem} in previous sections (see also numerics in Figure \ref{fig:pmf}). 
We thus propose a \enquote{local centering} version of BGS for Model \ref{ex:pmf}, where we reparametrise the random effects using $\rho$, i.e.\ at each iteration we update $(\rho, \bar{\textbf{u}})$, with $\bar{\textbf{u}}= \rho \textbf{u}$, and then $(\rho, \bar{\textbf{v}})$, for $\bar{\textbf{v}}= \rho \textbf{v}$ analogously. Algorithm \ref{alg:pmf} provides high-level pseudo-code for one iteration of the resulting scheme, full implementation details can be found in the supplementary material.
\begin{algorithm}
		\For{ $(\br,\s) \in \{(\textbf{u}, \textbf{v}), (\textbf{v}, \textbf{u}) \}$}
		{
			reparametrize $ (\rho, \textbf{r}) \rightarrow (\rho, \bar{\br})$, with $\bar{\br}:= \rho \br$\\
		draw  
		$ \rho \sim \mathcal{L}(\rho | \bar{\br}, \s, \tau_0, \textbf{y})$ \\
		draw $ \bar{\br} \sim \mathcal{L}(\bar{\br} | \rho, \s, \tau_0, \textbf{y})$ \\
		recover $  \br = \bar{\br} / \rho$ \\
	}
	draw $\tau_0 \sim \mathcal{L}( \tau_0 | \rho, \textbf{u},\textbf{v}, \textbf{y})$\\
	\caption{One iteration of BGS with local centering for Model \ref{ex:pmf}}
	\label{alg:pmf}
\end{algorithm} 

Regarding the UMCMC version of Algorithm \ref{alg:pmf}, since the high-dimensional full conditionals involved are multivariate Gaussian, we can implement the same coupling strategy as for Model \ref{ex:gcrem}. In particular, joint maximal couplings for the high-dimensional updates of $\bar{\br}$ can be implemented efficiently.
Below we provide numerical illustrations of the performances of the UMCMC version of Algorithm \ref{alg:pmf} and vanilla BGS. As discussed in more details in Remark \ref{rem:rotation}, we restrict ourselves  to the case  $d=1$.

\begin{rem}[Related literature on Bayesian factor models]
Model \ref{ex:pmf} is closely related to Bayesian factor analysis. With the same notation as in \eqref{eq:pmf}, a factor model \citep{FA2014}  can be written as
\begin{equation}
	\label{eq:fa}
	y_n | \mu, \textbf{F},\Lambda, \boldsymbol{\tau} \sim N({\mu}_{i[n]} + \Lambda_{j[n],:}\textbf{F}_{i[n]},\tau_0^{-1}),
\end{equation}
for $\boldsymbol{\mu}  \in \mathbb{R}^{I_1}$, $\textbf{F}=(\textbf{F}_i)_{i=1}^{I_1}$ being the collection of unknown factors, $\textbf{F}_i \in \mathbb{R}^d$, and $\Lambda \in \mathbb{R}^{I_2\times d}$ the factor loading matrix, with $d$ being the latent dimension. Indeed, factor models exhibit the same structure of Model \ref{ex:pmf}, and BGS schemes are also widely used in that context \citep{CONTI201431, Papastamoulis2022}, even if there the focus is usually on the full design case (where all the combinations users/films are observed) and on regimes where $I_2 \ll I_1$ and $I_1$ grows.
\end{rem}

\begin{rem}[Rotational invariance]\label{rem:rotation}
A well-known issue of Model \ref{ex:pmf} is the invariance with respect to joint rotations of $\textbf{u}$ and $\textbf{v}$.
This creates multimodality in the posterior, thus inducing slow convergence and lack of posterior interpretability.
Many ad hoc methodologies have been developed to deal with such issue, including constraining a priori the matrix of factor loadings or post-processing \citep{CONTI201431, Papastamoulis2022}. Although of interest, these issues and techniques are somehow orthogonal to our focus here, and thus we restrict to the case $d=1$ and leave further exploration 
to future work.
\end{rem}

\subsection{Numerical results}
We simulate data coming from Model \ref{ex:pmf} for different asymptotic regimes and parameter specifications. 
We consider $I_1= I_2=I \in\{100,200,500,1000 \}$ levels and data coming from Regime \ref{reg1} and \ref{reg2}. In Figure \ref{fig:pmf} we report the average meeting time for both the vanilla BGS and Algorithm \ref{alg:pmf} (Local Centering) as discussed above.
\begin{figure}[h!]
	\centering
	\begin{subfigure}{.5\textwidth}
		\centering
		\includegraphics[width=.9\linewidth]{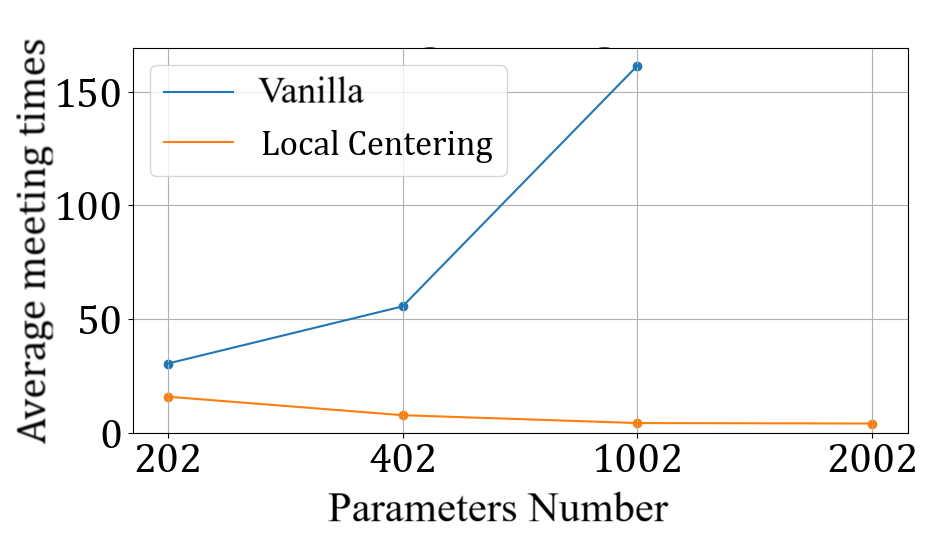}
	\end{subfigure}%
	\begin{subfigure}{.5\textwidth}
		\centering
		\includegraphics[width=.9\linewidth]{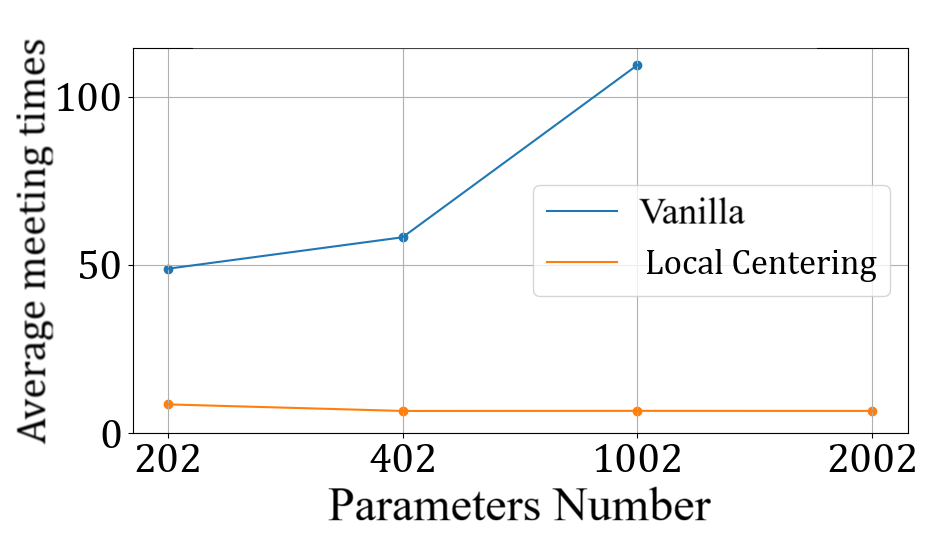}
	\end{subfigure}
	\caption{\add{Average meeting times for Probabilistic Matrix Factorization model. $I_1= I_2=I \in\{100,200,500,1000 \}$. Left: Regime \ref{reg1}, right: Regime \ref{reg2}.}}
	\label{fig:pmf}
\end{figure}
As expected, the slow mixing of vanilla BGS results in exploding meeting times of the coupled chains, which are not reported for $I$ greater than $500$ for visual convenience.
Remarkably, as few as $10$ iterations are on average sufficient for the coupled chains evolving according to Algorithm \ref{alg:pmf} to meet even in the high-dimensional cases.




\section*{Funding}
GZ acknowledges support from the European Research Council (ERC), through StG \enquote{PrSc-HDBayLe} grant ID 101076564.

\bibliography{bibliography}

\newpage
\appendix
\section{Background on couplings}
\label{s:sec:background}
\subsection{Maximal Coupling Algorithms}
\label{ssec:max_coup_alg}
We briefly review algorithms for sampling from maximal couplings of two distributions $p,q \in \mathcal{P}(\sX)$, with $\sX= \mathbb{R}^d$.  Provided $p$ and $q$ admit densities, there always exists an algorithm with maximal meeting probability \citep[Sec. 4.5 of Chap 1]{thorisson}, referred to as \emph{maximal rejection} (Algorithm \ref{alg:max_reg}).
However, in the case of spherically symmetric distributions (e.g.\ multivariate Gaussian with the same covariance matrix), an alternative approach called \emph{maximal reflection coupling} (see Algorithm \ref{alg:max_refl} or \cite{max_refle1,max_refle2}) allows for sampling with a deterministic cost and yet allowing for a good contraction of the square distance between the resulting draws (see e.g.\ Lemma \ref{lem:bound_gen}).

\paragraph{Maximal rejection coupling.} Algorithm \ref{alg:max_reg} reports the pseudo code for implementing a maximal rejection  coupling of $p,q$. 
\begin{algorithm}
 $\text{sample }\mathbf{X} \sim p$\\
 $\text{sample }W \sim U(0,1)$\\
 \If{$W p(\X) \le q(\mathbf{X})$}
 {$\text{set }  \mathbf{Y=X}$}
 \Else{
 sample $  \mathbf{Y}^* \sim q$ and
 $W^* \sim U(0,1)$\\
 \While{$W^* q(\mathbf{Y}^*) < p(\mathbf{Y}^*)$}
 {sample $\mathbf{Y}^* \sim q$ and
 $W^* \sim U(0,1)$}
  set $\mathbf{Y=Y^*}$
 }
 \caption{Maximal rejection coupling of $p,q \in \mathcal{P}(\mathbb{R}^d)$}
\label{alg:max_reg}
\end{algorithm}
The computational cost of Algorithm \ref{alg:max_reg} depends on the number of iterations required to accept the proposed sample. Only one sample is required if $\X$ is accepted as value for $\Y$, and this happens with probability $\Pr(p(\X) \,W \le q(\X))= 1-\| p-q\|_{TV}$. If instead $\X$ is rejected, the number of trials before acceptance follows a Geometric variable with parameter $\Pr(q(\Y^*) \, W^* > p(\Y^*))$, the latter being equal to $\| p-q \|_{TV}$. The resulting expected number of iterations is
$$
1-\| p-q \|_{TV} + \| p-q \|_{TV} (1/ \| p-q \|_{TV}+1) = 2.
$$ 
The variance of the expected number of iterations is equal to $\frac{2(1 - \|p - q\|_{TV})}{\|p - q\|_{TV}}$,
which tends to infinity as \(\|p - q\|_{TV}\) approaches zero. Gerber and Lee proposed accepting the first draw with a lower probability, thereby avoiding the problem of infinite variance, but at the cost of losing maximality (see the discussion in \cite{jacob2019unbiased}).

\paragraph{Maximal reflection coupling.}  
Algorithm \ref{alg:max_refl} reports an implementation of maximal reflection coupling for Gaussian distributions with same covariance matrix. 
\begin{algorithm}
 \add{let $A$ s.t.\ $\Sigma = A A^T$}, set $\mathbf{z}= A^{-1} (\boldsymbol{\xi}-\boldsymbol{\nu}), \, \mathbf{e}=\mathbf{z}/||\mathbf{z}||$ \\
 sample $\dot{\mathbf{X}} \sim N_d(\mathbf{0}, I_d), \, W \sim U(0,1)$\\
 \If{ $ W \le \exp\{-\frac{1}{2} \mathbf{z}^{\top}(2 \dot{\mathbf{X}} + \mathbf{z}) \}$} 
 { set $\dot{\Y} = \mathbf{\dot{X}}+\mathbf{z}$}
 \Else{
 $\mathbf{\dot{Y} = \dot{X}-2(e^\top \dot{X}) e}$ }
 $\mathbf{X}= A \mathbf{\dot{X}}+ \boldsymbol{\xi}$\\
 $\mathbf{Y}= A \mathbf{\dot{Y}}+ \boldsymbol{\nu}$
 \caption{Maximal reflection coupling of $N(\boldsymbol{\xi}, \Sigma)$ and $N(\boldsymbol{\nu}, \Sigma)$ }
\label{alg:max_refl}
\end{algorithm}
Note that, differently from Algorithm \ref{alg:max_reg}, in Algorithm \ref{alg:max_refl} no rejection mechanism is required and the algorithm's runtime is deterministic. Furthermore, in high-dimensional cases, this procedure shows favorable behaviors: in addition to being maximal, the algorithm contracts the distance between chains at a good rate. Thus, when applicable, it is generally preferred to Algorithm \ref{alg:max_reg}.
\subsection{$W_2$-optimal maps}
\label{ssec:w2opt}

For all univariate distributions, it is known that the \emph{monotone map} coupling (i.e.\ using the same random number for the inverse cdf method)
is optimal for every cost function of the form $c(x,y) = h(x-y)$ with $h$ convex \citep[Thm.2.9]{santambrogio}, and hence is $W_2$-optimal.
For $p,q$ univariate distributions, let $F_p(\cdot)$ and $F_q(\cdot)$ denote their cumulative density function (cdf). We define the inverse cdf as $$F^{-1}_p(u) := \inf_{t \in \mathbb{R}} \{t: F_p(t) \ge u \},$$ and $F^{-1}_q$ accordingly. It is then possible to sample from the $W_2$ optimal coupling as in Algorithm \ref{alg:it}.
\begin{algorithm}
	sample $U \sim U(0,1)$\\
	set $X = F_p^{-1}(U)$\\
	set $Y = F_q^{-1}(U)$\\
	\caption{monotone transport map for univariate distributions}
	\label{alg:it}
\end{algorithm}

No universal optimality result exists for general multivariate distributions, but a natural extension of the monotone map above, called \emph{common random number} (\emph{crn}) coupling, is indeed optimal for multivariate Gaussians whose covariance matrices commute \citep{gaussw2_1,gaussw2_4}, as stated below. 
\begin{lem}[Optimality of \emph{crn} coupling for Gaussian distributions]
	\label{lem:opt}
	Let $p=N(\boldsymbol{\xi}, \Sigma_1)$ and $q=N(\boldsymbol{\nu}, \Sigma_2)$ be $d$-dimensional Gaussian, with $ \Sigma_1 \Sigma_2 = \Sigma_2 \Sigma_1$. Define
	\begin{equation}
		\label{eq:general_pc}
		\Gamma^*:= N\left(
		\left(\begin{array}{c} \boldsymbol{\xi}  \\ \boldsymbol{\nu} \end{array} \right), \left(\begin{array}{cc} \Sigma_1 & F G^\top \\ G F^\top & \Sigma_2 \\ \end{array} \right)  
		\right),
	\end{equation}
	\add{where $F $ and $G$ are the principal square roots of $\Sigma_1$ and $\Sigma_2$, respectively.} Then $\Gamma^* \in \Gamma_{W_2}(p,q)$, i.e. $\Gamma^*$ is the $W_2$-optimal coupling of $p$ and $q$. 
\end{lem}

Thus, in order to obtain draws from $\Gamma^*$ in Lemma \ref{lem:opt}, one can simply sample $\textbf{Z} \sim N(\textbf{0}_{d}, {1}_{d})$ and then set
\begin{equation}
	\label{eq:general_pc_sample}
	\begin{cases}
		\X= \boldsymbol{\mu}+ F \textbf{Z},\\
		\Y = \boldsymbol{\nu} + G \textbf{Z}.
	\end{cases}
\end{equation}
We will refer to \eqref{eq:general_pc_sample} as \emph{crn} coupling.
Recall also that the $W_2$-optimal map is unique \citep{brenier}.
\add{Note that when $\Sigma_1 = \Sigma_2$, any square root of $\Sigma_1$ would give an optimal coupling.}

\section{Additional results and simulations}
\label{s:sec:add_res}

\subsection{Bound for non-reversible chains}
\label{ssec:bound_pi_inv}
For BGS with general updating order it is possible to adapt the result of Theorem \ref{thm:bound_expected_rev}.
\begin{thm}
	\label{thm:bound_expected}
	\add{Let $\pi=N(\boldsymbol{\mu},\Sigma)$ and $(\X^{t}, \Y^{t})_{t\ge 0}$ be a Markov chain marginally evolving with $\pi$-invariant BGS kernel and coupled via Algorithm \ref{alg:2s}.
	For any $\varepsilon$ satisfying \eqref{eq:eps_cond_block} and for any $\delta>0$, it holds}
	\begin{equation} 
		\label{eq:bound_thm}
		\add{T \leq 1+\tilde{f}_1(\|\X^0-\Y^0\|, \varepsilon, B, Q,\delta)  +  2\tilde{f}_2(\varepsilon,B, Q, \delta)\,,}
	\end{equation}
	where
	\[
		\add{\tilde{f}_1(r, \varepsilon, B,\delta)
		=\max \left( n^*_\delta, \left\lceil\frac{\log \|\X^{0}-\Y^{0}\|+ \log \sqrt{\kappa (Q)}-\log \varepsilon}{-\log \left( 1-\frac{1-\rho(B)}{1+\delta} \right)}\right\rceil\right),}
	\]
	\[
		\add{\tilde{f}_2(\varepsilon,B, \delta) = \max \left( n^*_\delta, \left\lceil \frac{ 2+  \log \kappa(Q) - \frac{1}{2}\log \rho(\Delta^{-1}) + \frac{1}{2}  \log K  -\frac{1}{2}\log \varepsilon }{-\log \left( 1-\frac{1-\rho(B)}{1+\delta} \right)} \right\rceil \right)\,,}
	\]
	and
 	\[n^*_\delta := \inf\left\{ n_0 \ge 1: \forall n \ge n_0 \; \; 1-\| N^n\|_2^\frac{1}{n} \ge \frac{1-\rho(N)}{1+\delta} \right\}\,.\]
	where $N=L^{-1}BL$ with $L$ such that $\Sigma = L L^\top$.

\end{thm}
The bound in \eqref{eq:bound_thm} features the additional term $n^*_\delta$ compared to the expression in Theorem \ref{thm:bound_expected_rev}.
The reason for it is that, due to the non-reversibility of $P$, $N$ is generally not symmetric. Thus $\|N^{n}\|^\frac{1}{n}\to\rho(N)$ from above \citep{Gel41}, but in general $\|N^{n}\|^\frac{1}{n}\ne \rho(N)$ for finite $n$. 
In order to make the result in \eqref{eq:bound_thm} fully informative, like the ones in the reversible and two-block cases, one would need to provide an explicit bound on $n^*_\delta$ for the given matrix $N$ under consideration.
In all our numerical experiments, we observed $n^*_\delta$ to be smaller than the second term and never the leading term of the bound, and we expect it to be well-behaved in our contexts of interest.
On the other hand, we are not aware of general tight bounds for $n^*_\delta$ and we thus left it as an explicit term in the bound.

\subsection{Additional simulations for Model \ref{ex:gcrem}}
Consider simulated data coming from Model \ref{ex:gcrem}, with the same specification as in Section \ref{ssec:numerics_gcrem} and Section \ref{sssec:sim_data}. We plot the average of the meeting times as a function of the total number of parameters of the model, i.e.\ $1+KI$ plus the number of scale parameters if any. 
Figure \ref{fig:three} reports results for the collapsed Gibbs sampler (left) and vanilla (right), for $K=2$, $I_1= I_2=I \in\{50, 100,250,500,1000 \}$ levels, Regime \ref{reg2}, fixed and free variances.
For each scheme, the corresponding bound of Theorem \ref{thm:bound_expected_rev} is also reported, using the true data generating values for the variance parameters.
\begin{figure}[h!]
	\centering
	\begin{subfigure}{.5\textwidth}
		\centering
		\includegraphics[width=.9\linewidth]{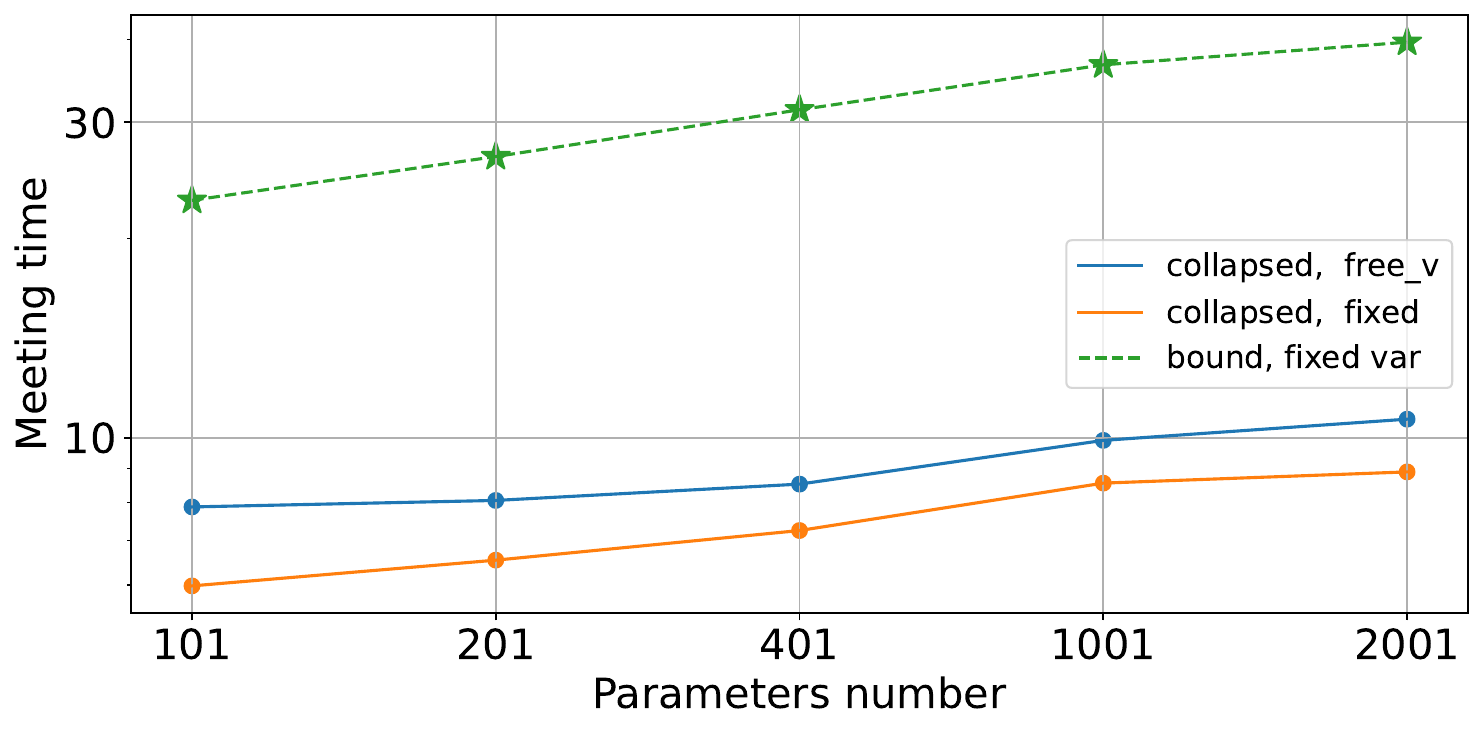}
	\end{subfigure}%
	\begin{subfigure}{.5\textwidth}
		\centering
		\includegraphics[width=.9\linewidth]{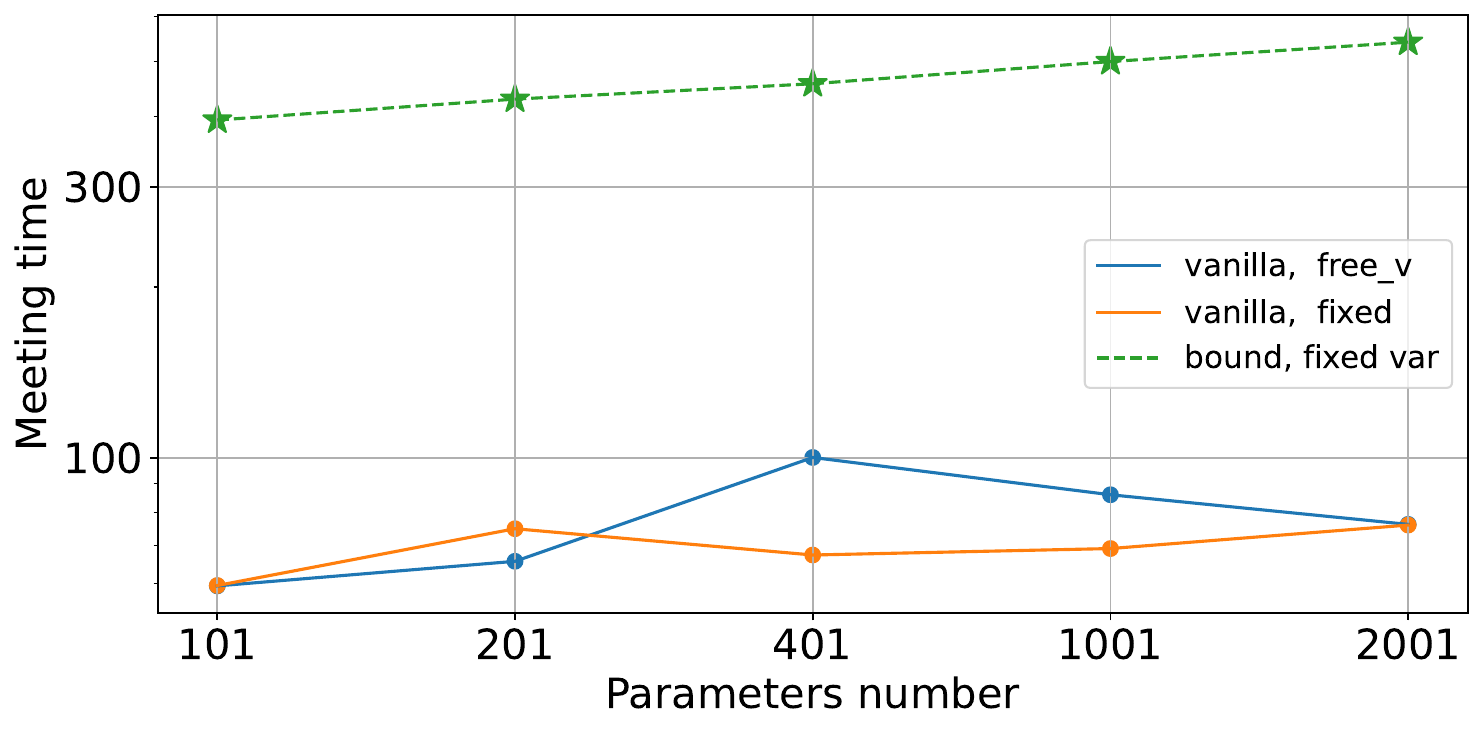}
	\end{subfigure}
	\caption{\add{Estimated meeting times and bounds for $K=2$, $I_1 = I_2$, $\tau_0=\tau_1 = \tau_2 =1$, Regime \ref{reg2}. Left: Algorithm \ref{alg:cg}, right: vanilla algorithm.}}
	\label{fig:three}
\end{figure}
\subsection{One-step vs two-step couplings for Model \ref{ex:gcrem}}
Consider the same specification as in Section  \ref{ssec:numerics_gcrem} and Section \ref{sssec:sim_data}.
In Figure \ref{fig:2s} we report the estimated distribution of meeting times for the collapsed Gibbs scheme, when traditional or \textit{two-step} coupling is implemented on a synthetic dataset with $K=3$ factors and Regime \ref{reg2}.  
\begin{figure}[h!]
	\centering
	\includegraphics[width=.49\linewidth]{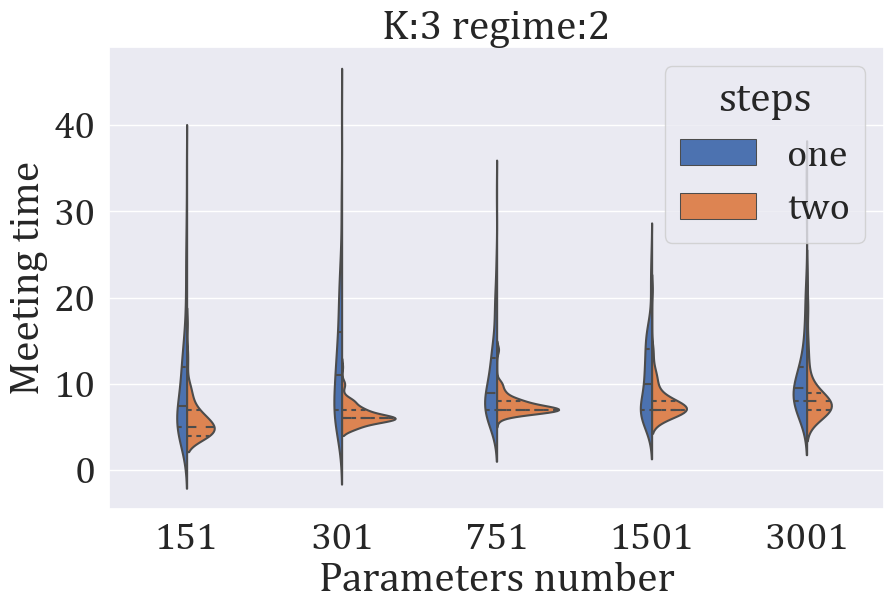}
	\caption{ Estimated distribution of meeting times for $K=3$, $I_1 = ... = I_3$, $\tau_k = 1$ for $k\in\{0,1,2,3\}$ one vs two-step. Regime \ref{reg2}, Algorithm \ref{alg:cg}.}
	\label{fig:2s}
\end{figure}
As can be seen from the above, the distribution of meeting times for the \emph{two-step} strategy is more concentrated on smaller values with considerably lighter tails, thus supporting the choice of a two-step coupling.
On the other hand, we see that, in our context, using one- versus two- step coupling is less influential than, for example, in the one of \citet[Fig.3]{biswas2019estimating}.

\subsection{Couplings of conditionally independent blocks}
\label{ssec:coupl_con_ind}
We now discuss how to implement 
$\bar{P}_{max}[P_k]$ and $\bar{P}_{W_2}[P_k]$, for $P_k$ as in \eqref{eq:kernel_k_update},
 in cases where the associated full conditional factorizes as 
\begin{align}\label{eq:cond_indep}
\pi \left(\x_{k}| \x_{-k} \right) =  \otimes_{i=1}^{I_K} \pi \left(x_{k,i}| \x_{-k} \right),
\end{align}
for $x_{k,i}$ denoting the $i$-th component of the vector $\x_{k}$, i.e. $\x_{k} = (x_{k,1},..., x_{k,I_k})$, and $I_k$ is large.
By \eqref{eq:cond_indep}, independently sampling from the univariate distributions $\pi \left(x_{k,i}| \x_{-k} \right)$ is equivalent to sampling directly from the entire block $\pi \left(\x_{k}| \x_{-k} \right)$. 
The same intuition extends to $W_2$ optimal couplings but not to maximal ones.
In particular, one has that the product of independent $W_2$-optimal couplings of $\pi \left(x_{k,i}| \x_{-k} \right)$ for $i=1,\dots,I_k$ is $W_2$-optimal for $\pi \left(\x_{k}| \x_{-k} \right)$ while the same is not true for maximal ones.
In particular, when $p$ and $q$ are two product measures, one has the following well-known facts, which we collect in a lemma whose proof we omit for brevity.
\begin{lem} 
	\label{lem:prod_tv}
	Let $p,q\in \sP(\sX_1\times\cdots\times\sX_d)$ with 
$p = \bigotimes_{i=1}^d p_i$ and $q = \bigotimes_{i=1}^d q_i$. 
Then 
$\mu_i\in\Gamma_{W_2}(p_i,q_i)$ for all $i=1,\dots,d$ implies $(\otimes_{i=1}^d\mu_i)\in\Gamma_{W_2}(p,q)$.
On the contrary, $\mu_i\in\Gamma_{max}(p_i,q_i)$ for all $i=1,\dots,d$ does not imply $(\otimes_{i=1}^d\mu_i)\in\Gamma_{max}(p,q)$ in general. In particular, one has
	\begin{equation}
		\label{eq:prod_tv}
		\min_{i=1,\dots,d} {\textstyle \Pr_{max}(p_i,q_i)} \ge {\textstyle \Pr_{max}}(p, q) \ge \prod_{i=1}^{d} {\textstyle \Pr_{max}}(p_i,q_i)\,,
	\end{equation}
where we use the notation $\Pr_{max}(p,q):=1-\|p-q\|_{TV}$, and all the inequalities can be strict.
\end{lem}
Lemma \ref{lem:prod_tv} implies that, under \eqref{eq:cond_indep}, we can simply take a contractive coupling $\bar{P}_{W_2}[P_k]$ which factorizes across coordinates. 
On the contrary, joint maximal couplings of $\pi \left(\x_{k}| \x_{-k} \right)$ do not factorize across coordinates.
The lower bound $\Pr_{max}(p,q)\geq \prod_{i=1}^{d}\Pr_{max}(p_i,q_i)$ in \eqref{eq:prod_tv} implies that setting $\|p_i,q_i\|_{TV}=O(d^{-1})$ ensures $\Pr_{max}(p,q)$ is bounded away from $0$.

Lemma \ref{lem:asymp} below quantifies the tightness of such lower bound for $d$-dimensional Gaussian distributions with the same variance-covariance matrix. 
We consider the regime where $d$ goes to infinity and the distance between each rescaled mean decreases with $d$, which is arguably descriptive of what happens when using the two-step algorithm (Algorithm \ref{alg:2s}) of Section \ref{ssec:2step} in high dimensions. 
\begin{lem}
	\label{lem:asymp}
	Consider $p= N(\boldsymbol{\mu}, diag(\boldsymbol{\sigma}))$ and $ q =  N(\boldsymbol{\nu}, diag(\boldsymbol{\sigma}))$, $d$-dimensional Gaussian distribution with $\boldsymbol{\sigma}=(\sigma_1^2, ..., \sigma_d^2)$ such that $\frac{{\mu_i}-{\nu_i}}{\sigma_{i}} = c_i d^{-\alpha}, i=1,...,d$, with  $0<\inf_i |c_i|\leq \sup_i |c_i|< + \infty$  and $\alpha > 0$, then:
	\begin{align*}
		&{\textstyle \Pr_{max}}(p, q) \asymp 
		\begin{cases}
			\frac{d^{\alpha-\frac{1}{2}}}{\sqrt{\pi} \bar{c}_d} 
			\exp\left(-\frac{\bar{c}_d^2}{\sqrt{2}} d^{-2\alpha+1} \right)
			&  \text{for } 0<\alpha \le \frac{1}{2}  \\
			1-  \frac{\bar{2 c_d}}{\sqrt{\pi}} d^{-\alpha +\frac{1}{2}} &\text{for } \alpha > \frac{1}{2}
		\end{cases}
		&\hbox{ as }d\to\infty,\\
		&\prod_{i=1}^{d}{\textstyle \Pr_{max}}(p_i, q_i)  \asymp   
		\exp(-d^{1-\alpha} \tilde{c}_d)  
		&\hbox{ as }d\to\infty,
	\end{align*}
	where  $\bar{c}_d:=\sqrt{\frac{\sum_{i=1}^d c_i^2}{8d}}$, $\tilde{c}_d = \frac{\sum_{i=1}^d |c_i|}{d \sqrt{2\pi}}$ and we write $f(x) \asymp g(x)$ whenever $\lim_{x \rightarrow + \infty} \frac{f(x)}{g(x)} = 1$.
\end{lem}
It follows that both probabilities go to zero for $0<\alpha <0.5$ as $d \rightarrow + \infty$, while for $0.5<\alpha <1$ only $\Pr_{max}(p,q)$ goes to 1. For $\alpha >1$, both probabilities converge to 1 (although with different regimes).
In Figure \ref{fig:d_tv_gauss} we report the ratio of the blocked and the component-wise meeting probabilities, i.e.\  $\Pr_{max}(p,q)/(\prod_{i=1}^{d}\Pr_{max}(p_i,q_i))$, for a $d$-dimensional Gaussian distribution with independent components, where $c_i=1$ for all $i$ and different values of $\alpha$, along with a dotted line representing the value 1.
\begin{figure}[h!]
	\centering
	\includegraphics[width=.65\linewidth]{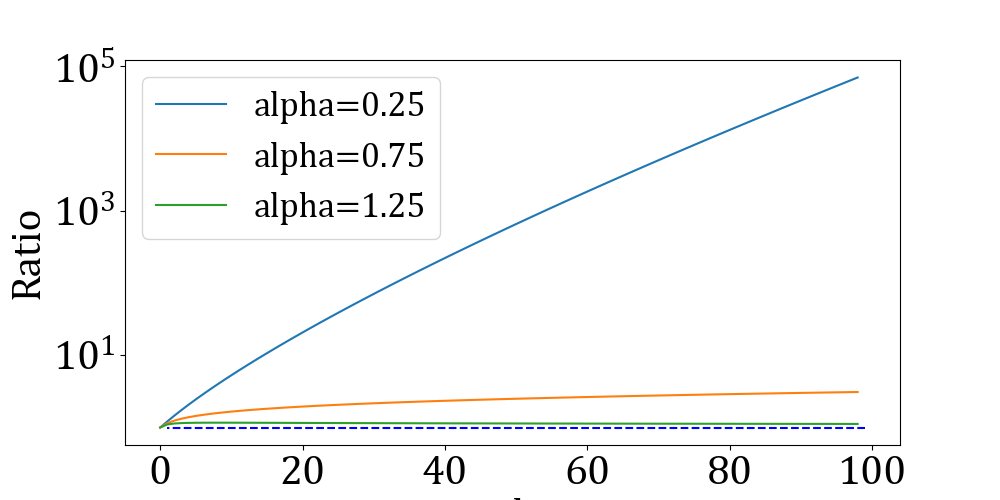}
	\caption{Ratio of blocked and component-wise meeting probability for $d$-dimensional Gaussian, different $\alpha$ values. Dimension on $x$-axis, logarithmic scale on $y$-axis.}
	\label{fig:d_tv_gauss}
\end{figure}
Figure \ref{fig:d_tv_gauss} shows that for $\alpha >1$, the blocked maximal coupling has meeting probabilities comparable to that of the independent counterpart.

\subsection{\add{Additional simulations for Model \ref{ex:ngcrem}}}

We study the performance of the No-U-Turn-Sampler (NUTS) \citep{hoffman2014no} applied to Model \ref{ex:ngcrem}. Note that this approach does not specifically use the structure of Model \ref{ex:ngcrem} and thus might be expected to be sub-optimal for that reason.

\begin{figure}[h!]
	\centering
		\includegraphics[width=.5\linewidth]{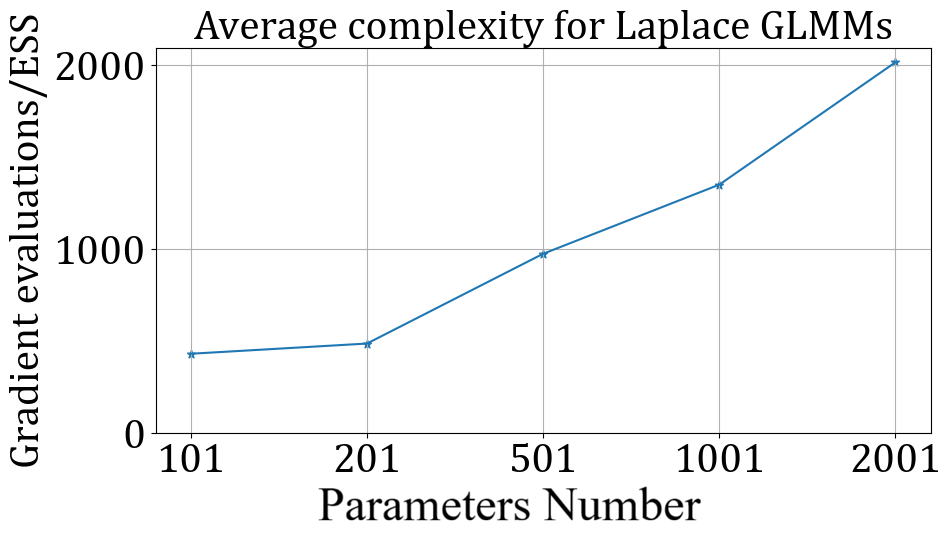}
	\caption{ Model \ref{ex:ngcrem} with Laplace response, $K=2$, $I_1= I_2$, $\tau_1 = \tau_2 = 1$, $b=1$.}
	\label{fig:stan}
\end{figure}

We illustrate in Figure \ref{fig:stan} the convergence speed of the STAN implementation \citep{stan} of NUTS with the default setting, for estimating Model \ref{ex:ngcrem} with different number of parameters.
Specifically, we report the average number of gradient evaluations per Effective Sample Size (ESS), considering the minimum ESS across parameters.

\section{Couplings for Metropolis-Hastings algorithms for product targets}
\label{ssec:MH}
\label{annex:mh}

In this section we discuss procedures for efficient coupling of  Metropolis-Hastings (MH) kernels for targets with independent components. The motivation for such a construction stems from Algorithm \ref{alg:non_gauss} applied to Model \ref{ex:ngcrem}, where each iteration consists of updating $K$ blocks of (conditionally) independent coordinates, since each $\mathcal{L}(\boldsymbol{\xi}_{k}| \mu, \mathbf{a}_{-k}, \boldsymbol{\tau},\mathbf{y})$ factorizes in $\prod_{i=1}^{I_k} \mathcal{L}({\xi}_{k,i} |\mu, \mathbf{a}_{-k}, \boldsymbol{\tau},\mathbf{y})$ for $k=1,...,K$, whose distribution might be known only up to constants.
Exploiting such independence in the coupling construction, one can derive coupling strategies whose meeting times grows logarithmically with $I_k$, see below.
Previous works on couplings for MH
kernels include \cite{OlearyWang}, 
where among other things the authors suggest 
employing a maximal reflection coupling on Gaussian proposals and paired acceptance, and \cite{gcrn}, where the authors focus on asymptotically optimally contractive couplings.

Consider a target distribution $\nu$ on $\sX = \mathbb{R}^{I_k}$ with independent components, i.e. \ $\nu=\otimes_{i=1}^{I_k} \nu_i$. The general Metropolis kernel targeting $\nu$ has the form
\begin{equation}
	\label{eq:mhk_b}
	P_b^{MH}(\x, d\x') = Q_b(\x, d\x') a_b(\x,\x') + \delta_{\x}(d\x') r_b(\x),
\end{equation}
where $Q_b(\x, d\x')$ denotes the proposal distribution on $\sX = \mathbb{R}^{I_k}$,
$a_b(\x,\x')$ is the Metropolis acceptance ratio, i.e.\ $  a_b(\x,\x') = 1 \land \frac{\nu(\x')}{\nu(\x)} \frac{Q_b(\x',\x )}{Q_b(\x, \x')}$, and $r_b(\x) = 1- \int_{\mathcal{X}} Q_b(\x, d\x')  a_b(\x, d\x') $. 
The standard way to sample from $P_b^{MH}$ in \eqref{eq:mhk_b} is sampling a proposal $\x'$ from $Q_b(\x, \cdot)$, compute $a_b(\x, \x')$ and accept if $U \sim U(0,1)$ is smaller than the acceptance ratio. Given the known independence structure of the target, however, it is possible to propose and accept/reject each component individually, leading to much higher acceptance rates and better dimensionality scaling. The resulting kernel, which is a product of univariate MH kernels, can be written as 
\begin{equation}
\label{eq:mhk_f}
P_f^{MH}(\x, d\x') =  \otimes_{i=1}^{I_k} P_i^{MH}(x_i, d x'_i) = \otimes_{i=1}^{I_k} \left(   Q_f(x_i, dx'_i) a_f(x_i,x'_i) + \delta_{x_i}(dx'_i) r_f(x_i) \right),
\end{equation}
where $Q_f(x_i, dx'_i)$ is a proposal kernel on $\mathbb{R}$, $a_f(x_i,x'_i)= 1 \land \frac{\nu_i(x'_i)}{\nu_i(x)} \frac{Q_f(x'_i,x_i)}{Q_f(x_i, x'_i)}$, and analogously  $r_f(x) = 1- \int_{\mathbb{R}} Q_f(x_i, dx'_i)  a_f(x_i, dx'_i)$.
Note that \eqref{eq:mhk_b} proposes and accept jointly all the components at once, while  \eqref{eq:mhk_f} does it component-wise.
The coupling strategy can exploit such independence in two different ways, namely factorizing both the proposal and acceptance step or only the acceptance step. 

Differently from models with conjugate full-conditional distributions (such as e.g.\ Models \ref{ex:gcrem} and \ref{ex:pmf}), (optimally) contractive couplings for MH kernels are difficult to implement, requiring numerical integration, and in our simulations they did not provide significant enough decrease in distance within subsequent steps to justify their use.
Similarly, simple \emph{crn} couplings of the MH kernels, i.e.\ using same random number for the proposal distributions (amounting at implementing the $W_2$ optimal coupling on the proposals whenever Gaussians) and acceptance steps, were also not effective in contracting efficiently the chains (specifically they typically soon reach a plateau distance not small enough to provide high chances of coalescence). For the above reasons, when using MH steps to update from high-dimensional and conditionally independent blocks, we avoid the two step strategy of Algorithm \ref{alg:2s} and instead concentrate on one-step, maximal-only strategies.

We consider kernels with synchronous acceptance, i.e. using same uniform for accept-reject in the $\x$ and $\y$ chain. For $a_b, a_f, r_b$ and $r_f$ as in \eqref{eq:mhk_b} and \eqref{eq:mhk_f}, we define
\begin{align*}
    &\bar{a}_b = \left(\begin{array}{c}  a_b(\x,\x')\cdot \textbf{1}_{I_k}  \\  a_b(\y,\y') \cdot \textbf{1}_{I_k} \end{array} \right) \in \mathbb{R}^{2I_k}, \qquad &\bar{a}_f= \left(\begin{array}{c}  a_f(x_i,x'_i) \\  a_f(y_i, y'_i) \end{array} \right) \in \mathbb{R}^{2}, \\
    &\Delta_b= \left(\begin{array}{c}   \delta_{\x}(d\x') r_b(\x) \\   \delta_{\y}(d\y') r_b(\y) \end{array} \right) \in \mathbb{R}^{2I_k},  \qquad &\Delta_f= \left(\begin{array}{c}   \delta_{x_i}(dx'_i) r_f(x_i) \\   \delta_{y_i}(dy'_i) r_b(y_i) \end{array} \right) \in \mathbb{R}^{2},
\end{align*}
where $\textbf{1}_{I_k}$ denotes the vector of ones of length $I_k$.
Below we illustrate numerically the performance of the following list of possible coupled kernels:
\begin{enumerate} 
\item \emph{Blocked reflection}: $\bar{P}_{b,r} :=  \bar{P}_{max}[Q_b] \odot \bar{a}_b + \Delta_b  $, where $\bar{P}_{max}[Q_b]$ is Algorithm \ref{alg:max_refl} and $\odot$ denotes the Hadamard product, i.e.\ component-wise product. 
\item \emph{Blocked maximal}: $\bar{P}_{b,m}:=  \bar{P}_{max}[Q_b] \odot \bar{a}_b +\Delta_b  $, where $\bar{P}_{max}[Q_b]$ is Algorithm \ref{alg:max_reg}. 
\item  \emph{Blocked factorized reflection}: $\bar{P}_{bf,r}:=   \otimes_{i=1}^{I_k} \left( \bar{P}_{max}[Q_b]_{[i]} \odot \bar{a}_f + \Delta_f \right)  $, where, if $(\x, \y) \sim \bar{P}_{max}[Q_b]$, the symbol $\bar{P}_{max}[Q_b]_{[i]}$ indicates the vector $(x_i,y_i)$, and $\bar{P}_{max}[Q_b]$ is Algorithm \ref{alg:max_refl}. 
\item  \emph{Blocked factorized maximal}: $\bar{P}_{bf,m}:= \otimes_{i=1}^{I_k}  \left( \bar{P}_{max}[Q_b]_{[i]} \odot \bar{a}_f + \Delta_f \right)  $, where $\bar{P}_{max}[Q_b]$ is Algorithm \ref{alg:max_reg}.
\item  \emph{Fully factorized reflection}:  $\bar{P}_{ff,r}:= \otimes_{i=1}^{I_k} \left(  \bar{P}_{max}[Q_f] \odot \bar{a}_f + \Delta_f  \right) $, where $\bar{P}_{max}[Q_b]$ is Algorithm \ref{alg:max_refl}.
\item   \emph{Fully factorized maximal}:  $\bar{P}_{ff,m}:=\otimes_{i=1}^{I_k} \left(  \bar{P}_{max}[Q_f] \odot \bar{a}_f + \Delta_f  \right)  $, where $\bar{P}_{max}[Q_b]$ is Algorithm \ref{alg:max_reg}. 
\end{enumerate}
We report in Algorithm \ref{alg:mh_max} the pseudo-code for one iteration of either $\bar{P}_{bf,r}$, $\bar{P}_{bf,m}$, $\bar{P}_{ff,r}$ or $\bar{P}_{ff,m}$, depending on the specification of $\bar{P}[Q]$.
\begin{algorithm}[h!]
	\textbf{Input:} $(\X^t, \Y^t)$, target $\nu$, proposal $Q$, desired coupling $\bar{P}$\\
	sample $(\X', \Y') \sim \bar{P}[Q]((\X^t, \Y^t), \cdot) $
	\For{$i = 1,...,I_k$}{
		sample $U  \sim U(0,1) $
		\If{$U \le \frac{\nu(X'_i)}{\nu(X^t_i)} \frac{Q(X'_i,X^t_i )}{Q(X^t_i, X'_i)} $}
		{set $X^{t+1}_i= X'_i$}
		\Else {set $X^{t+1}_i= X^t_i$}
		\If{$U \le \frac{\nu(Y'_i)}{\nu(Y^t_i)} \frac{Q(Y'_i,\Y^t_i )}{Q(Y^t_i, Y'_i)} $}
		{set $Y^{t+1}_i= Y'_i$}
		\Else {set $Y^{t+1}_i= Y^t_i$}
	}
	$t = t+1$\\	
	\caption{Coupling strategy for MH with independent target}
	\label{alg:mh_max}
\end{algorithm}

We provide a numerical illustration where $\nu$ is taken to be a product of independent Laplace distributions, i.e.\ $\nu = \bigotimes_{i=1}^{d} Lapl(0,1/\sqrt{2})$. 
Consider $(\X^t, \Y^t)_{t \ge 0}$  coupled chains marginally evolving via the kernels 1 to 6 above with 
$Q_b(\x, \cdot)= N(\x, \sqrt{2}I_d)$ or $Q_f(x_i, dx_i) =  N(x_i, \sqrt{2})$, where step-sizes are chosen following the guidance in \cite{optimal_scaling} for univariate Metropolis steps.
We plot in Figure \ref{fig:meeting_times} the average meeting times for coupled chains with different strategies, as the target dimension $d$ grows.
\begin{figure}[h!]
	\centering
	\begin{subfigure}{.99\textwidth}
		\centering
		\includegraphics[width=.5\linewidth]{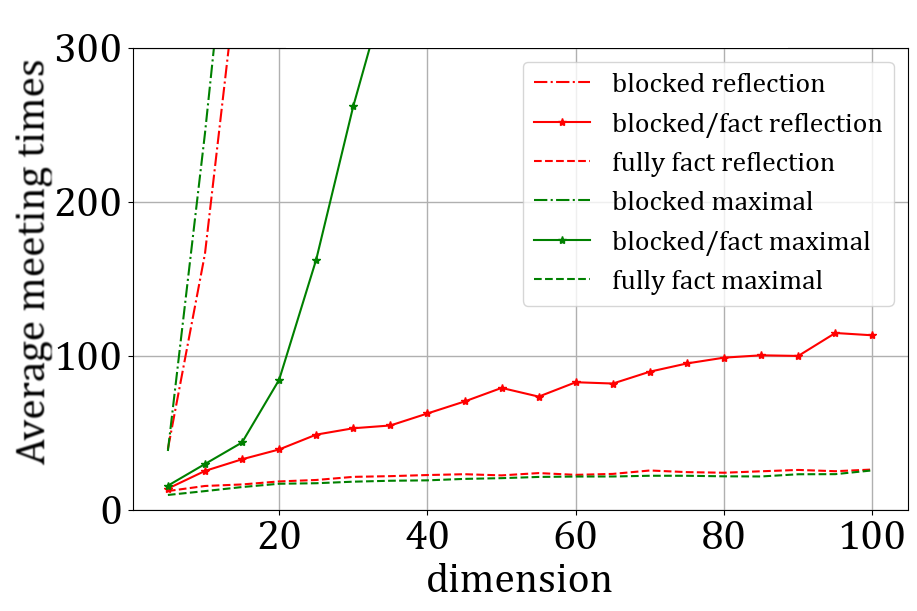}
	\end{subfigure}%
	\caption{Average meeting times for different dimensions, Laplace target.}
	\label{fig:meeting_times}
\end{figure}
As one might expect, Figure \ref{fig:meeting_times} shows that the strategies yielding smaller meeting times are those leveraging the independence structure of the target, proposing and accepting the components independently. \emph{Blocked} strategies instead  generally perform worse, with the sole exception of \emph{block/fact} reflection, due to the intrinsic contraction properties of Algorithm \ref{alg:max_refl}. 

To better illustrate the phenomenon, Figure \ref{fig:comp_prob} plots the proportion of components that did not meet for the same chains as in Figure \ref{fig:meeting_times}, for $d = \{3,100\}$.
\begin{figure}[h!]
	\centering
	\begin{subfigure}{.5\textwidth}
		\centering
		\includegraphics[width=.9\linewidth]{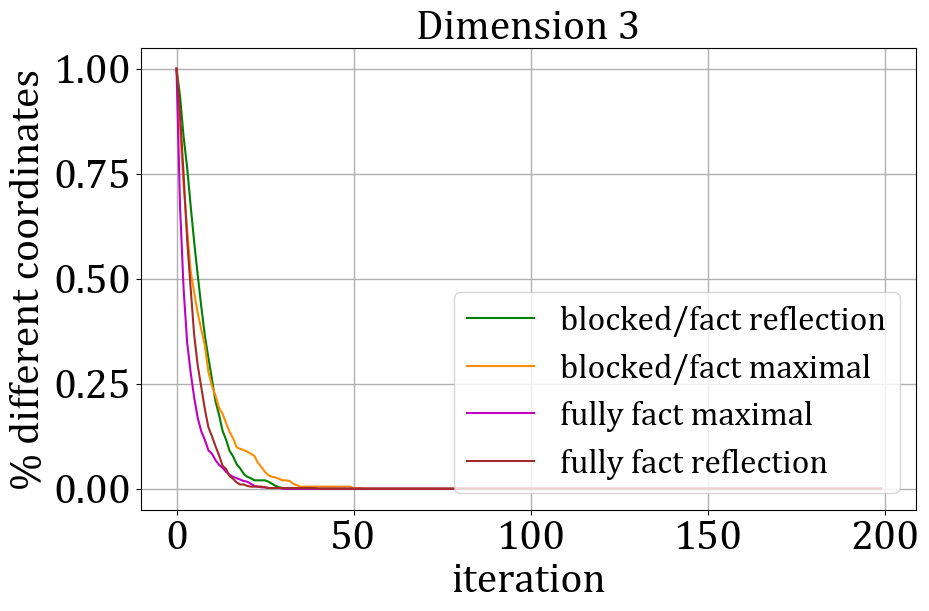}
	\end{subfigure}%
	\begin{subfigure}{.5\textwidth}
		\centering
		\includegraphics[width=.9\linewidth]{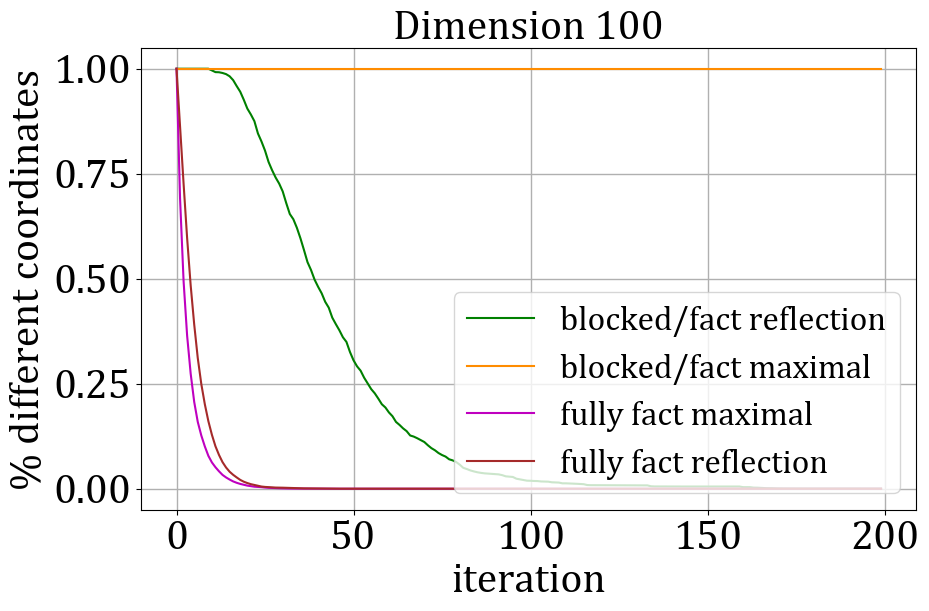}
	\end{subfigure}
	\caption{Estimated percentage of non coalesced components for blocked/component-wise proposals and component-wise acceptance; $d=3,100$.}
	\label{fig:comp_prob}
\end{figure}

The above examples suggest that, in case of conditionally independent blocks, fully-factorized couplings of the MH updates are preferable. 
This is consistent with the consideration that, 
in the fully-factorized case, 
the overall meeting time, which coincides with the one of the slowest component that coalesces, is simply the supremum of $d$ independent random variables, hence typically growing logarithmically as the dimensionality grows (or at least sub-linearly, see \cite{CORREA2021785}).

\section{Algorithmic implementation details}
\label{s:sec:algimpl}
In this section, we report the explicit expressions for the full-conditional distributions required to implement the proposed algorithms for Models \ref{ex:gcrem} and \ref{ex:pmf}.
\subsection{Full conditionals for Model \ref{ex:gcrem}}
\label{ssec:distributions}
Under \eqref{eq:gauss_crem}, the conditional distributions required to implement Algorithm \ref{alg:cg} are
\begin{align*}
\mathcal{L}(\mu | \mathbf{a}_{-k}, \boldsymbol{\tau}, \mathbf{y})&= N \left( \frac{1}{\sum_j s^{(k)}_j} \sum_j s^{(k)}_j \left( \tilde{y}_{k,j} - \frac{\sum_{l\ne k} \sum_i a_{l,i} n^{(k,l)}_{ji}}{n^{(k)}_j} \right), \frac{1}{\tau_k \sum_j s^{(k)}_j } \right),\\
\mathcal{L}(a_{k,i}|\textbf{a}_{-k},\mu, \boldsymbol{\tau}, \textbf{y})&= N\left( \frac{n^{(k)}_i \tau_0}{n^{(k)}_i \tau_0 + \tau_k} \left(\tilde{y}_{k,i} -\mu- \frac{ \sum_{l\ne k,0} \sum_{j=1}^{I_l} a_{l,j} n^{(k,l)}_{ij}}{n^{(k)}_i}\right), \frac{1}{n^{(k)}_i\tau_0 + \tau_k} \right), \\
\mathcal{L}(\tau_{k}|\textbf{a}, \mu, \boldsymbol{\tau}_{-k}, \textbf{y})&= Gamma \left( \frac{I_k-1}{2}, \frac{2}{\sum_{i=1}^{I_k} (a_{k,i})^2} \right),
\end{align*}
where 
$n_i^{(k)}=\sum_{n=1}^N\mathbb{I}(i_k[n]=i)$,  $s_j^{(k)}= n_j^{(k)}\tau_0 /(\tau_k + n_j^{(k)} \tau_0)$,    $n_{ji}^{(k,l)}=\sum_{n=1}^N\mathbb{I}(i_k[n]=j)\mathbb{I}(i_l[n]=i)$ denotes the number of observations of level $j$ of factor $k$ and $i$ of factor $l$ and finally $ \tilde{y}_{k,i}=\sum_{n: i_k[n]=i} y_n /|\{n: i_k[n]=i\} |$ is the average of all observations with level $i$ on factor $k$.
See also \citet[Eq. 4 and Prop. 2]{papaspiliopoulos2018scalable} for similar expressions.

\begin{algorithm}
    \textbf{Input:} $\X   = (\mu  , \mathbf{a} , \tau _0   , \dots ,\tau_K  )$, $\tilde{\X}  = (\tilde{\mu}  , \tilde{\mathbf{a}} , \tilde{\tau} _0   , \dots ,\tilde{\tau}_K  )$.\\
    \If{ $\|\X - \tilde{\X} \| > \varepsilon$}{
	\add{Sample $(\tau_0, \tilde{\tau}_0) \mid \textbf{a}, \mu, \boldsymbol{\tau}_{-0},\tilde{\textbf{a}}, \tilde{\mu}, \tilde{\boldsymbol{\tau}}_{-0}, \textbf{y}$, using CRN coupling in Algorithm \ref{alg:it}.\\}
	\For{k= 1,...,K}
	{
		\add{Sample $(\tau_k, \tilde{\tau}_k) \mid \textbf{a}, \mu, \boldsymbol{\tau}_{-k}, \tilde{\textbf{a}}, \tilde{\mu}, \tilde{\boldsymbol{\tau}}_{-k}, \textbf{y}$, using CRN coupling in Algorithm \ref{alg:it}.\\
		Sample $(\mu  , \tilde{\mu}  ) \mid \mathbf{a}_{-k} , \boldsymbol{\tau} , \tilde{\mathbf{a}}_{-k} , \tilde{\boldsymbol{\tau}} , \mathbf{y}
		$,\, using CRN coupling in \eqref{eq:general_pc_sample}.\\
        Sample $(\mathbf{a}_k  , \tilde{\mathbf{a}}_k  ) \mid \mu,  \mathbf{a}_{-k} , \boldsymbol{\tau} , \tilde{\mu},  \tilde{\mathbf{a}}_{-k} , \tilde{\boldsymbol{\tau}} , \mathbf{y}
		$,\, using CRN coupling in \eqref{eq:general_pc_sample}.\\}
	}
    }
    \If{ $\|\X - \tilde{\X} \| \leq \varepsilon$}{
	\add{Sample $(\tau_0, \tilde{\tau}_0) \mid \textbf{a}, \mu, \boldsymbol{\tau}_{-0},\tilde{\textbf{a}}, \tilde{\mu}, \tilde{\boldsymbol{\tau}}_{-0}, \textbf{y}$, using Algorithm \ref{alg:max_reg}.}\\
	\For{k= 1,...,K}
	{
		\add{Sample $(\tau_k, \tilde{\tau}_k) \mid \textbf{a}, \mu, \boldsymbol{\tau}_{-k}, \tilde{\textbf{a}}, \tilde{\mu}, \tilde{\boldsymbol{\tau}}_{-k}, \textbf{y}$, using Algorithm \ref{alg:max_reg}.\\
		Sample $(\mu  , \tilde{\mu}  ) \mid \mathbf{a}_{-k} , \boldsymbol{\tau} , \tilde{\mathbf{a}}_{-k} , \tilde{\boldsymbol{\tau}} , \mathbf{y}
		$,\, using Algorithm \ref{alg:max_reg}, or Algorithm \ref{alg:max_refl} when the variances are equal.\\
        Sample $(\mathbf{a}_k  , \tilde{\mathbf{a}}_k  ) \mid \mu,  \mathbf{a}_{-k} , \boldsymbol{\tau} , \tilde{\mu},  \tilde{\mathbf{a}}_{-k} , \tilde{\boldsymbol{\tau}} , \mathbf{y}
		$,\, using Algorithm \ref{alg:max_reg}, or Algorithm \ref{alg:max_refl} when the variances are equal.\\}
	}
    }
	\caption{\add{Coupled collapsed Gibbs kernel for Model \ref{ex:gcrem}}}
	\label{alg:cg_coupling}
\end{algorithm}

\subsection{Local centering algorithm for Model \ref{ex:pmf}}
Under \eqref{eq:pmf}, 
the conditional distributions required to implement Algorithm \ref{alg:pmf} are
\begin{align*}
    \mathcal{L}(\textbf{u}_i | \textbf{v}, \rho, \tau_0, \textbf{y} ) &= N( \bmu_{\textbf{u}_i}, Q_{\textbf{u}_i}^{-1}),
    &
    \mathcal{L}(\bar{\textbf{u}}_i | \textbf{v}, \rho, \tau_0, \textbf{y} ) &= \frac{1}{\rho }  N( \bmu_{\textbf{u}_i}, Q_{\textbf{u}_i}^{-1}),\\ 
    \mathcal{L}(\textbf{v}_j | \textbf{u}, \rho, \tau_0, \textbf{y} ) &= N( \bmu_{\textbf{v}_j}, Q_{\textbf{v}_j}^{-1}),
    &
    \mathcal{L}(\bar{\textbf{v}}_j | \textbf{u}, \rho, \tau_0, \textbf{y} ) &=  \frac{1}{\rho}  N( \bmu_{{\textbf{v}}_j}, Q_{{\textbf{v}}_j}^{-1}) , 
\end{align*}
where $\bmu_{\textbf{u}_i}$, 
$\bmu_{\textbf{v}_j}$, 
$Q_{\textbf{u}_i} = (q_{rs})_{r,s =1}^d$ and $Q_{\textbf{v}_j} = (p_{rs})_{r,s =1}^d$ are given by
\begin{align*}
q_{rr} & = 1+\tau_0 \rho^2 \sum_{n:i[n]=i} v_{j[n],r}^2,  &p_{rr}= 1+\tau_0 \rho^2 \sum_{n:j[n]=j} u_{i[n],r}^2 \text{  for $r=1,...,d$}, \\
q_{rs} & = \tau_0 \rho^2 \sum_{n:i[n]=i} v_{j[n],r}  v_{j[n],s}, &p_{rs}= \tau_0 \rho^2 \sum_{n:j[n]=j} u_{i[n],r}  u_{i[n],s}, \text{  for $r,s=1,...,d$},\\
\bmu_{\textbf{u}_i} &= Q_{\textbf{u}_i}^{-1} \left(\tau_0 \rho \sum_{n:i[n]=i}  \textbf{v}_{j[n]} y_n \right), &\bmu_{\textbf{v}_j} = Q_{\textbf{v}_j}^{-1}\left(\tau_0 \rho \sum_{n:j[n]=j} \textbf{u}_{i[n]} y_n \right)\,,
\end{align*}
and
\begin{align*}
    \mathcal{L}(\rho^{-2}| \bar{\textbf{u}}, \textbf{v}, \tau_0, \textbf{y} ) &= Gamma\left(a+\frac{dI_1}{2}, \left(\frac{1}{b}+ \sum_{i=1}^{I_1} \frac{\| \bar{\textbf{u}}_i\|^2}{2} \right)^{-1} \right),\\
    \mathcal{L}(\rho^{-2}| \textbf{u}, \bar{\textbf{v}}, \tau_0, \textbf{y} ) &= Gamma\left(a+\frac{dI_2}{2}, \left(\frac{1}{b}+ \sum_{j=1}^{I_2} \frac{\| \bar{\textbf{v}}_j\|^2}{2} \right)^{-1} \right),\\
     \mathcal{L}(\tau_0| \textbf{u}, \textbf{v}, \rho, \textbf{y} ) &= Gamma  \left(c + \frac{N}{2}, \left( \frac{1}{d} + \sum_{n=1}^N  \frac{(y_n-\rho \textbf{u}_{i[n]}\textbf{v}_{j[n]})^2}{2}   \right)^{-1} \right).
\end{align*}
For the vanilla scheme with improper prior $p(\rho) \propto 1$, then 
\begin{equation}
     \mathcal{L}(\rho^{-2}| {\textbf{u}}, \textbf{v}, \tau_0, \textbf{y} ) = TG \left(\frac{\sum_{n=1}^N \textbf{u}_{i[n]}^\top \textbf{v}_{j[n]} y_n}{ \sum_{n=1}^N (\textbf{u}_{i[n]}^\top \textbf{v}_{j[n]})^2}, \frac{1}{\tau_0  \sum_{n=1}^N (\textbf{u}_{i[n]}^\top \textbf{v}_{j[n]})^2 }; 0, + \infty \right).
\end{equation}
We report in Algorithm \ref{alg:pmf2} a more detailed pseudo-code for implementing the local centering approach described in Algorithm \ref{alg:pmf}.
\begin{algorithm}
	\vspace{6pt}
			$\bar{\textbf{u}}  = \rho \cdot \textbf{u} $\\ 
            $\rho = \left(Gamma\left(a+\frac{d \, I_1 }{2}, \left(\frac{1}{b}+ \sum_{i=1}^{I_1} \frac{\| \bar{\textbf{u}}_i\|^2}{2} \right)^{-1} \right)\right)^{-\frac{1}{2}}$\\
            \For{ $i=1,...,I_1$}
            { $\bar{\textbf{u}}_i \sim N \left(\bmu_{\bar{\textbf{u}}_i}, Q_{\bar{\textbf{u}}_i}^{-1} \right)$
            }
		$  \textbf{u} = \bar{\textbf{u}}/ \rho$ \\

  $\bar{\textbf{v}}  = \rho \cdot \textbf{v} $\\
            $\rho = \left(Gamma\left(a+\frac{d \, I_2 }{2}, \left(\frac{1}{b}+ \sum_{i=1}^{I_2} \frac{\| \bar{\textbf{v}}_i\|^2}{2} \right)^{-1} \right)\right)^{-\frac{1}{2}}$\\
            \For{ $i=1,...,I_2$}
            { $\bar{\textbf{v}}_i \sim N \left(\bmu_{\bar{\textbf{v}}_i}, Q_{\bar{\textbf{v}}_i}^{-1} \right)$
            }
		$  \textbf{v} = \bar{\textbf{v}}/ \rho$ \\
	$\tau_0 \sim Gamma  \left(c + \frac{N}{2}, \left( \frac{1}{d} + \sum_{n=1}^N  \frac{(y_n-\rho \textbf{u}_{i[n]}\textbf{v}_{j[n]})^2}{2}   \right)^{-1} \right)$\\
	\caption{One iteration of BGS with local centering for Model \ref{ex:pmf}}
	\label{alg:pmf2}
\end{algorithm}

\section{Proofs}
\label{s:sec:proofs}
\subsection{Proof of Lemma \ref{lem:crn_opt}}
\begin{proof}[Proof of Lemma \ref{lem:crn_opt}]
We first show that $\bar{P}^{c} \in \bar{P}_{W_2}[P]$.
Given any updating order $(k_1,...,k_K)$, let $\sigma$ be the permutation of $(1,...,K)$, such that $(k_1,...,k_K) = (\sigma(1),...,\sigma(K))$. Define $A= 1_d - diag(Q_{1,1}^{-1},..., Q_{K,K}^{-1})\, Q$, $A^*$  the matrix whose blocks are $A^*_{i,j} =A_{k_i,k_j}$ and also $B^* = (I-L^*)^{-1}U^*$, for $U^*$ and $L^*$ upper and lower decomposition of $A^*$, i.e. $U^*+L^* = A^*$. Lastly define the matrix $B^{(\sigma)}$ as $B^{(\sigma)}_{k_i, k_j} = B^*_{i,j}$. 

\add{By composing the kernels of $\bar{P}^{c}$, we obtain that $ (\X , \Y ) \sim \bar{P}^{c} ((\x, \y ) , \cdot) $ can be written as \citep[Lemma 1]{robsahu}
\[\begin{cases}
    \X = B^{(\sigma)} \x + \mathbf{b}^{(\sigma)} + F^{(\sigma)} \mathbf{Z} \,,\\
    \Y = B^{(\sigma)} \y + \mathbf{b}^{(\sigma)} + F^{(\sigma)} \mathbf{Z} \,,\\
\end{cases}\]
where $\mathbf{b}^{(\sigma)}=(I-B^{(\sigma)})\boldsymbol{\mu}$, $F^{(\sigma)}(F^{(\sigma)})^\top = \Sigma - B^{(\sigma)} \Sigma (B^{(\sigma)})^\top$ and $\mathbf{Z}\sim N(\mathbf{0}_d, 1_d)$. Which is $W_2$-optimal, since
\begin{align*}
    \mathbb{E}[\| \X - \Y \|^2] & = \| B^{(\sigma)}\x - B^{(\sigma)}\y \|^2 = W_2\left( P(\x, \cdot), P(\y, \cdot) \right)\,,
\end{align*}
and $P (\x , \cdot) = N (B^{(\sigma)} \x + \mathbf{b}^{(\sigma)}, \, \Sigma - B^{(\sigma)} \Sigma (B^{(\sigma)})^\top) $ \citep{wass2}. Analogously, one can show that $(\bar{P}^{c})^n \in \bar{P}_{W_2} [P ^n]$, where 
	$$ {P}^n\left(\x, \cdot  \right) \stackrel{d}= N\left((B^{(\sigma)})^n\x + \left(\sum_{j=0}^{n-1} (B^{(\sigma)})^j\right) \textbf{b}^{(\sigma)}, \Sigma -(B^{(\sigma)})^n\Sigma (B^{(\sigma)})^{n, \top}\right)\,. $$   }
\end{proof}

\subsection{Proofs of the results in Section \ref{sec:bound}}
\subsubsection{Proof of Theorem \ref{thm:bound_expected_rev}}
The proof of Theorem \ref{thm:bound_expected_rev} builds upon Lemma \ref{lem:contraction}, \ref{lem:max_block_coupl_bound} and \ref{lem:distance}, whose statements and proofs are deferred after the end of the former.
\begin{proof}[Proof of Theorem \ref{thm:bound_expected_rev}] 
In the following, we will state the results assuming that $\X^{0},\Y^{0}$ are fixed, or equivalently conditioning on their values, omitting the explicit conditioning in the notation for brevity.

\add{Denote by $(t_j)_{j \geq 1}$ as the sequence of times at which maximal coupling is attempted, i.e.\
\begin{align} 
\label{eq:t_k}
	t_j &:= \min	\{ t>t_{j-1}: \|\X ^t - \Y ^t \| < \varepsilon \} &j \ge 1,
\end{align}
with $t_0:= -1$ by convention.} 
Also, let $A_j$ be a binary variable indicating whether the maximal coupling attempt at $t_j$ is successful, i.e.\
\begin{align} 
\label{eq:a_k}
	A_j &:= \begin{cases}
		1 &\text{if } \X^{t_j+1} = \Y^{t_j+1}\\
		0 &\text{otherwise}
	\end{cases} 
	 &j \ge 1.
\end{align}
By faithfulness, $A_j = 1$ implies that $\X^t = \Y^t,$ for all $t \ge t_j$ and by convention $A_{j'} = 1$ for all $j' > j$. 
Thus, $T$ can be written as
\begin{equation}\label{eq:T_expr}
	T= t_1+1 + \sum_{j=1}^{+ \infty} (1-A_j)(t_{j+1}-t_j).
\end{equation}
We bound $\mathbb{E}[T]$ using the form of \eqref{eq:T_expr}. In particular, by Lemma \ref{lem:contraction}, we have
\begin{equation}
	t_1+1 \le f_1(\| L^{-1}(\X^0-\Y^0)\| , \varepsilon,B, Q),
\end{equation}
for $f_1$ defined therein. Note that, conditionally on $(\X^0, \Y^0)$, the bound is deterministic.

\add{By definition of Algorithm  \ref{alg:2s}, $\{A_j = 0 \}$ implies $\{ A_i = 0\}$ for $i \leq j$. Then, by Lemma \ref{lem:max_block_coupl_bound}, we have
\begin{gather}\label{eq:bound_prob}
	\begin{aligned}
	&\Pr(A_j=0  ) = \Pr(A_1=0, ..., A_{j-1}=0, A_j=0 ) \\
	&=\Pr(A_1 = 0 ,..., A_{j-1}=0 )\Pr(A_j=0 \mid A_{j-1}=0)\\
	&\leq \left( \frac{1}{2} \right)^{j-1} \Pr(A_j=0 \mid A_{j-1}=0) &j\geq 1,
	\end{aligned}
\end{gather}}
Combining the last equality with the Monotone Convergence Theorem, we can rewrite the third term of \eqref{eq:T_expr} as
\add{\begin{align}
	\label{eq:ak}
	\mathbb{E}&\left[\sum_{j=1}^{+ \infty}(1-A_j) (t_{j+1}-t_j) \right] = \sum _{j=1}^{+\infty }\Pr(A_j= 0)\,\mathbb{E}\left[ t_{j+1} - t_j \mid A_j = 0 \right]\nonumber\\
	&\quad \leq \sum _{j=1}^{+\infty }\left( \frac{1}{2} \right)^{j-1}\Pr(A_j= 0\mid A_{j-1}=0)\,\mathbb{E}\left[ t_{j+1} - t_j \mid A_j = 0 \right]\nonumber\\
	&\quad =  \sum_{j=1}^{+ \infty} \left( \frac{1}{2} \right)^{j-1}\mathbb{E}\left[ \Pr ( A _j = 0 \mid A_{j-1}=0, \X ^{t_j}, \Y ^{t_j}) \mathbb{E}\left[t_{j+1}-t_{j} \mid A_{j}=0, \X^  {t_{j}}, \Y^{t_{j}}\right] \right] \, .
\end{align}}
\add{By Lemma \ref{lem:distance}, we have
\begin{equation*}
		 \Pr ( A _j = 0 \mid A_{j-1}=0, \X ^{t_j}, \Y ^{t_j}) \mathbb{E}\left[t_{j+1}-t_{j} \mid A_{j}=0, \X^  {t_{j}}, \Y^{t_{j}}\right] \leq  f_2(\varepsilon,B, Q),
	\end{equation*}
 for $f_2$ defined therein. Note that the inequality above holds almost surely, as it is independent of $\X^{t_j}$ and $\Y^{t_j}$.
Substituting in \eqref{eq:ak} we obtain
\[  \mathbb{E}\left[\sum_{j=1}^{+ \infty}(1-A_j) (t_{j+1}-t_j) \right]  \le  f_2(\varepsilon,B, Q)  \sum_{j=1}^{+\infty} \left( \frac{1}{2} \right)^{j-1}
= 2\, f_2(\varepsilon,B, Q)  .\]
Finally, we obtain
\begin{align*}
	\mathbb{E}[T] &\le
	1+ f_1(\|\X^0-\Y^0\|, \varepsilon, B,Q)  + 2\, f_2(\varepsilon,B,Q),\\
  & \leq 3 + \frac{4 + \log \|\X^{0}-\Y^{0}\|- 2\log \varepsilon + 2.5\log \kappa(Q) +  \log K - \log \lambda_{min}(\Delta) }{-\log \rho(B)} \,,
\end{align*}
which concludes the proof.}
\end{proof}

\begin{lem}	\label{lem:contraction}
  \add{Under the assumptions of Theorem \ref{thm:bound_expected_rev}, we have
	\begin{equation*}
		t_1 \le f_1(\|  \X^0 - \Y^0 \|, \varepsilon, B, Q)\,,
	\end{equation*}
  where
  \[f_{1}\left(\|\X^{0}-\Y^{0}\|, \varepsilon, B,Q\right)=\left\lceil\frac{\log \|\X^{0}-\Y^{0}\|+ \log \sqrt{\kappa (Q)}-\log \varepsilon}{-\log \rho(B)}\right\rceil \,.\]}

\end{lem}

\begin{proof}[Proof of Lemma \ref{lem:contraction}]
	\add{We begin by establishing some notation.
	For $A \in \mathbb{R}^{m\times n}$ we denote by $\| A\|_2 = \sup _{\x \neq 0, \x \in \mathbb{R}^n}{\frac {\|A \x\|_{2}}{\|\x\|_{2}}}$ the induced $2$ norm. For a square matrix, we denote by $\rho(A)$ its spectral radius and by $\kappa(A) = \|A\|_2 \|A^{-1}\|_2$ its condition number.
	We denote by $L$ a square root of $\Sigma$, i.e.\ $LL^\top = \Sigma$. From which it follows that $\| L \| _2 \| L^{-1}\|_2 = \sqrt{\kappa (\Sigma)}=\sqrt{\kappa(Q)}$.
	Finally, we define $N= L^{-1}BL$. By $\pi$-reversibility of $P$, one has $ \Sigma B^\top = B \Sigma$ \citep[Proposition 4.27]{Khare2009RATESOC}, implying $ N = N^\top$.
	Hence, by properties of the spectral radius and symmetric matrices, we have $\| N\|_2 = \rho(N)  = \rho(B)$.}
	
	\add{Recall that for Algorithm \ref{alg:2s}, the chains evolve via \emph{crn} coupling up to time $t_1$. Then, by Lemma \ref{lem:rob}, we have
  \[
  \begin{aligned}
   \left\|\X^{t}-\Y^{t}\right\|&\overset{a.s.}{=}\| B^t(\X^{0}-\Y^{0})\| \\
   &=\|L N^t L^{-1}(\X^{0}-\Y^{0})\| \\
   &\leq \|L\|_2 \rho(N)^t \|L^{-1}\|_2 \|\X^{0}-\Y^{0}\| \\
   &=\sqrt{\kappa (Q)} \rho(B)^t \|\X^{0}-\Y^{0}\| \,.
  \end{aligned}
  \]
  which is less than $\varepsilon$ if 
  \[t>\frac{\log \|\X^{0}-\Y^{0}\|+ \frac{1}{2} \log \kappa (Q) -\log \varepsilon}{-\log \rho (B)}\,. \]}
\end{proof}

\begin{lem}
\label{lem:max_block_coupl_bound}
\add{Under the assumptions of Theorem \ref{thm:bound_expected_rev}, we have, for each $j \geq 1$,
\begin{equation}
  \label{eq:max_block_coupl_bound}
 \Pr(A_j = 0\mid A_{j-1}=0) \leq \frac{1}{2} \,.
\end{equation}}
\end{lem}
\begin{proof}
\add{Note that $\Pr(A_j = 0\mid A_{j-1}=0)  = \mathbb{E}[\Pr(A_j = 0\mid A_{j-1}=0,  \X ^{t_j}, \Y ^{t_j})\mid A_{j-1}=0]$. We will show that the inner probability is at most $1/2$ almost surely.}

\add{We want to show that the probability of the complementary event in \eqref{eq:max_block_coupl_bound} is at least $1/2$.
We will drop the conditioning on $A_{j-1} = 0 $ to reduce notational clutter.  
The complementary event is
\[\begin{aligned}
 \Pr\left(\X^{t_j+1}=\Y^{t_j+1} \mid \X^{t_j}, \Y^{t_j}\right) = \prod_{k=1}^{K}\Pr\left(\X_k^{t_j+1}=\Y_k^{t_j+1} \mid \X^{t_j}, \Y^{t_j}, \X_{1:k-1}^{t_j+1} = \Y_{1:k-1}^{t_j+1}\right) \,,
\end{aligned}\] 
where $\X ^{t_j}_{1:k}$ denotes the first $k$ blocks of $\X ^{t_j}$, and  $\X_{1:k-1}^{t_j} = \Y_{1:k-1}^{t_j}$ is ignored when $k=1$.}

\add{The conditional law for the update of block $k$ is
\[
\mathcal{L}\left( \X_k^{t_j+1} \mid \X^{t_j}, \X_{1:k-1}^{t_j+1} \right)  = N\left( \bb a _k + A^k \left[\X_{1:k-1}^{t_j+1}, \X_{k:K}^{t_j} \right] , Q_{kk} ^{-1}\right)\,,
\]
where $Q_{kk}$ is the $k$-th diagonal block of $Q = \Sigma^{-1}$, $\bb a_k$ is a given vector, $A^k $ is the row corresponding to block $k$ in the matrix $A = I - \Delta ^{-1} Q$, and $\left[\X_{1:k-1}^{t_j+1}, \X_{k+1:K}^{t_j} \right]$ is the vector with first $k-1$ blocks from $\X^{t_j+1}$ and the remaining ones from $\X^{t_j}$ \citep{robsahu}. 
Recall that given $p= N(\boldsymbol{\mu}, \Sigma)$ and $q =  N(\boldsymbol{\nu}, \Sigma)$, it holds
\begin{equation}
	\label{eq:d_tv_gauss}
	\| p-q \|_{TV}= \erf\left(\sqrt{\frac{(\boldsymbol{\mu}-\boldsymbol{\nu})^\top \Sigma^{-1} (\boldsymbol{\mu}-\boldsymbol{\nu}) }{8}} \right).
\end{equation}
In Algorithm \ref{alg:2s}, we consider block-wise maximal coupling, then
\[
\begin{aligned}
 \Pr(\X_k^{t_j+1}=\Y_k^{t_j+1} &\mid \X^{t_j}, \Y^{t_j}, \X_{1:k-1}^{t_j+1} = \Y_{1:k-1}^{t_j+1})  \\&= 1 - TV ( \mathcal{L}(\X _k^{t_j+1} \mid \dots) , \mathcal{L}(\Y_k^{t_j+1} \mid \dots) ) \\ &
  = 1- \erf \left( \frac{\| Q_{kk} ^{1/2} A^k [0_{1:k-1} ; \X_{k:K} ^{t_j} - \Y_{k:K} ^{t_j} ] \|}{\sqrt{8}} \right)\,.
\end{aligned}
\]
In order to lower bound the above product of conditional probabilities, we use the following inequality 
\[
\prod _{k=1} ^K (1 - p_k ) \geq \exp  \left( -\sum _{k} p_k   - \sum _k p_k ^2  \right)\geq \exp  \left( -\sqrt{K\sum _{k} p_k^2 }   - \sum _k p_k ^2 \right)  \,,
\]
that holds true if $p_k < 0.68$, for each $k$. Note that the exponential is lower bounded by $1/2$ if $\sqrt{K \sum _k p_k ^2 } + \sum _k p_k ^2 \leq \log (2)$, which is true if $\sqrt{K \sum _k p_k ^2 } \leq 1/2$, provided that $K>1$.
Note that $\sqrt{K \sum _k p_k ^2 } \leq 1/2$ implies that $p_k \leq  0.5 < 0.68$, for each $k$.}

\add{To conclude the proof, we show $\sqrt{K \sum _k p_k ^2 } \leq 1/2$. 
\[\begin{aligned}
  \sum_k \erf \left( \frac{\| Q_{kk} ^{1/2} A^k [0_{1:k-1} ; \X_{k:K} ^{t_j} - \Y_{k:K} ^{t_j} ] \|}{\sqrt{8}} \right)^2 &\leq \frac{4}{\pi}\sum_k \frac{\| Q_{kk} ^{1/2} A^k [0_{1:k-1} ; \X_{k:K} ^{t_j} - \Y_{k:K} ^{t_j} ] \|^2}{8},\\
  &\leq \frac{1}{2\pi}\sum_k \| Q_{kk} ^{1/2} A^k \| _2 ^2  \| [0_{1:k-1} ; \X_{k:K} ^{t_j} - \Y_{k:K} ^{t_j} ] \|^2,\\
  &\leq \frac{1}{2\pi} \| \X ^{t_j} - \Y ^{t_j}\|^2 \sum_k \| Q_{kk} ^{1/2} A^k \| ^2_2, \\
  &\leq \frac{1}{2\pi} \| \X ^{t_j} - \Y ^{t_j}\|^2 K  \| \Delta  ^{1/2} A \|^2_2, \\
  &\leq \frac{1}{2\pi} \| \X ^{t_j} - \Y ^{t_j}\|^2 K  \| \Delta  ^{1/2} A\Delta ^{-1/2}  \| _2^2\| \Delta ^{1/2}\|_2^2\,, \\
\end{aligned}\]
where the first inequality comes from $\erf(x) \leq 2x / \sqrt{\pi}$.
Note that $\Delta ^{1/2} A \Delta ^{-1/2}  =  I - \Delta ^{-1/2} Q \Delta ^{-1/2} = I - \bar{Q} $, and 
\[\| I - \bar{Q}  \|_2 = \max (|1 - \lambda _{min} (\bar{Q})| , | 1 - \lambda _{max} (\bar{Q})|)\leq \max ( 1, \lambda_{max}(\bar{Q}))\,,\]
since $\bar{Q}$ is positive definite and symmetric. 
Since $\|\X ^t - \Y ^t\| < \varepsilon$, we have
\[
\sqrt{K\sum_k p_k^2} \leq \frac{1}{\sqrt{2\pi}}  K \max (1 , \rho(\bar{Q})) \, \sqrt{\rho (\Delta)}\, \varepsilon\,.
\]
If we set $\varepsilon$ as in \eqref{eq:eps_cond_block}, then the product of the conditional probabilities is at least $1/2$. }
\end{proof}

\begin{lem}
	\label{lem:distance}
	\add{Under the assumptions of Theorem \ref{thm:bound_expected_rev}, we have
	\begin{align*}
		&\Pr ( A _j = 0 \mid A_{j-1}=0, \X ^{t_j}, \Y ^{t_j}) \mathbb{E}\left[t_{j+1}-t_{j} \mid A_{j}=0, \X^  {t_{j}}, \Y^{t_{j}}\right]
  \leq  f_2 (\varepsilon, B, Q)\,, &a.s., 
	\end{align*}
	where
	\[ f_2 (\varepsilon, B, Q) =\left\lceil \frac{ 2+  \log \kappa(Q) - \frac{1}{2}\log \lambda_{min}(\Delta) + \frac{1}{2}  \log K  -\frac{1}{2}\log \varepsilon }{-\log \rho(B)} \right\rceil\,.
  \]}
\end{lem}
\begin{proof}[Proof of Lemma \ref{lem:distance}]
	\label{proof:lem_distance}
\add{We start by noting that the result of Lemma \ref{lem:contraction} can be extended for every $t_k - t_{k-1}$ with $  k \ge 2$, since the argument relies only on the form of Algorithm \ref{alg:2s}. Then
	\[
\begin{aligned}
\mathbb{E}\left[t_{j+1}-t_{j} \mid A_j = 0 , \X ^{t_{j}}, \Y ^{t_{j}}\right] &= \mathbb{E}\left[\mathbb{E}[t_{j+1}-t_{j} \mid \X^{t_{i}+1}, \Y^{t_{j}+1},A_j = 0, \X^{t_{1}}, \Y^{t_{j}}] \mid \ldots\right]\\
&\leq \mathbb{E}\left[f_{1}(\|\X^{t_{j}+1}-\Y^{t_{j}+1}\|, \varepsilon, B, Q) \mid A_j=0, \X^{t_{j}}, \Y^{t_{j}}\right]\\
&\leq  f_{1}\left(\mathbb{E}[\|\X^{t_{j}+1}-\Y^{t_{j}+1}\| \mid A_j=0, \X^{t_{j}}, \Y^{t_{j}}], \varepsilon, B, Q\right)\,.
\end{aligned}\]
where the last inequality follows from Jensen applied to $f_1(\cdot,\varepsilon, B ,Q)$.  Now, recall that $\Pr ( A_j = 0 \mid A_{j-1}= 0,  \X ^{t_j}, \Y ^{t_j})\leq 1/2$, by Lemma \ref{lem:max_block_coupl_bound}, so
\[
\begin{aligned}
 \Pr ( A _j = 0 \mid& A_{j-1}=0, \X ^{t_j}, \Y ^{t_j}) \mathbb{E}\left[t_{j+1}-t_{j} \mid A_{j}=0, \X^  {t_{j}}, \Y^{t_{j}}\right] \leq \\
  & \leq \frac{\Pr ( A_j = 0 \mid \dots ) \log \mathbb{E}[\|\X^{t_{j}+1}-\Y^{t_{j}+1}\| \mid \dots] + \frac{1}{2}\log \kappa(Q)-\frac{1}{2}\log \varepsilon  }{-\log \rho(B)}\,.
\end{aligned}\]}

\add{In the remaining part of this proof, we will use Lemma \ref{lem:bound_gen} to upper bound $$\Pr(A_j = 0 \mid A_{j-1}=0, \X ^{t_j}, \Y ^{t_j})\log  \mathbb{E}[\|\X^{t_{j}+1}-\Y^{t_{j}+1}\| \mid A_j=0, \X^{t_{j}}, \Y^{t_{j}}]\,.$$
We drop the conditioning on $\X ^{t_j}$ and  $\Y ^{t_j}$ to reduce notational clutter.
  Let $M_{j,k}=0$ be the event that the maximal coupling at time $t_j$ is unsuccessful for block $k$. Then, by definition of Algorithm \ref{alg:2s}, we have that the above quantity is equal to
  \[
  \sum _k\Pr ( M_{j,k} = 0 \mid A_j = 0) \log \mathbb{E}[\| \X^{t_j +1} - \Y^{t_j +1} \| \mid M_{j,k}=0]\,.\]
  Note that $\Pr ( M_{j,k} = 0 \mid A_j = 0) \leq\Pr ( M_{j,k} = 0 )$, since $\{ M_{j,k} =0\} \subseteq \{ A_j = 0\}$. Moreover $\sum _k \Pr ( M_{j,k} =0 ) \leq 1$.}

  \add{First, we bound the expectation. $M_{j,k} = 0$ implies that the first $k-1$ maximal coupling are successful, i.e.\ $\X _{1:k-1}^{t_j +1} = \Y _{1:k-1}^{t_j +1}$, after which the $k$-th block is updated with an unsuccessful maximal coupling and the remaining $K-k$ blocks are updated using CRN coupling. By definition of Gaussian CRN coupling, each block update is a constrained minimization w.r.t. the norm induced by $Q = \Sigma ^{-1}$, which implies that almost surely holds
  \[
  \begin{aligned}
    \| \X^{t_j +1} - \Y^{t_j +1} \|_Q& = \| L^{-1}(\X^{t_j +1} - \Y^{t_j +1}) \|  \\
    &\leq \| L^{-1}\left[ (\X ^{t_j + 1 }- \Y ^{t_j + 1}) _{1:k}\,; (\X ^{t_j}- \Y ^{t_j})_{k+1:K} \right]  \|\\
    & \leq \| L^{-1}\|_2 \left \| \left[ \mathbf{0}_{1:(k-1)}\,; \X ^{t_j + 1 }_k - \Y ^{t_j + 1} _{k}\,; (\X ^{t_j}- \Y ^{t_j})_{k+1:K}  \right] \right \| \,.
  \end{aligned} 
  \]
  For the $k$-th block, we can use Lemma \ref{lem:bound_gen} to obtain
  \[
  \mathbb{E}[\| X_k^{t_j +1} - Y_k^{t_j +1} \| \mid M_{j,k}=0] \leq 5 \dfrac{ \| U^k (\X^{t_j} - \Y^{t_j}) \| }{\| Q_{kk}^{1/2} U^k (\X^{t_j} - \Y^{t_j}) \| ^2} 
  \leq 5\dfrac{ (\lambda_{min}(Q_{kk}))^{-1/2}}{\| Q_{kk}^{1/2} U^k (\X^{t_j} - \Y^{t_j}) \| } \,,
  \]
  where $U^k$ is the $k$-th block row of the upper triangular part of $A$, and we used that $\| Q_{kk}^{1/2} U^k (\X^{t_j} - \Y^{t_j}) \|\leq 1$ that we will prove in the following steps. By combining these two inequalities, we bound the expectation as follows
  \[\begin{aligned}
    \mathbb{E}[\| \X^{t_j +1}- &\Y^{t_j +1} \| \mid M_{j,k}=0] \leq \| L \|_2 \mathbb{E}[\| L^{-1}(\X^{t_j +1} - \Y^{t_j +1}) \| \mid M_{j,k}=0] \\
    &\leq \| L \|_2 \| L ^{-1}\| _2 \, \mathbb{E}\left [\left \| \left[ \mathbf{0}_{1:(k-1)}\,; \X ^{t_j + 1 }_k - \Y ^{t_j + 1} _{k}\,; (\X ^{t_j}- \Y ^{t_j})_{k+1:K}  \right] \right \| \mid M_{j,k}=0\right ] \\
    &\leq \sqrt{\kappa(Q)} \, \left( \mathbb{E}\left [\| \X ^{t_j + 1 }_k - \Y ^{t_j + 1} _{k}  \| \mid M_{j,k}=0\right ] +  \| (\X ^{t_j}- \Y ^{t_j})_{k+1:K}\| \right) \\
    &\leq \sqrt{\kappa(Q)} \, \left( 5\dfrac{ (\lambda_{min}(Q_{kk}))^{-1/2}}{\| Q_{kk}^{1/2} U^k (\X^{t_j} - \Y^{t_j}) \|} + \varepsilon  \right) \,. \\
  \end{aligned}\]
  Note that if we take the logarithm of the upper bound, we can further obtain
  \[\begin{aligned}
    \log &\mathbb{E}[\| \X^{t_j +1}- \Y^{t_j +1} \| \mid M_{j,k}=0] \leq \log \sqrt{\kappa(Q)}  + \log \left( 5\dfrac{ (\lambda_{min}(\Delta))^{-1/2}}{\| Q_{kk}^{1/2} U^k (\X^{t_j} - \Y^{t_j}) \|} + \varepsilon  \right) \,, \\
    &= \frac{1}{2} \log \kappa(Q)  + \log \left( 5\dfrac{ (\lambda_{min}(\Delta))^{-1/2}}{\| Q_{kk}^{1/2} U^k (\X^{t_j} - \Y^{t_j}) \|}   \right) + \log \left( 1 + \varepsilon \dfrac{\| Q_{kk}^{1/2} U^k (\X^{t_j} - \Y^{t_j}) \|}{5  (\lambda_{min}(\Delta ))^{-1/2}}  \right) \,. \\
    &= \frac{1}{2} \log \kappa(Q)  + \log \left( 5\dfrac{ (\lambda_{min}(\Delta))^{-1/2}}{\| Q_{kk}^{1/2} U^k (\X^{t_j} - \Y^{t_j}) \|}   \right) +   \varepsilon \dfrac{\| Q_{kk}^{1/2} U^k (\X^{t_j} - \Y^{t_j}) \|}{5  (\lambda_{min}(\Delta))^{-1/2}}   \,. \\
  \end{aligned}\]}

  \medskip
  \add{Denote with $\alpha _k = \| Q_{kk}^{1/2} U^k (\X^{t_j} - \Y^{t_j}) \|$. In the proof of Lemma \ref{lem:max_block_coupl_bound}, we have shown that $\sqrt{K \sum _k \alpha_k ^2 }\leq \sqrt{\pi /2 }$. From which we have $S := \sum_k \alpha _k \leq \sqrt{\pi / 2}$ and $\alpha _k \leq \sqrt{\pi / 2K } < 1$, if $K>1$. Now, we can bound the sum over $k$ as
  \[\begin{aligned}
    &\sum _k\Pr ( M_{j,k} = 0 ) \log \mathbb{E}[\| \X^{t_j +1} - \Y^{t_j +1} \| \mid M_{j,k}=0] \leq \\
    &\leq  \frac{1}{2} \log \kappa(Q) +    \dfrac{\varepsilon}{5  }  \sqrt{\lambda_{min}(\Delta)}+ \log 5 - \frac{1}{2}\log \lambda_{min}(\Delta) + \sum _k\Pr ( M_{j,k} = 0 ) \log \left( \dfrac{ 1}{\alpha _k}   \right) \,. \\
\end{aligned}\]
Note that 
\[\Pr ( M_{j,k} = 0 ) =\Pr ( M_{j,k} = 0\mid M_{j,k-1} \neq 0 ) \cdot\Pr ( M_{j,k-1} \neq 0 ) \leq \erf \left( \frac{\alpha _k}{\sqrt{8}} \right)\leq  \frac{\alpha _k}{\sqrt{2\pi}} \,,\]
since $\erf(x) \leq 2x / \sqrt{\pi}$. Define now $p_k = \alpha _k / S$, so that $\sum _k p_k = 1$. We can bound the last term as
\[\begin{aligned}
  \sum _k\Pr ( M_{j,k} = 0 ) \log \left( \dfrac{ 1}{\alpha _k}   \right) &\leq \sum _k \frac{\alpha _k}{\sqrt{2\pi}} \log \left( \dfrac{ 1}{\alpha _k}   \right) \\
  &= \frac{S}{\sqrt{2\pi}} \sum _k p_k \log \left( \dfrac{1}{S p_k}   \right) \\
  &= \frac{S}{\sqrt{2\pi}} \left( -\log S + \sum _k p_k \log \left( \dfrac{1}{ p_k}   \right)  \right)\\
  &\leq  \frac{S}{\sqrt{2\pi}} ( -\log S + \log K )  \,,\\
\end{aligned}\]
from the bound on the entropy of a discrete distribution. Since $S \leq \sqrt{\pi / 2}$, we have that $S / \sqrt{2\pi} \leq 1/2$, and $-S \log S \leq \frac{2}{5} $. We can conclude that
\[\begin{aligned}
  \sum _k\Pr ( M_{j,k} = 0 ) &\log \mathbb{E}[\| \X^{t_j +1} - \Y^{t_j +1} \| \mid M_{j,k}=0] \leq \\
  &\leq  \frac{1}{2} \log \kappa(Q) +    \dfrac{\varepsilon}{5  }  \sqrt{\lambda_{min}(\Delta)}+ \log 5 - \frac{1}{2}\log \lambda_{min}(\Delta) + \frac{1}{2}  \log K + \frac{2}{5\sqrt{2\pi }}   \,, \\
  &\leq 2+  \frac{1}{2} \log \kappa(Q) - \frac{1}{2}\log \lambda_{min}(\Delta) + \frac{1}{2}  \log K   \,, \\
\end{aligned}\]
where we used $\varepsilon < 1$ and $\lambda_{min}(\Delta) \leq 1$.}
\end{proof}

\subsubsection{Proof of Lemma \ref{lem:bound_gen}}

In order to prove Lemma \ref{lem:bound_gen}, we use an instrumental lemma, namely Lemma \ref{lem:bound12}. In the following, we denote $TG(\mu, \sigma^2;a,b)$ a truncated Gaussian with mean parameter $\mu$, variance parameter $\sigma^2$, and constrained between $a$ and $b$.
\begin{lem}
	\label{lem:bound12}
	Let $\sigma\in(0,\infty)$ and $\alpha \in \mathbb{R}$ and ${X} \sim TG(0,\sigma^2; \alpha, + \infty)$. It holds that
	$$ \mathbb{E}[{X}] \le \max(0,\alpha) + \sigma \sqrt{\frac{2}{\pi}}, \; \; \; \mathbb{E}[{X}^2] \le \sigma^2 + \alpha^2+\sqrt{\frac{2}{\pi}} \alpha \sigma.$$
\end{lem}
\begin{proof}[Proof of Lemma \ref{lem:bound12}]
	\label{proof:unitg}
	For ${T} \sim TG(\mu,\sigma^2; \alpha, +\infty)$, we know
	\begin{equation} 
		\label{eq:fm}
		\mathbb{E}[{T}] = \mu + \frac{\phi\left(\frac{\alpha-\mu}{\sigma}\right)}{1- \Phi\left(\frac{\alpha-\mu}{\sigma}\right) }\sigma,
	\end{equation}
	\begin{equation}
		\label{eq:sm}
		\mathbb{E}[{T}^2] =\sigma^2 + \sigma^2 \frac{\frac{\alpha-\mu}{\sigma} \phi\left(\frac{\alpha-\mu}{\sigma}\right)}{1- \Phi\left(\frac{\alpha-\mu}{\sigma}\right) }+ \mu^2 + 2\mu\sigma \frac{\phi\left(\frac{\alpha-\mu}{\sigma}\right)}{1- \Phi\left(\frac{\alpha-\mu}{\sigma}\right) }, 
	\end{equation}
	where $\phi(\cdot), \Phi(\cdot)$ denote respectively the density and the cumulative functions of the standard normal distribution. We divide the proof in the cases $\alpha <0 $ and $\alpha \ge 0$.
	
	Consider $\alpha <0$. Denote by $c_{\mu,\sigma^2; \alpha}$ the normalizing constant $c_{\mu,\sigma^2;\alpha}=\int_{\alpha}^{+\infty}  e^{\frac{-(x-\mu)^2}{2\sigma^2}} dx$. Since $X\sim TG(0,\sigma^2; \alpha, + \infty)$ and $\alpha <0$, we have
	\begin{equation*}
		\mathbb{E}[{X}]= \int_{\alpha}^{+\infty} x \frac{e^{\frac{-x^2}{2\sigma^2}}}{c_{0,\sigma^2;\alpha}} dx \le \int_{0}^{+\infty} x  \frac{e^{\frac{-x^2}{2\sigma^2}}}{c_{0,\sigma^2;\alpha}} dx.
	\end{equation*}
	Multiplying and dividing by $c_{0,\sigma^2;0}$ and recalling that for $Y \sim TG(0,\sigma^2;0,+\infty)$ one has $\mathbb{E}[Y] = \sqrt{\frac{2}{\pi}} \sigma$ from \eqref{eq:fm},  then
	$$ \int_0^{+\infty} x  \frac{e^{\frac{-x^2}{2\sigma^2}}}{c_{0,\sigma^2;\alpha}} dx \le \frac{c_{0,\sigma^2;0}}{c_{0,\sigma^2; \alpha}} \sqrt{\frac{2}{\pi}} \sigma.$$ 
	Furthermore since $ \frac{c_{0,\sigma^2;0}}{c_{0,\sigma^2; \alpha}} = \frac{c_{0,\sigma^2;0}}{\int_{\alpha}^0 e^{-\frac{t^2}{2\sigma^2}}dx +\, c_{0,\sigma^2;0}} <1$, it follows that 
 \begin{equation}
 \label{eq:res1}
     \mathbb{E}[{X}] \le \sqrt{\frac{2}{\pi}} \sigma.
 \end{equation}
	
	Consider now $\alpha \ge 0$. We prove 
 \begin{equation}
 \label{eq:res2}
 \mathbb{E}[{X}] \le  \alpha + \sqrt{\frac{2}{\pi}} \sigma= \mathbb{E}[Y], \end{equation} 
 with $Y \sim TG(\alpha, \sigma^2; \alpha, + \infty)$.
	We exploit stochastic ordering: if there exists a coupling between ${X} \sim TG(0,\sigma^2; \alpha, + \infty)$ and $ Y \sim TG(\alpha, \sigma^2; \alpha, + \infty)$ such that $\Pr({X}<Y) = 1$, then the desired result follows.
	Given that the Gaussian distribution belongs to the exponential family, it has monotone likelihood ratio in its canonical statistics, that is $x$, hence implying stochastic ordering.
	
	For the second moment, by \eqref{eq:sm} and the bound just found in \eqref{eq:res1} and \eqref{eq:res2}, we get
	\begin{align*}
		\mathbb{E}[{X}^2] = \sigma^2 + \alpha \sigma \mathbb{E}[Y] \le \sigma^2 + \alpha^2+\sqrt{\frac{2}{\pi}} \alpha \sigma, 
	\end{align*}
	for $Y \sim TG(0,1; \frac{\alpha}{\sigma}, +\infty)$.
\end{proof}

\begin{proof}[Proof of Lemma \ref{lem:bound_gen}]
	\label{proof:bound_gen}
	\add{Denote with $\mathbf{z} = \Sigma^{-1/2}(\bb \xi-\bb \nu)$, and $\mathbf{e} = \mathbf{z} / \| \mathbf{z}\|$, as in Algorithm \ref{alg:max_refl}, and denote $\zeta = \| \mathbf{z} \|$.
  If we condition on $\X \neq \Y$, then
  \begin{equation}\label{eq:trunc_norm}
    2 \mathbf{z}^\top \dot{\X}  + \mathbf{z}^\top \mathbf{z} = 2 \zeta \,\mathbf{e}^\top  \dot{\X} + \zeta^{2} > - 2 \log U \,, \qquad U \sim \text{Unif}(0,1)\,.
  \end{equation}
  Moreover, it holds that
  \[
  \begin{aligned}
    \X-\Y  =\bb \xi-\bb \nu+\Sigma^{1 / 2}(\dot{\X}-\dot{\Y}) 
    &=\Sigma^{1 / 2}\left(\mathbf{z}+2(\bb e ^{\top} 
    \dot{\X}) \bb e \right) \\
    &= \left( \zeta^2 + 2 \zeta \mathbf{e}^\top \dot{\X} \right) \frac{1}{\zeta }\Sigma ^{1 / 2} \mathbf{e} \\
    & =  \left( \zeta^2 + 2 \zeta \mathbf{e}^\top \dot{\X} \right) \frac{1}{\zeta ^2}(\bb \xi - \bb \nu)\,.
  \end{aligned}
  \]
  We can obtain the lower bound by noticing
  \[\mathbb{E}\left[\|A(\X-\Y)\|\mid \X \neq \Y\right] \geq \mathbb{E} [-2 \log U ] \frac{\|A(\bb \xi - \bb \nu)\|}{\zeta^2} = 2 \frac{\|A(\bb \xi - \bb \nu)\|}{\zeta^2}\,.
  \]}

  \add{We move to the upper bound.
  Since $\dot{\X}\sim N(\mathbf{0}, I_d)$, we have that $\mathbf{e}^\top \dot{\X} \sim N(0,1)$. By \eqref{eq:trunc_norm}, we have that $W = 2 \zeta \mathbf{e}^\top \dot{\X}$ conditional on $U$ is distributed as a truncated normal $TN\left( 0, 4\zeta^2 ; -2\log U - \zeta ^2, +\infty  \right)$. Then, by Lemma \ref{lem:bound12}, we have
  \[
  \mathbb{E}[W] =\mathbb{E}[\mathbb{E}[W \mid U]] \leq \mathbb{E}\bigg[ \max ( 0, -\zeta^2 - 2 \log U ) + \sqrt{\frac{8}{\pi} }\,\zeta\bigg] = 2 e^{-\frac{\zeta^2}{2}} + \sqrt{\frac{8}{\pi}} \zeta\,,
  \]
  from which the result follows.}
\end{proof}

\subsubsection{Proof of Theorem \ref{thm:tail_bound_block}}
The proof of Theorem \ref{thm:tail_bound_block} builds upon Lemma \ref{lem:time_diff_tail_bound}, which can be found below.

\begin{proof}[Proof of Theorem \ref{thm:tail_bound_block}]
  \add{Recall that 
  \[T = 1 + t_1 + \sum _{j\geq 1} (1-A_j ) (t_{j+1} - t_j)\,,\]
  therefore, if we denote with $G = \sum _{j\geq 1} (1-A_j ) (t_{j+1} - t_j)$, we have for each $x> 0 $ 
  \[\{ T \geq 1 +y +x g \} \subseteq  \{t_1 > y\} \cup \{ G > g \} \cup \{ (1-A_j ) (t_{j+1 } - t_j) > x, \text{ for some } j \leq g \}\,.\]
  By union bound, we have
  \[Pr ( T \geq 1 + y +  x g ) \leq Pr (t_1 > y) + Pr ( G > g ) +  \sum _{j=1} ^g Pr ( (1-A_j ) (t_{j+1 } - t_j) > x )\,.\]}
  
  \add{By Lemma \ref{lem:max_block_coupl_bound}, the tail of $G$ is bounded by the tail of a geometric random variable with parameter $1/2$, therefore $Pr ( G > g ) \leq \left( \frac{1}{2} \right)^g$. Moreover, by Lemma \ref{lem:contraction}, we have $Pr (t_1 > y) =0$, if $y \geq f_1(\|\X ^0 - \Y ^0\|, \varepsilon, B, Q)$. Finally, by Lemma \ref{lem:time_diff_tail_bound}, we have
  \[Pr ( (1-A_j ) (t_{j+1 } - t_j) > x ) \leq  \left( \frac{1}{2} \right)^j Pr ( t_{j+1} -t_j > x \mid A_j = 0) \leq \left( \frac{1}{2} \right) ^j U(x)\,,\]
  where $U(x)$ denotes the upper bound in Lemma \ref{lem:time_diff_tail_bound}. 
  By combining these three bounds, we have for any $y \geq f_1(\|\X ^0 - \Y ^0\|, \varepsilon, B, Q)$ and any $g \geq 1$
  \[Pr ( T \geq 1 + y +  x g ) \leq \left( \frac{1}{2} \right)^g + U(x) \sum _{j=1} ^g \left( \frac{1}{2} \right) ^j \leq \left( \frac{1}{2} \right)^g + U(x)\,.\]}

  \add{We can set $g = \floor {\log U(x) /(- \log 2)}$ to obtain
  \[Pr \left(  T \geq 1 + y + x \left( -\frac{\log U(x)}{\log 2} +1 \right) \right)  \leq 2 U(x)\,.\]
  Now, let 
  \[\begin{aligned}
    x \left(\frac{-\log U(x)}{\log 2} +1\right)& = x \left( \frac{1}{16 \log 2}\lambda_{min}(\Delta)\,\varepsilon^2\,\left( \frac{e^{-\log \rho (B)x }}{\kappa(  Q)}-1 \right)^2  - \frac{\log 7}{ \log 2 }+1\right) \\
    &\leq   \frac{1}{16 \log 2}\frac{\lambda_{min}(\Delta)\,\varepsilon^2}{\kappa(  Q)^2} \, x e^{-2\log \rho (B)x } = t(x) \,,
  \end{aligned}\]
  for each $x>0$. If we denote the last upper bound with $t(x) $, we have that is strictly increasing in $x$ on the set $[0, +\infty)$, and its inverse is given by
  \[x(t) = -\frac{1}{2\log \rho (B)} W \left(\gamma  \,t\right)\,,\qquad
  \gamma = 32\log 2 \frac{-\log \rho(  B) \,\kappa (  Q)^2}{\lambda _{min}(  \Delta)\,\varepsilon^2 }\,,\]
  where $W$ denotes the Lambert W function, i.e., the inverse of the function $w \to w e^w$, on the set $w \in [-1, +\infty)$. Therefore, for any $y = t> f_1(\|\X ^0 - \Y ^0\|, \varepsilon, B, Q)$, we have $P(T > 1+2t) \leq 2 U(x(t)) $. 
  Note that for each $z=\gamma  \,t > 32\log 2$ (which is implied by the condition of Theorem \ref{thm:tail_bound_block}), we can lower bound
  \[W(z) \geq \sqrt{ \frac{z}{\log z} } \,, \qquad \quad \left( \sqrt{\frac{z}{\log z}} - 1 \right)^2 \geq \frac{1}{4\log 2}\frac{z}{\log z}\,,\]
  From which we obtain
  \[
    Pr ( T > 1 + 2t ) \leq 14\,\exp \left( -\frac{1}{2}\frac{-\log \rho (B)\, t}{\log (\gamma t)}\right) \,.
  \]}
\end{proof}

\begin{lem}\label{lem:time_diff_tail_bound}
  \add{\[
  Pr(t_{j+1 }- t_j > x \mid A_j = 0) \leq 7\,\exp \left( -\frac{1}{16}\lambda_{min}(\Delta)\,\varepsilon^2\,\left( \frac{e^{-\log \rho (B)x }}{\kappa(  Q)}-1 \right)^2\right)\,.
  \]}
\end{lem}
\begin{proof}
  \add{Note that $Pr(t_{j+1 }- t_j >x\mid A_j = 0) = \mathbb{E}[Pr(t_{j+1 }- t_j > x \mid A_j = 0, \X ^{t_j}, \Y^{t_j})\mid A_j = 0]$, we will find a bound for the inner probability that holds almost surely independently of $\X ^{t_j}$ and $\Y^{t_j}$. By Lemma \ref{lem:contraction}, we have 
  \[\begin{aligned}
    Pr(t_{j+1 }- t_j > x \mid &A_j = 0, \X ^{t_j}, \Y^{t_j})\leq Pr\left( f_1(\| \X ^{t_j + 1} - \Y ^{t_j +1 }\| , \varepsilon,   B ,   Q ) > x \mid A_j = 0, \dots \right)\\
    &= Pr\left(\| \X ^{t_j + 1} - \Y ^{t_j +1 }\| > \frac{\varepsilon}{\sqrt{\kappa(  Q)}} \, e^{-\log \rho (B)x }  \mid A_j = 0, \X ^{t_j}, \Y^{t_j} \right)\,.
  \end{aligned}\]
  As in Lemma \ref{lem:max_block_coupl_bound}, we let $M_{j,k}=0$ be the event that the maximal coupling at time $t_j$ is unsuccessful for block $k$. Then, we can write the above probability as 
  \[\sum _{k=1}^K Pr (M_{j,k}=0 \mid A_j=0,\dots) Pr \left(\| \X ^{t_j + 1} - \Y ^{t_j +1 }\| > \frac{\varepsilon\, e^{-\log \rho (B)x } }{\sqrt{\kappa(  Q)}} \Bigg \vert M_{j,k} = 0, \X ^{t_j}, \Y^{t_j} \right)\,,\]
  where we know that $\sum _k Pr (M_{j,k}=0 \mid A_j=0,\X ^{t_j}, \Y^{t_j}) = 1$.}

  \add{Using the same argument as Lemma \ref{lem:max_block_coupl_bound}, we have that, conditioned on $M_{j,k}=0$, it  holds 
  \[\begin{aligned}
    \| \X^{t_j +1} - \Y^{t_j +1} \| &\leq \|  L \|_2 \| \X^{t_j +1} - \Y^{t_j +1} \|_Q \\
    &\leq \sqrt{\kappa (  Q)} \left \| \left[ \mathbf{0}_{1:(k-1)}\,; \X ^{t_j + 1 }_k - \Y ^{t_j + 1} _{k}\,; (\X ^{t_j}- \Y ^{t_j})_{k+1:K}  \right] \right \| \\
    &\leq \sqrt{\kappa (  Q)} \left( \| \X_k ^{t_j+1} - \Y ^{t_j+1}_k\| + \varepsilon \right)\,.
  \end{aligned}\]
  Therefore, we can bound the probability with
  \[Pr \left(\| \X ^{t_j + 1} _k- \Y ^{t_j +1 }_k\| >\varepsilon\left(  \frac{\, e^{-\log \rho (B)x } }{\kappa(  Q)} -1 \right) \Bigg \vert M_{j,k} = 0, \X ^{t_j}, \Y^{t_j} \right)\,.\]
  Now, denote with $\alpha _k = \| Q_{kk}^{1/2} U^k (\X^{t_j} - \Y^{t_j}) \|$. From the proof of Lemma \ref{lem:bound_gen}, we have that, conditioned on $M_{j,k}=0$, $\X^{t_j}$ and $\Y ^{t_j}$, holds
  \[\|X ^{t_j+1} - \Y ^{t_j+1} \| \overset{d}{=} \left( \alpha _k ^2 + 2\alpha _k W  \right) \frac{\| U^k (\X^{t_j} - \Y^{t_j}) \|}{\alpha _k ^2}\,,\]
  where $W \sim N(0,1)$, and $\alpha _k ^2 + 2\alpha _k W > -2 \log U$, for some $U \sim Unif(0,1)$ independent of $W$. Moreover, $\| U^k (\X^{t_j} - \Y^{t_j}) \| \leq (\lambda_{min}(\Delta))^{-1/2} \alpha _k$. Therefore, we can furhter bound the probability as
  \[\begin{aligned}
    &Pr \left( \left( \alpha _k ^2 + 2\alpha _k W  \right) \frac{\| U^k (\X^{t_j} - \Y^{t_j}) \|}{\alpha _k ^2} >\varepsilon\left(  \frac{\, e^{-\log \rho (B)x } }{\kappa(  Q)} -1 \right)  \Bigg \vert M_{j,k} = 0, \X ^{t_j}, \Y^{t_j} \right)\\
    &\leq \mathbb{E}\left[ Pr \left( \alpha _k ^2 + 2\alpha _k W  >\alpha _k \,\varepsilon\, (\lambda_{min}(\Delta))^{1/2}\left(  \frac{\, e^{-\log \rho (B)x } }{\kappa(  Q)} -1 \right)  \Bigg \vert \alpha _k ^2 + 2\alpha_k W > -2\log (U) \right) \right]\,.
  \end{aligned}\]
  Denote with $\ell =\alpha _k \, \varepsilon\, (\lambda_{min}(\Delta))^{1/2}\left(  \frac{\, e^{-\log \rho (B)x } }{\kappa(  Q)} -1 \right)$ and recall that for a standard Gaussian distribution holds $Pr ( Z > x) \leq e^{-x^2/2}$, for all $x > 0$. Than, we can compute the inner probability as
  \[\begin{aligned}
    \int _0 ^1 &Pr ( \alpha _k ^2 + 2\alpha _k W > \max (\ell , -2 \log u ) )  du \leq \int _0 ^1 e^{- \frac{(\max (\ell , -2\log u )- \alpha_k ^2)^2}{8 \alpha_k ^2}}du \\
    &= \int _0 ^{+\infty } e^{- \frac{(\max (\ell , t)- \alpha_k ^2)^2}{8 \alpha_k ^2}}e^{-\frac{t}{2}} dt \\
    &= \int _0 ^{\ell } e^{- \frac{(\ell- \alpha_k ^2)^2}{8 \alpha_k ^2}}e^{-\frac{t}{2}} dt + \int _{\ell}^{+\infty} e^{- \frac{(t- \alpha_k ^2)^2}{8 \alpha_k ^2}}e^{-\frac{t}{2}} dt \\
    &= \int _0 ^{\ell } e^{- \frac{(\ell- \alpha_k ^2)^2}{8 \alpha_k ^2}}e^{-\frac{t}{2}} dt + 2 \sqrt{2\pi}\alpha _k\int _{\ell}^{+\infty} \frac{1}{\sqrt{8\pi \alpha _k ^2}} e^{- \frac{(t+ \alpha_k ^2)^2}{8 \alpha_k ^2}} dt \\
    &=  e^{- \frac{(\ell- \alpha_k ^2)^2}{8 \alpha_k ^2}}+ 2 \sqrt{2\pi}\alpha _k \, e^{- \frac{(\ell+ \alpha_k ^2)^2}{8 \alpha_k ^2}}  \\
    &\leq  e^{- \frac{\ell^2}{16\alpha_k ^2}}  \left( e^{ \frac{1}{8}\alpha _k ^2} + 2 \sqrt{2\pi}\alpha _k \right)  \,,
  \end{aligned}\]
  where in the last inequality we used that $(a-b)^2 \geq \frac{1}{2}a^2 - b ^2$ and $(a+b)^2 \geq \frac{1}{2}a^2 + b ^2$, for any $a,b >0$. 
  From the proof of Lemma \ref{lem:max_block_coupl_bound}, recall that $\alpha _k  <1$, if $K>1$. Therefore, $e^{ \frac{1}{8}\alpha _k ^2} + 2 \sqrt{2\pi}\alpha _k < e^{ \frac{1}{8}} + 2 \sqrt{2\pi}\leq 7$. By substituting back $\ell$, we obtain the result.}
\end{proof}

\subsection{Proof of Corollary \ref{cor:gcrem_bound}}
The proof of Corollary \ref{cor:gcrem_bound} follows directly from Corollary \ref{cor:bound_rev} and Lemma \ref{lem:a2} below. 
In this section, we will always assume $P^{(F)}= P_2 P_1$ and $B^{(F)}$ as in Lemma \ref{lem:rob} accordingly. Let also $L$ be the block triangular matrix such that $L L^\top = \Sigma$.

\begin{lem}
	\label{lem:a2}
	Let $B^{(F)}$ and $B^{(FB)}$ be respectively the auto-regressive matrices induced by $P^{(F)}$ and $P^{(FB)}$, the forward backward kernel of Section \ref{ssec:bound_pirev}, for $\pi = N(\boldsymbol{\mu}, \Sigma)$, with $K=2$ blocks. Let $N^{(F)}=L^{-1}B^{(F)}L$ and $ N^{(FB)}= L^{-1}B^{(FB)}L$. For all $ t > 1$ one has
	\[ \| \left(N^{(F)}\right)^t\|_2 \le \rho \left( \left(B^{(FB)}\right)^{t-1}\right) =\rho \left(B^{(F)} \right)^{t-1},\]
 where $\| \cdot \|_2$ is the induced $2$-norm. 
\end{lem}
\begin{proof}
	For a two-block Gaussian it holds $\mathbb{E}[\x_{(i)} | \x_{(j)}] = A_{ij} \x_{(j)} +\textbf{a}_{(i)}$ for $i\ne j \in \{1,2\}$. So
	\begin{align}
		\label{eq:bbf}
		&B^{(F)} = \left(\begin{array}{c|c}
			0 &A_{12}\\  \hline
			0 & A_{21} A_{12}
		\end{array} \right),
		&B^{(FB)} = \left(\begin{array}{c|c}
			0 &A_{12} A_{21} A_{12}\\  \hline
			0 & A_{21} A_{12}
		\end{array} \right).
	\end{align}
	Note that from the above $\rho \left(B^{(F)} \right) = \rho( A_{21} A_{12}) = \rho \left(B^{(FB)}\right)$.
	One can rewrite \eqref{eq:bbf} as $B^{(F)}= A_2 A_1$ and $B^{(FB)}= A_1 A_2 A_1$ for
	\begin{align}
		\label{eq:a2}
		&A_1 = \left(\begin{array}{c|c}
			1 &0\\  \hline
			A_{21} & 0
		\end{array} \right), 
		&A_2 = \left(\begin{array}{c|c}
			0 &A_{12}\\ \hline
			0 & 1
		\end{array} \right) .
	\end{align}
	Simple algebra shows that $A_1^2=A_1, A_2^2=A_2$ and $(B^{(F)})^t = A_2 (B^{(FB)})^{t-1}$. Furthermore: 
	\begin{equation}
		\label{eq:n}
		(N^{(F)})^t = L^{-1} (B^{(F)})^t L = L^{-1}A_2 L L^{-1} (B^{(FB)})^{t-1} L = \tilde{A}_2 (N^{(FB)})^{t-1},\end{equation}
	where we defined $\tilde{A}_2:= L^{-1}A_2 L$. By symmetry of $N^{(FB)}$ and submultiplicativity of the matrix norm, it follows: 
	$$\| (N^{(F)})^t \|_2 = \| \tilde{A}_2 (N^{(FB)})^{t-1}\|_2 \le \| \tilde{A}_2\|  _2\rho(N^{(FB)})^{t-1}.$$
\end{proof}

\subsection{Proof of Theorem \ref{thm:bound_expected}}
\begin{proof}[Proof of Theorem \ref{thm:bound_expected}]
	\add{With the same reasoning as in Theorem \ref{thm:bound_expected_rev}, one can show that
	\[T \le 1+\tilde{f}_1(\|\X^0-\Y^0\|, \varepsilon, B, Q,\delta)  +  2\tilde{f}_2(\varepsilon,B, Q, \delta)\,,\]
	where
	\begin{align*}
		\tilde{f}_1(r, \varepsilon, B,\delta)
		&=\max \left( n^*_\delta, \left\lceil\frac{\log \|\X^{0}-\Y^{0}\|+ \log \sqrt{\kappa (Q)}-\log \varepsilon}{-\log \left( 1-\frac{1-\rho(B)}{1+\delta} \right)}\right\rceil\right),
	\end{align*}
	and
	\begin{equation*}
		\tilde{f}_2(\varepsilon,B, \delta) = \max \left( n^*_\delta, \left\lceil \frac{ 2+  \log \kappa(Q) - \frac{1}{2}\log \lambda_{min}(\Delta) + \frac{1}{2}  \log K  -\frac{1}{2}\log \varepsilon }{-\log \left( 1-\frac{1-\rho(B)}{1+\delta} \right)} \right\rceil \right).
	\end{equation*}
	Then combining the results with the same reasoning as in the proof of Theorem \ref{thm:bound_expected_rev} gives the result. }
	
	\add{The form of $\tilde{f}_1$ comes from a generalization of $f_1$ in Lemma \ref{lem:contraction}. Following the steps of the proof of Lemma \ref{lem:contraction}, we have
	\[
	\left\|\X^{t}-\Y^{t}\right\|=\| B^t(\X^{0}-\Y^{0})\| \leq\sqrt{\kappa (Q)} \| N^t\| \|\X^{0}-\Y^{0}\| \,.
	\]
	The definition of $n^*_\delta$ implies that for all $t \ge n^*_\delta$:
	\[\| N^t\|_2 \le \left( 1-\frac{1-\rho(N)}{1+\delta} \right)^t = \left( 1-\frac{1-\rho(B)}{1+\delta} \right)^t  \,.\]
	Imposing $\| \X ^t - \Y ^t \| < \varepsilon$ and solving for $t$ leads to the result.}
	
	\add{As for $\tilde{f}_2$, the result follows from substituting $f_1$ with $\tilde{f}_1$ in the proof of Lemma \ref{lem:distance}.}
\end{proof}

\subsection{Proof of Lemma \ref{lem:asymp}}
\begin{proof}[Proof of Lemma \ref{lem:asymp}]
	By \eqref{eq:d_tv_gauss}, and the definition of $\bar{c}_d$,
it holds 
	\begin{equation}
		\| p-q \|_{TV} = \erf\left(\bar{c}_d d^{-\alpha + \frac{1}{2}}  \right).
	\end{equation}
	If $\alpha>\frac{1}{2}$, as $d \rightarrow + \infty $, Taylor expanding the $\erf$ function around $0$ gives
	\begin{equation} 
		\label{eq:erf_expansion}
		{\textstyle \Pr_{max}}(p, q)= 1- \|p-q\|_{TV} \asymp 1-  \frac{2 \bar{c_d}}{\sqrt{\pi}} d^{-\alpha +\frac{1}{2}}.
	\end{equation}
	If instead $0 <\alpha < \frac{1}{2}$, the argument of the $\erf$ function goes to $+\infty$. We exploit Gaussian tail bounds to characterize the behaviour. Recall indeed that
	$$ \erf(x) = 2 \Phi(\sqrt{2} x)-1,$$ $${\textstyle \Pr_{max}}(p,q) = 1-  \erf\left(\bar{c}_d d^{-\alpha + \frac{1}{2}}  \right) = 2 \left(1-\Phi\left(\sqrt{2} d^{-\alpha + \frac{1}{2}} \bar{c}_d \right)\right),$$ where $\Phi(\cdot)$ indicates the standard Gaussian cumulative. Furthermore for $x$ going to infinity, it holds $1-\Phi(x) \asymp \frac{\phi(x)}{x}$, where $\phi$ denotes the density function of the standard Gaussian, it follows
	$${\textstyle \Pr_{max}}(p,q )   \asymp \frac{d^{\alpha-\frac{1}{2}}}{\sqrt{\pi} \bar{c}_d} e^{-\frac{\bar{c}_d^2}{\sqrt{2}} d^{-2\alpha+1} }.$$
 
	On the other hand, considering the product of independent maximal couplings, the argument of each $\erf$ function goes to $0$ as $d \rightarrow + \infty$ and hence we exploit the same expansion as in \eqref{eq:erf_expansion}, getting
	\begin{align*}
		\prod_{i=1}^{d}{\textstyle \Pr_{max}}(p_i, q_i) &= \prod_{i=1}^d \left(1-\erf \left(d^{-\alpha} \sqrt{\frac{c_i^2}{8}} \right)\right) \\
		&=  \prod_{i=1}^d \left(1- \frac{|c_i|}{\sqrt{2\pi}} d^{-\alpha} + o(d^{-\alpha}) \right) \asymp e^{-d^{1-\alpha} \tilde{c}_d},
	\end{align*} 
	where $\tilde{c}_d = \frac{\sum_{i=1}^d |c_i|}{d \sqrt{2\pi}}$. 
\end{proof}
\end{document}